\documentclass[a4paper, 11pt]{article}

\usepackage{graphics}
\usepackage{hyperref}
\usepackage{amsthm,amsmath,amssymb,color}
\hypersetup{colorlinks=true,citecolor=blue,linkcolor=blue, urlcolor=blue}

\usepackage{a4wide}

\newtheorem{theorem}{Theorem}

\newtheorem{lemma}[theorem]{Lemma}

\newtheorem{observation}[theorem]{Observation}
\newtheorem{corollary}[theorem]{Corollary}
\newtheorem{remark}[theorem]{Remark}

\newtheorem{definition}[theorem]{Definition}

\makeatletter
\def\prob#1#2#3{\goodbreak\begin{list}{}{\labelwidth\z@ \itemindent-\leftmargin
                        \itemsep\z@  \topsep6\p@\@plus6\p@
                        \let\makelabel\descriptionlabel}
                \item[\it Name]#1
               \item[\it Instance]                #2
                \item[\it Output]#3
                \end{list}}
\makeatother

\def\y{\mathbf{y}}

\def\r{\mathbf{r}}

\def\t{\mathbf{t}}

\def\that{\mathbf{\hat{t}}}
\def\yhat{\mathbf{\widehat{y}}}

\def\ones{\mathbf{1}}

\def\epsilon{\varepsilon}
\def\yuck{2.64}

\newcommand\Tree{\mathbb{T}_{k,\Delta}}

\def\HyperIndSet{\#\mathsf{HyperIndSet}(k,\Delta)}
\def\DomSet{\#\mathsf{DomSet}(\Delta)}
\def\RegDomSet{\#\mathsf{RegDomSet}(\Delta)}

\def\wb{\mathbf{w}}

\def\MM{U}

\let\phi\varphi
\def\coe{c}

\def\zeroes{\mathbf{0}}
\def\cnode{C}
\def\corigin{C^*}

\title{Approximation via Correlation Decay
when Strong Spatial Mixing Fails\thanks{This work was done in part while the authors were visiting the Simons Institute for the Theory of Computing.}}
\author{
  Ivona Bez\'{a}kov\'{a}\thanks{
  Department of Computer Science, Rochester Institute of Technology, Rochester, NY, USA. Research
supported by NSF grant CCF-1319987.}
\and
Andreas Galanis\thanks{
  Department of Computer Science, University of Oxford, Wolfson Building, Parks Road, Oxford, OX1~3QD, UK.
  The research leading to these results has received funding from the European Research Council under
  the European Union's Seventh Framework Programme (FP7/2007-2013) ERC grant agreement no.\ 334828. The paper
  reflects only the authors' views and not the views of the ERC or the European Commission.
  The European Union is not liable for any use that may be made of the information contained therein.}
  \and
  Leslie Ann Goldberg$^\ddag$
  \and
  Heng Guo\thanks{School of Informatics, University of Edinburgh, Informatics Forum, Edinburgh, EH8~9AB, UK.}
\and
 Daniel \v{S}tefankovi\v{c}\thanks{
Department of Computer Science, University of Rochester,
Rochester, NY 14627.  Research
supported by NSF grant CCF-1563757.}  
 }

\date{23 January 2019}

\begin{document}

\maketitle
\begin{abstract}
Approximate counting via correlation decay is the core algorithmic technique
used in the sharp delineation of the computational phase transition that arises in the approximation of
the partition function of anti-ferromagnetic two-spin models.

Previous analyses of correlation-decay algorithms implicitly depended on the occurrence of \emph{strong spatial mixing}.
This, roughly, means that one uses worst-case analysis of the recursive procedure that
creates the sub-instances. In this paper, we develop a new  analysis method that is more
refined than the worst-case analysis. We take the shape of instances in  the computation
tree into consideration and we amortise against certain ``bad''  instances that are created
as the recursion  proceeds. This enables us to show correlation decay and to obtain an FPTAS
even when strong spatial mixing fails.

We apply our technique to the problem of approximately counting
independent sets in hypergraphs with degree upper-bound $\Delta$ and with a lower bound~$k$
on the arity of hyperedges.
Liu and Lin gave an FPTAS for $k\geq 2$ and $\Delta \leq 5$ (lack of strong spatial mixing was
the obstacle preventing this algorithm from being generalised to $\Delta=6$). Our technique
gives a tight result for $\Delta=6$, showing that there is an FPTAS for $k\geq 3$ and $\Delta \leq 6$.
The best previously-known approximation    scheme for $\Delta =6$
is the Markov-chain simulation based FPRAS of Bordewich, Dyer and Karpinski,
which only works for $k\geq 8$.

Our technique also applies for larger values of~$k$, giving an FPTAS
for
$k\geq \Delta$. This bound is not substantially stronger than existing randomised results
in the literature. Nevertheless, it
gives the first deterministic approximation scheme in this regime.
Moreover, unlike existing results, it leads to an FPTAS for counting dominating sets in
regular graphs with sufficiently large degree.

We further demonstrate that in the hypergraph independent
set model, approximating the partition function is NP-hard even
within the uniqueness regime. Also, approximately counting dominating sets of bounded-degree graphs
(without the regularity restriction) is NP-hard.
\end{abstract}

\section{Introduction}
We develop a new method for analysing correlation decays in spin systems.
In particular, we take the shape of instances in  the computation tree into consideration
and we amortise against certain ``bad''  instances that are created as the recursion  proceeds.
This enables us to show correlation decay and to obtain an FPTAS even when strong spatial mixing fails.
To the best of our knowledge, strong spatial mixing is a requirement for all previous
correlation-decay based algorithms.
To illustrate our technique, we  focus on the computational complexity of approximately counting
independent sets in \emph{hypergraphs},
or equivalently on counting the  satisfying assignments of monotone CNF formulas.

The problem of counting independent sets
in \emph{graphs} (denoted $\#\mathsf{IS}$)
is extensively studied.
A beautiful connection has been established, showing that approximately counting independent sets in graphs of maximum degree $\Delta$
undergoes a computational transition which coincides with the uniqueness phase transition
from statistical physics on the infinite $\Delta$-regular tree.
The computational transition can be described as follows.
Weitz \cite{Weitz} designed an FPTAS for counting independent sets on graphs with maximum degree at most $\Delta=5$.
On the other hand, Sly \cite{Sly10} proved that there is no FPRAS for
approximately counting independent sets on graphs with maximum degree
at most $\Delta=6$ (unless $\mathrm{NP}=\mathrm{RP}$).
The same connection has been established in the more general context of approximating
the partition function of the hard-core model \cite{Weitz,MWW,Sly10,GGSVY,GSV:arxiv, SlySun}
and in the even broader context of
approximating the partition functions of generic
antiferromagnetic 2-spin models \cite{SST, GSV:arxiv, SlySun, LLY}
(which includes, for example, the antiferromagnetic Ising model).
As a consequence,
the boundary for the existence of efficient approximation algorithms for these models has been mapped out\footnote{Note that
there are non-monotonic examples of antiferromagnetic 2-spin systems
where the boundary is more complicated because the uniqueness threshold fails to be monotonic in $\Delta$ \cite{LLY2}.
However, this can be cleared up by stating the uniqueness condition as uniqueness for all $d\leq \Delta$. See \cite{LLY2,LLY} for details.}.

Approximate counting via correlation decay is the core technique
in the algorithmic developments which enabled the sharp
delineation of the computational phase transition. Another
standard approach for approximate counting, namely Markov chain
Monte Carlo (MCMC) simulation, is also conjectured to work up to
the uniqueness threshold, but the current analysis tools that we
have do not seem to be powerful enough to show that. For example,
sampling independent sets via MCMC simulation is known to have
fast mixing only for graphs with  degree at most~$4$
\cite{LV99,DG00}, rather than obtaining the true threshold of $5$.

In this work, we consider counting independent sets in hypergraphs with upper-bounded vertex degree, and lower-bounded hyperedge size.
A hypergraph $H=(V,\mathcal{F})$ consists of a vertex set~$V$ and a set $\mathcal{F}$ of hyperedges,
each of which is a subset of $V$.
A hypergraph is said to be \emph{$k$-uniform} if every hyperedge contains exactly $k$ vertices.
Thus, a $2$-uniform hypergraph is the same as a graph.
We will consider the more general case where each hyperedge has arity at least $k$, rather than exactly $k$.

An independent set in a hypergraph $H$ is a subset of vertices that does not contain a hyperedge as a subset.
We will be interested in computing $Z_H$, which is the total number of independent sets in $H$ (also referred to as the partition function of~$H$).
Formally,  the problem of counting independent sets has two parameters --- a degree upper bound~$\Delta$ and a lower bound $k$ on the arity of hyperedges.
The problem is defined as follows.\footnote{Equivalently, one may think of this problem as the problem of counting satisfying assignments of a monotone CNF formulas,
where vertices are variables and hyperedges are clauses. Being out of the independent set (as a vertex) corresponds to being true (as a variable).}

\prob{ $\HyperIndSet$.}
{A hypergraph $H$ with maximum degree at most $\Delta$
where each hyperedge has cardinality (arity) at least~$k$.}
{The number $Z_H$ of independent sets in $H$.}

Previously, $\HyperIndSet$ has been studied using the MCMC technique by Borderwich, Dyer, and Karpinski \cite{BDK08,BDK} (see also \cite{DG00}).
They give an FPRAS for all $k\ge \Delta+2\ge 5$
and for  $k\ge 2$ and $\Delta=3$.
Despite equipping path coupling with optimized metrics obtained using linear programming,
these bounds are not tight for small $k$.
Liu and Lu \cite{LiuLu} showed that there exists an FPTAS for all $k\ge 2$ and $\Delta\le 5$ using the correlation decay technique.

Thus, the situation seems to be similar to the graph case ---   given the analysis tools that we have, correlation-decay brings
us closer to the truth than the best-tuned analysis of MCMC simulation algorithms.
On the other hand, the technique of Liu and Lu \cite{LiuLu} does not extend beyond $\Delta=5$.
To explain the reason why it does not, we need to briefly describe the correlation-decay-based algorithm framework introduced by Weitz \cite{Weitz}.
The main idea is to build a recursive procedure  for computing the marginal probability
that
 any given vertex is in the independent set. The recursion
 works by examining   sub-instances with  ``boundary conditions'' which require
 certain vertices to be in, or out, of the independent set. The recursion
  structure is called a ``computation tree''.
Nodes of the tree correspond to intermediate instances,
and boundary conditions are different in different branches.
The computation tree
allows one to compute the marginal probability exactly but the time needed to do so may
be exponentially large since, in general, the tree is exponentially large.
Typically, an approximate marginal probability is obtained by truncating the computation tree  to logarithmic depth
so that the (approximation) algorithm runs in   polynomial time.
If the correlation between boundary conditions at the leaves of the (truncated) computation tree
and the marginal probability at the root
decays exponentially with respect to the depth,
then the error incurred from the truncation is small and the algorithm succeeds in obtaining a close approximation.

All previous instantiations under this framework require a property
called \emph{strong spatial mixing} (SSM)\footnote{See Section~\ref{sec:SSM} for a definition.},
which roughly states that, conditioned on \emph{any} boundary condition on intermediate nodes, the correlation decays.
In other words, SSM  guards against the worst-case boundary conditions
that might be created by the recursive procedure.

Let the $(\Delta-1)$-ary $k$-uniform hypertree $\Tree$ be the recursively-defined hypergraph
in which each vertex has $\Delta-1$ ``descending'' hyperedges, each containing  $k-1$ new vertices.
\begin{observation}\label{obs:strong}
  Let $k\geq 2$. For $\Delta\geq 6$, strong spatial mixing \emph{does not hold} on $\Tree$.
\end{observation}
Observation~\ref{obs:strong} follows from the fact that the infinite $(\Delta-1)$-ary tree $\mathbb{T}_{2,\Delta}$
can be embedded in the hypertree $\Tree$, and from well-known facts about the phase transition on $\mathbb{T}_{2,\Delta}$.

Observation~\ref{obs:strong} prevents  the generalisation of Liu and Lu's algorithm \cite{LiuLu}
so that it applies for $\Delta\ge 6$,
even with an edge-size lower bound $k$.
The problem is that the construction of the computation tree involves constructing
intermediate instances in which the arity of a hyperedge can be as small as~$2$.
So, even if we start with a $k$-uniform hypergraph, the computation tree will contain
instances with  small hyperedges.
Without strong spatial mixing, these small hyperedges cause problems in the analysis.
Lu, Yang and   Zhang \cite{LYZ}  discuss  this problem and say ``How to avoid this effect is a major open question whose solution may have applications in many other problems.'' This question motivates our work.

 To overcome this difficulty, we introduce a new amortisation technique in the analysis.
Since lack of correlation decay is caused primarily by the presence of small-arity hyperedges within the intermediate
instances, we keep track of such hyperedges.
Thus, we track not only the correlation, but also combinatorial properties of the intermediate instances in the computation tree.
Using this idea, we obtain the following result.
\newcommand{\statethmmain}{There is an FPTAS for $\#\mathsf{HyperIndSet}(3,6)$.}
\begin{theorem}\label{thm:main}
\statethmmain
\end{theorem}
Note that $\#\mathsf{HyperIndSet}(2,6)$ is $\mathrm{NP}$-hard to approximate due to \cite{Sly10}, so our result is tight for $\Delta=6$.
This also shows that $\Delta=6$ is the first case where the complexity of approximately counting independent sets differs on hypergraphs and graphs,
as for $\Delta\le 5$ both admit an FPTAS \cite{LiuLu}.
Moreover, Theorem \ref{thm:main} is stronger than the best MCMC algorithm \cite{BDK} when $\Delta=6$ as \cite{BDK} only works for $k\ge 8$.

We also apply our technique to large $k$. \newcommand{\statethmtwo}{Let $k$ and $\Delta$ be two integers such that $k\ge \Delta$ and $\Delta\ge 200$.
 Then there is an FPTAS for the problem $\HyperIndSet$.}
\begin{theorem}\label{thm:main2}
\statethmtwo
\end{theorem}
In the large $k$ case, our result is
not substantially stronger than
that obtained by analysis of the MCMC algorithm \cite{BDK}
($k\geq \Delta$ rather than $k\geq \Delta + 2$) but it is incomparable since our algorithm is deterministic rather than randomised. Perhaps more importantly, the bound $k\geq \Delta$ allows us to connect the problem of counting independent sets in hypergraphs with the problem of counting dominating sets in $\Delta$-regular graphs and show that the latter admits an FPTAS when  $\Delta$~is sufficiently large. Recall that a dominating set in a graph $G$ is a subset $S$ of the vertices such that every vertex not in $S$ is adjacent to at least one vertex in $S$. We then consider the following problem.

\prob{ $\RegDomSet$.}
{A $\Delta$-regular graph $G$.}
{The number of dominating sets in $G$.}

Our theorems have the following corollary.
\newcommand{\statecorfour}{For all positive integers $\Delta$ satisfying either $\Delta\leq 5$ or $\Delta\geq 199$, there is an FPTAS for the problem $\RegDomSet$.}
\begin{corollary}\label{cor:domsetfptas}
\statecorfour
\end{corollary}
We remark that Corollary~\ref{cor:domsetfptas} cannot be obtained using the result of \cite{BDK}, i.e., the seemingly small difference between $k\geq \Delta$ and $k\geq \Delta+2$ does matter in deriving Corollary~\ref{cor:domsetfptas} from Theorems~\ref{thm:main} and~\ref{thm:main2}. We should also emphasise that it is necessary to consider $\Delta$-regular graphs as inputs to the dominating set problem, since otherwise for graphs of maximum degree $\Delta$ (not necessarily regular), we show that the problem is NP-hard to approximate for $\Delta\geq 18$ (Theorem~\ref{thm:domsethardness}). It is relevant to remark here that we believe that Corollary~\ref{cor:domsetfptas} should hold for all $\Delta$; to do this, it would be sufficient to remove the restriction $\Delta\geq 200$ from Theorem~\ref{thm:main2}. Note that, while we do not know how to remove the restriction, it would at least be possible to improve ``200'' to some smaller number. However, we have chosen to stick with ``200'' in order to keep the proof accessible. We explain next the difficulties in obtaining Theorems~\ref{thm:main} and~\ref{thm:main2}.

The main technical difficulty in correlation-decay analysis is
bounding a function that we call the ``decay rate''. This boils
down to solving an optimisation problem with $(k-1)(\Delta-1)$
variables. In previous work (e.g.\ \cite{SSSY}), this optimisation
has been solved using a so-called ``symmetrization'' argument,
which reduces the problem to a univariate optimisation via
convexity. However, the many variables represent different
branches in the computation tree. Since our analysis takes the
shape of intermediate instances in the tree into consideration,
the symmetrization argument does not work for us, and different
branches  take different values at the maximum. This problem is
compounded by the fact that
  the shape
  of the sub-tree consisting of ``bad'' intermediate instances
  is heavily lopsided, and the
  assignment of variables achieving the maximum is far from uniform.
Given these problems, there does not seem to be a clean solution
to the optimisation in our analysis. Instead of optimizing, we
give an upper bound on the maximum decay rate. In Theorem
\ref{thm:main}, as $k$ and $\Delta$ are small, the number of
variables is manageable, and our bounds are much sharper than
those in Theorem \ref{thm:main2}. On the other hand, because of
this, the proof of Theorem \ref{thm:main2} is much more
accessible, and we will use Theorem~\ref{thm:main2} as a running
example to demonstrate our technique.

We also provide some insight on the hardness side.
Recall that for graphs it is $\mathrm{NP}$-hard to approximate $\#\mathsf{IS}$ beyond the uniqueness threshold ($\Delta=6$)
\cite{Sly10}.\footnote{For graphs the uniqueness and SSM thresholds coincide, but for hypergraphs they differ.}
We prove that it is $\mathrm{NP}$-hard to approximate $\#\mathsf{HyperIndSet}(6,22)$ (Corollary \ref{cor:inapprox1}).
In contrast, we show that uniqueness holds on the $6$-uniform $\Delta$-regular hypertree iff $\Delta\leq 28$ (Corollary \ref{cor:uniqsmall}).
Thus, efficient approximation schemes cease to exist well below the uniqueness threshold on the hypertree.
In fact, we show that this discrepancy grows exponentially in $k$: for large $k$,
it is $\mathrm{NP}$-hard to approximate $\HyperIndSet$ when $\Delta\geq 5\cdot 2^{k/2}$ (Theorem~\ref{thm:inapprox} and Corollary~\ref{cor:inapprox2}),
despite the fact that uniqueness holds on the hypertree for all $\Delta\leq 2^{k}/(2k)$ (Lemma~\ref{lem:technical2criterion}).
Theorem~\ref{thm:inapprox} follows from a rather standard reduction to the hard-core model on graphs.
Nevertheless, it demonstrates that the computational-threshold  phenomena in the hypergraph case ($k>2$) are substantially different from those in the graph case ($k=2$).

As mentioned earlier, there are models where efficient (randomised) approximation schemes  exist
(based on MCMC simulation) even though SSM does not hold.
In fact, this can happen even
 when uniqueness does not hold.
A striking example is the ferromagnetic Ising model (with external field).
As~\cite{SST} shows, there are parameter regimes where uniqueness holds but strong spatial mixing fails.
It is easy to modify the parameters so that even uniqueness fails.
Nevertheless, Jerrum and Sinclair \cite{JS93} gave an MCMC-based FPRAS that applies for all parameters and for general graphs (with no degree bounds).
It is still an open question to give a correlation decay based FPTAS for the ferromagnetic Ising model.

\subsection{Note added in revision --- recent work}

After our work, Hermon, Sly, and Zhang \cite{HSZ} have shown that the Glauber dynamics for sampling independent sets on $n$-vertex $k$-uniform hypergraphs has mixing time $O(n\log n)$, if $\Delta\le c2^{k/2}$ for some constant $c$.
Combined with our hardness result, Corollary~\ref{cor:inapprox2}, this establishes a sharp computational complexity phase transition, up to constants.
Moreover, Moitra \cite{Moitra} has  given a deterministic algorithm for approximately counting and sampling the satisfying assignments of $k$-CNF formulas, provided 
that the variable degree $\Delta\lesssim 2^{k/60}$. Moitra's work relies on connections to the Lov\'asz local lemma.
His work improves our $\Delta\le k$ bound in Theorem \ref{thm:main2} and removes the monotonicity assumption.
Despite this recent progress, Theorem \ref{thm:main} remains the best algorithmic result for $k=3$.

\subsection{Outline of Paper}
In Section~\ref{sec:prelim}, we first give some preliminaries. We give a formal definition of strong spatial mixing (Section~\ref{sec:SSM}) and a reformulation of $\HyperIndSet$ as the problem of
counting satisfying assignments in monotone CNF formulas (Section~\ref{sec:reformulation}).  This will allow us to use the computation tree used by Liu and Lu \cite{LiuLu}. A formal description of the computation tree of \cite{LiuLu} is given in Section~\ref{sec:treecomp}.

In Section~\ref{sec:proofapproach}, we give an overview of our proof approach, i.e., the main idea behind our new amortisation technique.
Section~\ref{sec:theoremproofs} concludes  the proof of Theorems~\ref{thm:main} and~\ref{thm:main2},
using two (not yet proved) technical lemmas  (Lemma~\ref{lem:decay-general} and~\ref{lem:potentialfunctionstar})
which solve a complicated multivariate optimisation problem, and represent the bulk of the technical work of the paper. Section~\ref{sec:domsetfptas} gives the proof of Corollary~\ref{cor:domsetfptas}.

In Section~\ref{sec:largeDelta}, we give the proof of Lemma~\ref{lem:decay-general},
which applies to the large-$\Delta$ setting of Theorem~\ref{thm:main2} and is by far the technically simpler of the
two lemmas.  Section~\ref{sec:decayrate} contains the proof of the technically more
challenging Lemma~\ref{lem:potentialfunctionstar} which applies to the $k=3,\Delta=6$ setting of Theorem~\ref{thm:main}.

Section~\ref{sec:hardnesstotal} gives the formal statements and proofs of the hardness results stated in the Introduction; Section~\ref{sec:hardness} has the hardness results for independent sets in hypergraphs and Section~\ref{sec:domsets} has the hardness results for dominating sets in graphs. Also, Section~\ref{sec:uniqueness} studies the uniqueness threshold on the $k$-uniform $\Delta$-regular hypertree (and gives the proofs of the uniqueness statements made in the Introduction). Finally, Section~\ref{sec:edc} gives the proof for several technical inequalities used in Section~\ref{sec:decayrate}.

Throughout the paper we use computer algebra to prove multivariate
polynomial inequalities over the field of real numbers (the
coefficients of the polynomials are rational). More specifically,
we use the \textsc{Resolve} command in Mathematica. The
underlying quantifier elimination algorithm (described in~\cite{CAD})
provides a rigorous decision procedure that determines feasibility
of a collection of polynomial inequalities.

\section{Preliminaries}\label{sec:prelim}

\subsection{Strong Spatial Mixing}\label{sec:SSM}

For the purposes of this section, it will be convenient to view the independent set model as a 2-spin model. Namely, if $H=(V,\mathcal{F})$ is a hypergraph, each independent set $I$ can be viewed as a $\{0,1\}$-assignment $\sigma$ to the vertices in $V$, where a vertex $v$ is assigned the spin 1 under $\sigma$ if $v\in I$ and $0$ otherwise.

We denote by $\Omega_H$ the set of all independent sets in $H$. The Gibbs distribution $\mu_H(\cdot)$ is the uniform distribution over $\Omega_H$. The Gibbs distribution of $H$ can clearly be viewed as the uniform distribution over those assignments $\sigma:V\rightarrow \{0,1\}$ which encode a valid independent set of $H$. For an assignment $\sigma:V\rightarrow \{0,1\}$ and a subset $\Lambda\subset V$, we denote by $\sigma_\Lambda$ the restriction of $\sigma$ to the subset $\Lambda$.

  For a hypergraph $H=(V,\mathcal{F})$ and a subset $\Lambda\subset V$, we denote by $H_\Lambda$ the subgraph of $H$ induced by $\Lambda$, i.e.,  $H_\Lambda:=(\Lambda,\bigcup_{e\in\mathcal{F}}(e\cap \Lambda))$. Also, for a vertex $v\in V$ and $\Lambda\subset V$, we denote by $\mathrm{dist}(v,\Lambda)$ the length of the shortest path\footnote{A path in a hypergraph with hyperedge set $\mathcal{F}$ is a sequence of edges $e_0,\hdots, e_\ell\in \mathcal{F}$ such that $e_i\cap e_{i+1}\neq \emptyset$ for all $i=0,\hdots,\ell-1$.} between $v$ and a vertex of $\Lambda$.

\begin{definition}\label{def:strongspatialmixing}
Let $\delta:\mathbb{Z}_+\rightarrow[0,1]$. The independent set model exhibits \emph{strong spatial mixing}
 on a hypergraph $H=(V,\mathcal{F})$ with decay rate $\delta(\cdot)$ iff for every $v\in V$, for every $\Lambda \subset V$, for any two configurations
 $\eta,\eta':\Lambda\rightarrow\{0,1\}$ encoding independent sets of $H_\Lambda$, it holds that
\[\Big|\mu_H(\sigma(v)=1\mid \sigma_\Lambda=\eta)-\mu_H(\sigma(v)=1\mid \sigma_\Lambda=\eta')\Big|\leq \delta\big(\mathrm{dist}(v,\Lambda')\big),\]
where $\Lambda'$ denotes the set of vertices in $\Lambda$ such that $\eta$ and $\eta'$ differ.
\end{definition}

\subsection{Reformulation in terms of Monotone CNF formulas}\label{sec:reformulation}

The problem of counting the independent sets of a hypergraph has  an equivalent formulation
in terms of monotone CNF formulas.
In order to describe the equivalent formulation, we first describe the
problem of counting satisfying assignments of monotone CNF formulas.

A monotone CNF formula $C$ consists of a set of variables $V$ and
a set of clauses $\{c_1,c_2,\ldots\}$. Each clause $c_i$ is
associated with some subset $S_i$ of $V$ and is the disjunction of
all variables in~$S_i$. The \emph{arity} of a clause $c_i$,
denoted $|c_i|$, is defined to be~$|S_i|$. For a variable $x\in
V$, its \emph{degree} $d_x(C)$ is the number of clauses where $x$
appears. The \emph{maximum degree} of~$C$ is given by $\max_{x\in
V} d_x(C)$.

\begin{definition}
Let $\mathcal{C}_{k,\Delta}$ be the set of all monotone CNF formulas which have
maximum degree at most~$\Delta$ and whose clauses have arity at least~$k$.
\end{definition}

Note that a formula in $\mathcal{C}_{k,\Delta}$ may have some clauses with arbitrarily large arities.
A \emph{satisfying assignment} of the formula is an assignment of truth values to the variables
which makes the formula evaluate to ``true''.

Suppose that $H=(V,\mathcal{F})$ is a hypergraph
with maximum degree at most~$\Delta$ where each hyperedge has arity at least~$k$.
Let $C$ be the corresponding formula in $\mathcal{C}_{k,\Delta}$
with variable set~$V$.
The correspondence is that each hyperedge $S_i$ of $H$
is associated with exactly one clause $c_i$ of~$C$.
Independent sets of $H$   are in one-to-one correspondence
with satisfying assignments  of~$C$ ---
a variable is assigned value ``true''
in an assignment
if and only if it is out of the corresponding independent set.

Going the other direction, any monotone CNF formula can be viewed
as a hypergraph.
In the technical sections of this paper, we use the
monotone CNF formulation.

In this article, when we consider a monotone CNF formula $C$
we will typically use $n$ to denote $|V|$.
Variables in $V$  will be denoted by $x_1,x_2,\hdots$
    When $x$ and $C$ are clear from context, we will  sometimes use $d$ to denote $d_x(C)$.
When $C$ is clear from context,
we will sometimes use $\Delta$ to denote $\max_{x\in V} d_x(C)$.

\subsection{The Computation Tree}\label{sec:treecomp}

In this section, we set up relevant notation and give an exposition of the computation tree of Liu and Lu \cite{LiuLu} which will also be used in our proof (though our analysis will be different).
The computation tree of Liu and Lu is given in terms of
the monotone CNF version of the problem.
Below we give the relevant definitions and notation; our notation aligns as much as possible with that of \cite{LiuLu}.

 Our goal is to approximately count the number of satisfying assignments of a formula $C\in \mathcal{C}_{k,\Delta}$, which we denote by $Z(C)$. Since $C$ is monotone, an assignment $\sigma:V\rightarrow \{0,1\}$ is satisfying if, for every clause in $C$, there is at least one variable $x\in c$ with $\sigma(x)=1$. Note that $Z(C)>0$ since the all-1 assignment satisfies every monotone CNF formula. For convenience, we will use the simplified notation ``$x=1$'' to denote (the set of) satisfying assignments of $C$  in which $x$ is set to 1,
and  we similarly use  ``$x=0$''. We associate the formula~$C$ with
a probability distribution    in which each satisfying assignment  has probability mass $1/Z(C)$.
We will denote probabilities with respect to this distribution by $\Pr_C(\cdot)$.

Let $x$ be a variable in $V$. Define $R(C,x):=\frac{\Pr_C(x=0)}{\Pr_C(x=1)}$, this is well-defined since $\Pr_C(x=1)>0$ by the monotonicity of $C$.  In fact, the  monotonicity of $C$ also implies that $0\leq R(C,x)\leq 1$, where the upper bound follows from the  fact that, for every satisfying assignment with $x=0$, flipping the assignment of $x$ to 1 does not affect  satisfiability.
Our interest in the quantity $R(C,x)$ stems from the following simple lemma from~\cite{LiuLu}.
 \newcommand{\statelemRs}{
Let $k$ and $\Delta$ be positive integers.
Suppose that
there is a polynomial-time algorithm (in $n$ and $1/\epsilon$)
that takes an $n$-variable formula $C\in \mathcal{C}_{k,\Delta}$,
a variable $x$ of~$C$, and an $\epsilon>0$
and computes a quantity
$ \widehat{R}(C,x)$ satisfying
$|\widehat{R}(C,x)-R(C,x)|\leq \epsilon$.
Then, there  is an FPTAS
which approximates $Z(C)$ for every $C\in \mathcal{C}_{k,\Delta}$.
}
\begin{lemma}[{\cite{LiuLu}}]\label{lem:basic}
\statelemRs
\end{lemma}

\begin{proof}
The proof is actually identical to the argument in \cite[Appendix A]{LiuLu}.
We include the proof for completeness, and also because
an examination of the proof is necessary to check that
the FPTAS for approximating $Z(C)$
invokes the algorithm that computes $\widehat{R}(C,x)$
only on formulas $C$ whose clauses have arity at least $k$ (and maximum degree $\Delta$).

Let $\epsilon>0$ and $C$ be a monotone CNF formula $C$ with maximum degree $\Delta$ whose clauses have arity at least $k$. Let $x_1,\hdots,x_n$ be the variables in $C$. Let $C_i$ be the formula obtained from $C$ by setting $x_1=\cdots=x_{i}=1$ and removing all the clauses that are satisfied (i.e., all clauses that contain a variable from $x_1,\hdots,x_i$).  We have
\begin{equation}\label{eq:inverseZs}
\begin{aligned}
\frac{1}{Z(C)}&=\mbox{$\Pr_C$}(x_1=\hdots=x_n=1)=\prod^{n-1}_{i=0}\mbox{$\Pr_C$}(x_{i+1}=1\mid x_1=\cdots=x_{i}=1)\\
&=\prod^{n-1}_{i=0}\mbox{$\Pr_{C_i}$}(x_{i+1}=1)=\prod^{n-1}_{i=0}\frac{1}{1+R(C_i,x_{i+1})}.
\end{aligned}
\end{equation}
Note that every  $C_i$ is a monotone CNF formula  with maximum degree $\Delta$ whose clauses have arity at least $k$. By the assumption in the lemma, we can compute (in $\text{poly}(n,1/\epsilon)$ time) quantities $\widehat{R}(C_i,x_{i+1})$ such that
\begin{equation}\label{eq:approxRs}
\big|\widehat{R}(C_i,x_{i+1})-R(C_i,x_{i+1})\big|\leq \epsilon/(100n) \mbox{ for all } i=0,\hdots,n-1.
\end{equation}
Let
\begin{equation}\label{eq:approxZs}
\widehat{Z}(C)=\prod^{n-1}_{i=0}\big(1+\widehat{R}(C_i,x_{i+1})\big).
\end{equation}
It is not hard to conclude from \eqref{eq:inverseZs}, \eqref{eq:approxRs} and \eqref{eq:approxZs} that $(1-\epsilon)Z(C)\leq \widehat{Z}(C)\leq (1+\epsilon)Z(C)$. This completes the proof.
\end{proof}

Liu and Lu \cite{LiuLu} established that a computation tree approach gives a recursive procedure
for \emph{exactly} calculating  $R(C,x)$  for any monotone CNF formula $C$ and any variable $x\in C$.  We next give the details of this recursive procedure (see \cite[Lemma 5]{LiuLu}). First, we introduce the following definitions.
\begin{definition}
Let $C$ be a monotone CNF formula and let $x$ be a variable in $C$. We call the variable $x$ \emph{forced} (in $C$) if $x$ appears in a clause of arity 1 in $C$ (note that in every satisfying assignment of $C$ it must be the case that $x=1$ and hence $R(C,x)=0$). We call the variable $x$ \emph{free} if $x$ does not appear in any clause of $C$ (note that $R(C,x)=1$ in this case).
\end{definition}

\begin{definition}
Let $C$ be a monotone CNF formula and let $c$ be a clause in $C$. We call the clause $c$ \emph{redundant} (in $C$) if there is a clause $c'$ in $C$ such that $c$ is a (strict) superset of $c'$ (note that removing $c$ from $C$ does not affect the set of satisfying assignments of $C$).
\end{definition}
We next give the details of the computation tree. The nodes in the computation tree will be pairs $(C,x)$ such that
\begin{equation}\label{eq:invariant}
\mbox{ $C$ is a monotone CNF formula and $x$ is a variable which is not forced in $C$.}
\end{equation}
Let $(C,x)$ satisfy \eqref{eq:invariant}. We first perform a
pre-processing step on $C$ which involves (i) initially removing
all of the redundant clauses, (ii) then, removing all clauses of
arity 1. Note that part (ii) of the preprocessing step removes all
forced variables that were present in $C$; at the time of the
removal, forced variables appear only in clauses of arity 1 since
part (i) of the preprocessing step has already removed all
redundant clauses in $C$ (and hence all clauses of arity greater
than 1 that contain forced variables).  Denote the formula after
the completion of the preprocessing step by $\widetilde{C}$. Note
that every clause in $\widetilde{C}$ is also a clause in the
initial formula $C$. It follows that $x$ is not forced in
$\widetilde{C}$. Further, since removing redundant clauses does
not change the set of satisfying assignments of $C$ and $x$ is not
forced in $C$, we have that $R(\widetilde{C},x)=R(C,x)$.

If $x$ is free in $\widetilde{C}$ (the formula after the pre-processing step), then the start node $(C,x)$ is (declared) a leaf of the computation tree (note that in this case $R(C,x)=1$). In the sequel, we assume that $x$ is not free in $\widetilde{C}$. Denote by $\{c_i\}_{i\in[d]}$ the clauses where $x$ occurs in $\widetilde{C}$ and let $w_{i}=|c_i|-1$ (note that $d\geq 1$). We will use $\wb$ to denote the vector $(w_1,\hdots,w_d)$. The variables in clause $c_i$ other than $x$ will be denoted by $x_{i,1},\hdots,x_{i,w_i}$. For the pair $(C,x)$, we next   construct pairs $(C_{i,j},x_{i,j})$ for $i\in [d]$ and $j\in [w_i]$, where $C_{i,j}$ is an appropriate  subformula obtained from $\widetilde{C}$, roughly, by hard-coding (some of) the occurrences of the  variables in $\widetilde{C}$ to either 1 or 0 (this will be explained below and will
henceforth be referred to as pinning)\footnote{Note that our notation for $i,j$ is different from the one in \cite{LiuLu}; there, the roles of $i,j$ are interchanged.}.

Precisely, for $i\in[d]$, let $C_i$ be the formula obtained from
$\widetilde{C}$ by removing clauses $c_1,\hdots,c_{i-1}$ (note
that this has the same effect as pinning the occurrences of $x$ in
these clauses to 1) and pinning the occurrences of $x$ in
$c_{i+1},\hdots,c_{d}$ to 0 (this corresponds to removing $x$ from
these clauses, and thus reducing their arities). For $j\in [w_i]$,
the formula $C_{i,j}$ is obtained from $C_i$ by further removing
clause $c_i$ and pinning all the occurrences of
$x_{i,1},\hdots,x_{i,j-1}$  to 0.

Before proceeding, let us argue that the pairs $(C_{i,j},x_{i,j})$ satisfy \eqref{eq:invariant} for all $i\in [d]$ and $j\in [w_i]$. For such $i,j$, we first prove that $C_{i,j}$ is a (satisfiable)  monotone CNF formula. That is, we prove that the various pinnings in the construction of $C_{i,j}$ from $\widetilde{C}$ do not pin all variables of some clause of $\widetilde{C}$ to 0. For the sake of contradiction, assume otherwise. Observe that $C_{i,j}$ is obtained from $\widetilde{C}$ by either removing some clauses or by pinning some occurrences of the variables to $0$. Clearly, removal of clauses does not affect satisfiability, so we may focus on the effect of pinning. For $i\in [d]$ and $j\in [w_i]$, the only variables whose (some of the) occurrences in $\widetilde{C}$  get pinned to 0 are $x,x_{i,1},\hdots,x_{i,j-1}$. Since we assumed (for contradiction) that $C_{i,j}$ is   unsatisfiable, it must be the case that there exists a clause $c'$ in $\widetilde{C}$ all of whose variables are (a subset of) $x,x_{i,1},\hdots,x_{i,j-1}$. It follows that $c_i$ is redundant in $\widetilde{C}$ since it is a strict superset of clause $c'$. This gives a contradiction, since the pre-processing operation ensures that $\widetilde{C}$ has no redundant clauses. Thus, $C_{i,j}$ is satisfiable as wanted. Next, we show that $x_{i,j}$ is not forced in $C_{i,j}$. First, observe that $x_{i,j}$ is not forced in $\widetilde{C}$ since the second part of the preprocessing step ensures that $\widetilde{C}$ does not contain forced variables. Thus, the only way that $x_{i,j}$ can be forced in $C_{i,j}$ is if there existed a clause $c'$ in $\widetilde{C}$ whose variables were $x_{i,j}$ together with a subset of $x,x_{i,1},\hdots,x_{i,j-1}$. Since $\widetilde{C}$ includes $c_i$ and $\widetilde{C}$ does not have redundant clauses, it must be the case that $c'=c_i$. It remains to observe that $C_{i,j}$ does not include (any subclause of) $c_i$, from which it follows that $x_{i,j}$ is not forced in $C_{i,j}$.

We are now ready to state the relation between $R(C,x)$ and the quantities $R(C_{i,j},x_{i,j})$ with $i\in[d]$ and $j\in [w_i]$.

\begin{lemma}[{\cite[Lemma 5]{LiuLu}}]\label{lem:LiuLurecursion}
It holds that
\begin{equation}\label{eq:Frecursion}
R(C,x)=\prod^{d}_{i=1}\bigg(1-\prod^{w_i}_{j=1}\frac{R(C_{i,j},x_{i,j})}{1+R(C_{i,j},x_{i,j})}\bigg).
\end{equation}
\end{lemma}
\begin{proof}
The proof is identical to the proof of \cite[Lemma 5]{LiuLu} (which in turn builds on the technique of \cite{Weitz}), we give the proof for completeness.

Recall that $\widetilde{C}$ is the formula after the preprocessing step and that $R(\widetilde{C},x)=R(C,x)$.
We may assume that $x$ is not free in $\widetilde{C}$ (otherwise, it holds that $R(\widetilde{C},x)=1$, which coincides with the evaluation of the right hand side of \eqref{eq:Frecursion} under the standard convention that the empty product evaluates to 1).

Equation \eqref{eq:Frecursion} follows immediately from the following two equalities.
\begin{equation}\label{eq:mhnol1}
R(\widetilde{C},x)=\prod_{i\in[d]}R(C_i,x), \quad
R(C_i,x)=1-\prod^{w_i}_{j=1}\frac{R(C_{i,j},x_{i,j})}{1+R(C_{i,j},x_{i,j})}.
\end{equation}
The first equality in \eqref{eq:mhnol1} is a consequence of a telescoping expansion of $\frac{\Pr_{\widetilde{C}}(x=0)}{\Pr_{\widetilde{C}}(x=1)}$. To see this, let $\widetilde{C}'$ be the formula obtained from $\widetilde{C}$ by replacing, for all $i\in [d]$, the occurrence of the variable $x$ in clause $c_i$ by a new variable $x'_i$. We have that
\begin{align*}
R(\widetilde{C},x)&=\frac{\Pr_{\widetilde{C}}(x=0)}{\Pr_{\widetilde{C}}(x=1)}=\frac{\Pr_{\widetilde{C}'}(x'_1=0,\hdots, x'_d=0)}{\Pr_{\widetilde{C}'}(x'_1=1,\hdots, x'_d=1)}\\
&=\prod_{i\in[d]}\frac{\Pr_{\widetilde{C}'}(x'_1=1,\hdots,x'_{i-1}=1,x'_{i}=0,\hdots x'_d=0)}{\Pr_{\widetilde{C}'}(x'_1=1,\hdots,x'_i=1,x'_{i+1}=0,\hdots x'_d=0)}=\prod_{i\in[d]}\frac{\Pr_{C_i}(x=0)}{\Pr_{C_i}(x=1)},
\end{align*}
which yields the first equality in \eqref{eq:mhnol1} after substituting  $R(C_i,x)=\frac{\Pr_{C_i}(x=0)}{\Pr_{C_i}(x=1)}$.

For the second equality in \eqref{eq:mhnol1}, observe that $x$ appears only in clause $c_i$ of the formula $C_i$, and thus (denoting by $C_i\backslash c_i$ the formula which is obtained from $C_i$ by deleting clause $c_i$)
\[R(C_i,x)=\frac{\Pr_{C_i}(x=0)}{\Pr_{C_i}(x=1)}=1-\mbox{$\Pr_{C_i\backslash c_i}$}(x_{i,1}=0,\hdots,x_{i,w_i}=0)=1-\prod^{w_i}_{j=1}\mbox{$\Pr_{C_{i,j}}(x_{i,j}=0)$},\]
which proves the desired equality after substituting  $\Pr_{C_{i,j}}(x_{i,j}=0)=\frac{R(C_{i,j},x_{i,j})}{1+R(C_{i,j},x_{i,j})}$.
\end{proof}

By applying \eqref{eq:Frecursion} recursively, it is not hard to see that one can compute the quantity $R(C,x)$ \emph{exactly}. Of course, exact computation using this scheme will typically require exponential time, so as in \cite{LiuLu} we will stop the recursion at some (small) depth $L$ to keep the computations feasible within polynomial time. This will yield a quantity $R(C,x,L)$ and the hope is that, by choosing $L$ appropriately, the error $|R(C,x,L)-R(C,x)|$ will be sufficiently small.

In light of \eqref{eq:Frecursion}, a natural way to define $R(C,x,L)$ for integer $L\geq 0$ is as follows\footnote{\label{foot:start}Note that the value 1 of $R(C,x,L)$ when $L\leq 0$ is somewhat arbitrary since $L\leq 0$ corresponds to stopping the recursion. Our choice of the value 1 will be convenient for technical reasons that will become apparent in the proof of the upcoming Lemma~\ref{lem:method}.}.
\begin{equation*}
R(C,x,L)=\begin{cases} 1,& \mbox{if $x$ is free in $\widetilde{C}$ or $L=0$},\\
\prod^{d}_{i=1}\big(1-\prod^{w_i}_{j=1}\frac{R(C_{i,j},x_{i,j},L-1)}{1+R(C_{i,j},x_{i,j},L-1)}\big),& \mbox{otherwise}.
\end{cases}
\end{equation*}
It is immediate that when the formula $C$ has maximum degree bounded by a constant and, further, every clause has arity also bounded above by a constant, one can compute $R(C,x,L)$ in time polynomial in $n$ whenever $L=O(\log n)$. To account for formulas where the arities of the clauses can be arbitrarily large (but still where the degrees of variables are bounded by a constant), one needs  to be more careful with clauses of large arity (i.e., when their arity as a function of $n$ is $\omega(1)$, say $\log n$). As in \cite{LiuLu}, we will account for this more general setting by pruning the recursion earlier whenever we encounter clauses with large arity.

\begin{definition} \label{eq:lwi}
For any integer~$w$, let $l_{w}:=\left\lceil \log_6 (w+1)\right\rceil$.
\end{definition}
Note that $l_1=\hdots=l_5=1$.
For integer $L$, we set
\begin{equation}\label{eq:algo}
R(C,x,L)=\begin{cases} 1,& \mbox{if $x$ is free in $\widetilde{C}$ or $L\leq 0$},\\
\prod^{d}_{i=1}\big(1-\prod^{w_i}_{j=1}\frac{R(C_{i,j},x_{i,j},L-l_{w_i})}{1+R(C_{i,j},x_{i,j},L-l_{w_i})}\big),& \mbox{otherwise}.
\end{cases}
\end{equation}

The particular choice of the logarithm base in
Definition~\ref{eq:lwi}  is not very important as long as it is a big enough constant.
The quantity $R(C,x,L)$ is typically called a ``message'' (because it gets passed up the computation tree).

\begin{remark}\label{rem:boundrs}
For formulas $C$ with a variable $x$ which is not forced in $C$, we have the lower bound $R(C,x)\geq (1/2)^{d_x(C)}$. The bound is simple to see using \eqref{eq:Frecursion} and the fact that $R(C_{i,j},x_{i,j})\leq 1$ for all $i\in [d_x(C)]$ and $j\in [w_i]$. Similarly, for all integers $L$ and all nodes $(C,x)$ in the computation tree we have the bound $R(C,x,L)\geq (1/2)^{d_x(C)}$.
\end{remark}

\section{Proof Outline}\label{sec:proofapproach}

We want to guarantee that the error $|R(C,x,L)-R(C,x)|$ is exponentially small in $L$.
Notice that if we run the recursion long enough, it computes the true value; namely, $R(C,x,\infty)=R(C,x)$.
More precisely, we will prove the following two lemmas,
which correspond to the settings of Theorem~\ref{thm:main} and Theorem~\ref{thm:main2}, respectively.
Recall that $\mathcal{C}_{k,\Delta}$ is the set of all monotone CNF formulas which have
maximum degree
at most~$\Delta$ and whose clauses have arity at least~$k$.
Our proof will use the following constant.
\begin{definition}\label{def:alpha}
Let $\alpha=1-10^{-4}$.
\end{definition}

\newcommand{\statelemDelta}{
There exists a constant $\tau>0$ such that for every $C \in \mathcal{C}_{3,6}$,
every variable $x\in C$, and
every integer~$L$,
\[|R(C,x,L)-R(C,x,\infty)|\leq \tau \alpha^L.\]
}
\begin{lemma}\label{lem:Delta6}
\statelemDelta
\end{lemma}

\newcommand{\statelemgenDelta}{Let $k$ and $\Delta$ be two integers such that $k\ge \Delta$ and $\Delta\ge 200$.
There exists a constant $\tau>0$ such that for every $C \in \mathcal{C}_{k,\Delta}$,
every variable $x$ in $C$,  and every integer~$L$,
$$
  |R(C,x,L)-R(C,x,\infty)|\leq \tau \alpha^L.
$$
}
\begin{lemma}\label{lem:delta-general}
\statelemgenDelta
\end{lemma}

Our proof uses correlation decay techniques together with a new method which accounts for the shape of the computation tree.
Lemma \ref{lem:delta-general} is technically simpler and is proved in Section~\ref{sec:largeDelta} to better illustrate the idea.
Lemma \ref{lem:Delta6} is proved in Section~\ref{sec:decayrate}.
In the rest of this section, we give an overview of our overall proof strategy.

To analyze the error of the recursion, the standard approach so far in the literature has been to show that,
for a node $(C,x)$ in the computation tree, the quantity $|R(C,x,L)-R(C,x,\infty)|$
is bounded by $\alpha\max_{i,j}|R(C_{i,j},x_{i,j},L-1)-R(C_{i,j},x_{i,j},\infty)|$ for some constant $0<\alpha<1$.
This allows one to inductively deduce that $|R(C,x,L)-R(C,x,\infty)|$ decays exponentially in $L$.
This approach has been extremely successful when strong spatial mixing holds~\cite{SST,LLY,LiuLu,SSSY,YZ,LYZ}.

In fact, this step-wise decay seldom holds if we track $R(C,x,L)$ directly.
Instead, the analysis is usually done by tracking
$\Phi(R(C,x,L))$ for an appropriate potential function~$\Phi$.
In particular, let $\Phi:(0,1]\rightarrow \mathbb{R}$ be a potential function that satisfies:
\begin{align}
  \mbox{$\Phi$ is continuously differentiable on $(0,1]$, and $\phi:=\Phi'$ satisfies $\phi(z)> 0$ for $z\in(0,1]$.}
  \label{eq:conditionsPhi}
\end{align}
The usual approach is
to show that $|\Phi(R(C,x,L))-\Phi(R(C,x,\infty))|$ decays exponentially in $L$,
which is sufficient to imply lemmas like Lemma~\ref{lem:Delta6} and Lemma~\ref{lem:delta-general}.

In our setting, this inductive approach is problematic since, inside the computation tree,
we are faced with the possibility that the formula at the root of a subtree could have many arity-$2$ clauses.
For $\Delta\geq 6$, these subtrees prohibit the application of the above proof scheme
since they are in non-uniqueness and hence the desired step-by-step decay is no longer present,
regardless of the choice of the potential function.

While arity-$2$ clauses are problematic, clauses with larger arities do at least lead to good decay of correlation in a single step.
In general, as the arity gets larger, the decay gets better.
Thus, our approach is to do an amortised analysis.
In a single step, we track both the one-step decay of correlation and  the number of variables in the current formula
that are pinned to $0$.
These $0$~pinnings will decrease the effective arity of clauses, and will later lead to worse decay.

More formally, we will track a specific quantity $m(C,x,L)$ which is assigned to each node in the computation tree.
Let $\corigin$ be the original monotone CNF formula and let $(\cnode,x)$ be a node in the computation tree.
As explained in Section~\ref{sec:treecomp}, each clause $c'$ of $\cnode$ is obtained from a clause $c$ of $\corigin$ by pinning a certain number of variables to $0$ (possibly none),
which effectively is the same as removing those variables from the clause.
We call these $0$-pinnings \emph{deficits} and let $\max\{0,k-|c'|\}$ be the number of deficits of $c'$.
Note that a clause of arity larger than $k$ is considered to have no deficit,
although some variables of it may have been pinned to $0$.
\begin{definition}
Let $D(\cnode)=\sum_{c'\in \cnode}\max\{0,k-|c'|\}$ denote the  sum of the deficits of the clauses in $\cnode$.
\end{definition}
Observe that if a clause $c$ of $\corigin$ does not show up in $\cnode$, it does not contribute any deficits.
For any node $(\cnode,x)$ in the computation tree, let
\[m(\cnode,x,L):=\delta^{D(\cnode)}\Phi(R(\cnode,x,L))\]
where $\delta\in (0,1)$ is a constant that we will choose later,
and the potential function $\Phi$ will be specified shortly in
Definition~\ref{phi:definition}. Crucially, the root formula
$\corigin$ satisfies
\[|m(\corigin,x,L)-m(\corigin,x,\infty)|=|\Phi(R(\corigin,x,L))-\Phi(R(\corigin,x,\infty))|,\]
since at the root no variable is pinned yet and $D(\corigin)=0$.
Thus, the key step is to show that the quantity $|m(C,x,L)-m(C,x,\infty)|$ decays exponentially with $L$;
we will show that, for an arbitrary node $(\cnode,x)$ in the computation tree, it holds that
\begin{equation}\label{eq:stepbystep}
  |m(\cnode,x,L)-m(\cnode,x,\infty)|\leq \alpha\, \mbox{$\max_{i,j}$}\,|m(\cnode_{i,j},x_{i,j},L-1)-m(\cnode_{i,j},x_{i,j},\infty)|,
\end{equation}
where
$\alpha= 1-10^{-4}$ is from Definition~\ref{def:alpha}.

In previous applications of the correlation decay technique, the
ordering of the children of each node $(\cnode,x)$ is usually
arbitrary. Since we want to take the shape of the computation tree
into consideration, this ordering becomes important to us. We will
order clauses in the order of increasing size, except that we
leave arity $2$ clauses to the end.

Unfortunately, the quantity~$m(C,x,L)$ is more complicated than the plain message~$R(C,x,L)$
and it is even more complicated than~$\Phi(R(C,x,L))$,
since it incorporates combinatorial information about the formula $C$ and thus it does not satisfy a simple recursion (unlike  $R(C,x,L)$).
Nevertheless, we are able to define a multi-variable quantity $\kappa_*$
(see \eqref{eq:kappa-general}) and
we will show  (see Lemma~\ref{lem:method})
that when   $\kappa_*\leq 1$, inequality \eqref{eq:stepbystep} holds.

We will use the following potential function, which satisfies \eqref{eq:conditionsPhi}
as required. In general, the choice of an appropriate potential function is guided by an ``educated guess''.
\begin{definition}\label{phi:definition}
Let $\chi=1/2$ and $\psi=13/10$.
Define
$$   \Phi(z):=\frac{1}{\chi \psi}\log\left(\frac{z^{\chi}}{\psi-z^{\chi}}\right).$$
Let $\phi(z)$ denote $\Phi'(z)$ so that
$$
  \phi(z):=\Phi'(z)=\frac{1}{z (\psi - z^\chi)}.$$
  \end{definition}
  \begin{remark}\label{remark:invert}
The exact values of $\chi$ and $\psi$ do not matter at this stage, but it is important   that $0<\chi \leq 1$ and $\psi>1$. For such values of $\chi$ and $\psi$,  $\Phi^{-1}$ exists  and is uniquely defined over the range of $\Phi(z)$  for $z\in (0,1]$.
  \end{remark}

\subsection{A general framework to bound the error}\label{sec:genframework}

First let us calculate how the number of deficits changes in one step of the recursion. To avoid trivialities, we assume $k\geq 2$.
Let $(C,x)$ be a node in the computation tree.
As in Section~\ref{sec:treecomp},
we first perform a pre-processing step on~$C$ which removes redundant clauses and
removes clauses of arity~$1$, producing a new formula~$\widetilde{C}$.
Every clause in $\widetilde{C}$ is a clause of~$C$,
so $D(\widetilde{C}) \leq D(C)$.
Also, $x$ is neither forced nor free in~$\widetilde{C}$ (otherwise, $(C,x)$ is a leaf in the computation tree).
Let $d=d_x(\widetilde{C})$ and
let $c_1,\ldots,c_d$ be the clauses where $x$ occurs in~$\widetilde{C}$.
Recall that $w_i = |c_i|-1$.
As we mentioned in Section~\ref{sec:proofapproach}, the order of $c_1,\ldots,c_d$ is important.
We will order clauses in order of increasing size, except that we leave arity-$2$ clauses to the end.
Here is some notation to describe the ordering.
Let $b_\ell$ denote the number of clauses amongst $c_1,\ldots,c_d$ with arity~$\ell$.
We will use the variables $b'_\ell$ to denote cumulative sums for $\ell>2$, so
$b'_2=0$ and,
for $\ell\geq 3$,
$b'_\ell = b'_{\ell-1} + b_\ell$.
We order the clauses so that, for $\ell \geq 3$,
clauses $c_{b'_{\ell-1}+1},\ldots,c_{b'_{\ell}}$ have arity~$\ell$.
Finally, clauses $c_{d-b_2+1},\ldots,c_d$ have arity~$2$.
Let $s_i$ be the sum of the deficits of clauses $c_1,\ldots,c_i$.
Thus,
$$s_i = \sum_{t=1}^i \max(0,k-w_t-1).$$

We will now consider how the deficits change when we construct the node $(C_{i,j},x_{i,j})$ from
$(\widetilde{C},x)$ according to the method described in Section~\ref{sec:treecomp}, where $j\in [w_i]$.
\begin{itemize}
\item The arity-$2$ clauses in $\widetilde{C}$ are always removed in the construction of $C_{i,j}$, resulting
in a loss of deficit of $b_2 (k-2)$.
\item The clauses $c_1,\ldots, c_{i}$ are also removed
in the construction of $C_{i,j}$
resulting in an additional loss of
deficit of $s_{\min(i,d-b_2)}$. (The minimum is to avoid double-counting if $i>d-b_2$ since in that cases
some of these clauses have arity $2$, and have already been counted.)
\item  The occurrences of~$x$ are pinned to~$0$ in clauses~$c_{i+1},\ldots,c_d$.
Consider some $t\in \{i+1,\ldots,d\}$.
If the arity of $c_t$ is greater than~$k$ then this pinning does not cause any increase in deficit.
Also, if the arity of $c_t$ is~$2$, then the clause will be removed, so there is no increase in deficit.
Thus, the increase in deficit from these pinnings is at most
$\max(0,b'_k-i)$.
\item If $i \leq d-b_2$ then $w_i>1$
and all occurrences of $x_{i,1},\ldots,x_{i,j-1}$ are pinned to~$0$, resulting
in an increase in deficit of at most $(j-1)(\Delta-1)$.
\end{itemize}
Let $\mathbf{1}_{i\leq d-b_2}$ be zero-one indicator variable for
the event that $i$ is at most $d-b_2$. Then, putting these observations
together,
we conclude that
$$
D(C_{i,j}) \leq D(\widetilde{C}) - b_2 (k-2) - s_{\min(i,d-b_2)} + \max(0,b'_k-i) + (j-1)(\Delta-1)\mathbf{1}_{i\leq d-b_2}.
$$
Since $D(\widetilde{C}) \leq D(C)$, we have
\begin{equation}
\label{eq:deficits}
D(C_{i,j}) \leq D(C) - b_2 (k-2) - s_{\min(i,d-b_2)} + \max(0,b'_k-i) + (j-1)(\Delta-1)\mathbf{1}_{i\leq d-b_2}.
\end{equation}

 Recall that the recursion for $R(C,x,L)$ depends on the function $F^{d,\wb}(\r)$ implicitly defined by \eqref{eq:algo}, i.e.,
\begin{equation}\label{eq:recursion}
  F^{d,\wb}(\r):=\prod^{d}_{i=1}\bigg(1-\prod^{w_i}_{j=1}\frac{r_{i,j}}{1+r_{i,j}}\bigg),
\end{equation}
where $r_{i,j}\in [0,1]$ for all $i\in [d]$ and  $j\in[w_i]$.
In particular, unless $x$ is free in $\widetilde{C}$ or $L\leq 0$, we have
 $$R(C,x,L)=F^{d,\wb}(\{R(C_{i,j},x_{i,j},L-l_{w_i})\}).$$

The $m(C,x,L)$ variables also satisfy a recursion which could be made explicit
by mapping them back to the $R(C,x,L)$ variables, though we will not
directly analyse this  recursion.
Instead, we will define a quantity $\kappa_*^{d,\wb}(\r)$
which tracks the  rate at which $|m(C,x,L)-m(C,x,\infty)|$ decays in the recursion.
Specifically, define $\kappa_*^{d,\wb}(\r)$ as follows.
\begin{equation}
\label{eq:kappa-general}
\kappa_*^{d,\wb}(\r) :=
\sum_{i=1}^{d}
\sum_{j=1}^{w_i}  \alpha^{-l_{w_i}}
\delta^{(b_2 (k-2) + s_{\min(i,d-b_2)} - \max(0,b'_k-i) - (j-1)(\Delta-1)\mathbf{1}_{i\leq d-b_2})}
\frac{\phi(F^{d,\wb}(\r))}{\phi(r_{i,j})}
    \left|\frac{\partial F^{d,\wb}(\r)}{\partial r_{i,j}}\right|  .
  \end{equation}

The main step in the proofs of  Lemma \ref{lem:Delta6} and Lemma \ref{lem:delta-general} will be
  to bound $\kappa_*^{d,\wb}(\r)$. By construction, the elements in $\wb$ are in increasing order, apart from the
  $1$s at the end, and the bound on $\kappa_*^{d,\wb}(\r)$ will use this fact, so we give the following definition.
 \begin{definition}\label{def:suitable}
 Let $w_0=2$.
 A vector $\wb = w_1,\ldots,w_d$ is \emph{suitable} if  its entries are positive integers and
 there is a $t \in \{0,\ldots,d\}$ such
 for all $j$ in $\{1,\ldots,t\}$, $w_j \geq  w_{j-1}$
 and for all $j$ in $\{t+1,\ldots,d\}$, $w_j=1$.
 Given a suitable vector $\wb$ we use the following global notation (which
 depends implicitly on $\wb$): $b_\ell$ is the number of entries of $\wb$
 which are equal to $\ell-1$. $b'_k = b_3+\cdots + b_k$. Finally,
 $s_i = \sum_{t=1}^i \max(0,k-w_t-1)$.
  \end{definition}

\begin{lemma}\label{lem:method}
Suppose that $\Delta$ and $k$ are integers with $\Delta\geq 2$ and $k\geq 3$.
Suppose that there are constants $0<\delta<1$ and $\MM>0$ such that, for all
$1\leq d \leq \Delta$, all suitable
$\wb=w_1,\ldots,w_d$,
and all $\r$ satisfying
$(1/2)^{\Delta-1}\ones\leq \r \leq \ones$,
it holds that
\begin{equation}\label{eq:corrdecay}
  \kappa_*^{d,\wb}(\r)\leq \begin{cases}1,& \mbox{ when } d\leq \Delta-1,\\ \MM,& \mbox{ when } d=\Delta.\end{cases}
\end{equation}

Then there exists a constant $\tau>0$ such that, for  every
$C\in \mathcal{C}_{k,\Delta}$, every variable  $x\in C$, and  every integer $L$, it holds that
\begin{equation}\label{eq:target}
|R(C,x,L)-R(C,x,\infty)|\leq \tau \alpha^{L}.
\end{equation}
\end{lemma}

\begin{proof}
 We will show that there is a constant $\hat{\tau}>0$
such that for all such $C$, $x$ and $L$,  it holds that
\begin{equation}\label{eq:qqwwqqb}
|m(C,x,L)-m(C,x,\infty)|\leq \hat{\tau} \alpha^L.
\end{equation}
Assuming \eqref{eq:qqwwqqb} for the moment, let us conclude \eqref{eq:target}.
Consider $C\in \mathcal{C}_{k,\Delta}$. We may assume that $L>0$ since for $L=0$ the inequality \eqref{eq:target} holds for all sufficiently large $\tau$ (any $\tau\geq 1$ works).
Consider any $x\in C$ and consider the computation tree rooted at $(C,x)$.
By the definition of $m(\cdot,\cdot,\cdot)$ and since by assumption $D(C)=0$, we have that
\begin{align}
|m(C,x,L)-m(C,x,\infty)|&=\big|\Phi(R(C,x,L))-\Phi(R(C,x,\infty))\big|.\label{eq:qqwwqq3}
\end{align}
Let   $\eta = (1/2)^{\Delta-1}$ and let
\begin{equation*}
K^{\mathrm{min}}_\Phi:=\min_{x\in [\eta,1]} \Phi'(x)=\min_{x\in [\eta,1]} \phi(x), \qquad \mbox{ and }\qquad K^{\mathrm{max}}_\Phi:=\max_{x\in [\eta,1]} \Phi'(x)=\max_{x\in [\eta,1]} \phi(x).
\end{equation*}
Since $\phi$ is continuous and $\phi(x)>0$ for all $x\in [\eta,1]$, we have that $K^{\mathrm{min}}_\Phi$ and $K^{\mathrm{max}}_\Phi$ are positive. We have that
\begin{equation}\label{eq:qqwwqq4}
|R(C,x,L)-R(C,x,\infty)|\leq \frac{1}{K^{\mathrm{min}}_\Phi}|\Phi(R(C,x,L))-\Phi(R(C,x,\infty))|.
\end{equation}
To see \eqref{eq:qqwwqq4}, we may assume that $R(C,x,L)\neq R(C,x,\infty)$ (otherwise the inequality holds at equality), in which case the inequality follows by an immediate application of the Mean Value Theorem to the function $\Phi$. Combining \eqref{eq:qqwwqq3}, \eqref{eq:qqwwqq4} with \eqref{eq:qqwwqqb} yields
\eqref{eq:target} with $\tau=\hat{\tau}/K^{\mathrm{min}}_\Phi$, as desired.

To prove \eqref{eq:qqwwqqb}, we will first show a slightly weaker claim. Namely, for all nodes $(C,x)$ in the computation tree where $C\in \mathcal{C}_{k,\Delta}$  and $x\in C$ is a variable with degree \emph{$\leq \Delta-1$},  for all integer $L$, it holds that
\begin{equation}\label{eq:qqwwqq}
|m(C,x,L)-m(C,x,\infty)|\leq  K^{\mathrm{max}}_\Phi\alpha^L.
\end{equation}
For $L\leq 0$, we have that
\begin{align*}
|m(C,x,L)-m(C,x,\infty)|&=\delta^{D(C)}|\Phi(R(C,x,L))-\Phi(R(C,x,\infty))|\\
&\leq K^{\mathrm{max}}_\Phi|R(C,x,L)-R(C,x,\infty)|\leq K^{\mathrm{max}}_\Phi,
\end{align*}
where in the first inequality we used that $\delta\in(0,1]$, $D(C)\geq 0$ and an application of the Mean Value theorem analogous to the one used in \eqref{eq:qqwwqq4}, while in the second inequality we used that for $L\leq 0$ it holds that $R(C,x,L)=1$ (by definition) and $0\leq R(C,x,\infty)\leq 1$. Since $0<\alpha<1$, this proves \eqref{eq:qqwwqq} for $L\leq 0$.

To prove \eqref{eq:qqwwqq} for integer $L>0$ we proceed by induction on $L$. Namely, we assume that $L>0$ and  that \eqref{eq:qqwwqq} holds for all smaller values  than $L$ (the base cases $L\leq 0$ have already been shown).

Recall from \eqref{eq:invariant} that $x$ is not forced in $C$ and that $\widetilde{C}$ is the formula obtained by removing the redundant clauses in $C$. We may assume that $x$ is not free in $\widetilde{C}$ --- otherwise, observe that  $R(C,x,L)=R(C,x,\infty)=1$ and thus $m(C,x,L)=m(C,x,\infty)$, so that \eqref{eq:qqwwqq} holds. We will thus focus on $x$ which appear only in (a non-zero number of) clauses in $\widetilde{C}$ of arity $\geq 2$.
Let $d=d_x(\widetilde{C})$. For $i\in [d]$ and $j\in [w_i]$, it holds that
\begin{equation}\label{eq:mmmLinf}
\begin{aligned}
m(C_{i,j}, x_{i,j}, L-l_{w_i})&=\delta^{D(C_{i,j})}\,\Phi(R(C_{i,j},x_{i,j},L-l_{w_i}))\\
 m(C_{i,j}, x_{i,j}, \infty)&=\delta^{D(C_{i,j})}\,\Phi(R(C_{i,j},x_{i,j},\infty)).
\end{aligned}
\end{equation}
Denote by $\r^{(1)}$ the vector whose coordinates are given by $R(C_{i,j},x_{i,j},L-l_{w_i})$ for $i\in [d]$ and $j\in [w_i]$. Denote also by $\r^{(2)}$ the vector whose coordinates are given by $R(C_{i,j},x_{i,j},\infty)$ for $i\in [d]$ and $j\in [w_i]$. Observe that
\[R(C,x,L)=
R(\widetilde{C},x,L)=F^{d,\wb}\big(\r^{(1)}\big) \mbox{ and }
R(C,x,\infty)=
R(\widetilde{C},x,\infty)=F^{d,\wb}\big(\r^{(2)}\big).\]
Note also  that  $\eta\mathbf{1}\leq \r^{(1)},\r^{(2)}\leq \mathbf{1}$ (cf. property \eqref{eq:invariant} for the nodes of the computation tree, \eqref{eq:algo}, footnote \ref{foot:start} and Remark~\ref{rem:boundrs}).

We next bound $\big|\Phi\big(F^{d,\wb}(\r^{(1)})\big)-\Phi\big(F^{d,\wb}(\r^{(2)})\big)\big|$ in terms of $\max_{i,j}\big|\Phi\big(r^{(1)}_{i,j}\big)-\Phi\big(r^{(2)}_{i,j}\big)\big|$. For convenience, denote $F:=F^{d,\wb}$ and for $i\in [d]$ and $j\in [w_i]$, let
\[z^{(1)}_{i,j}:=\Phi\big(r^{(1)}_{i,j}\big),\quad z^{(2)}_{i,j}:=\Phi\big(r^{(2)}_{i,j}\big).\]

For $\theta\in[0,1]$, let $z_{i,j}(\theta):=\theta\,
z^{(1)}_{i,j}+(1-\theta)\,z^{(2)}_{i,j}$ and let
$r_{i,j}(\theta):=\Phi^{-1}(z_{i,j}(\theta))$ (note that the
inverse $\Phi^{-1}$ exists and is uniquely defined in the interval
$\Phi([\eta,1])$
 cf.\ Remark~\ref{remark:invert} and \eqref{eq:conditionsPhi}). Denote by $\r(\theta)$ the vector whose coordinates are $r_{i,j}(\theta)$ and note that $\eta\ones\leq \r(\theta)\leq \ones$. Finally, let   $h(\theta):=\Phi(F(\r(\theta)))$. Observe that $h$ is differentiable for all values of $\theta\in[0,1]$. By applying the Mean Value theorem to the function $h(\theta)$, we obtain that  there exists $\theta_0\in(0,1)$ such that
\begin{align*}
\Phi\big(F\big(\r^{(1)}\big)\big)-\Phi\big(F\big(\r^{(2)}\big)\big)=\left.\frac{\partial h}{\partial \theta}\right|_{\theta=\theta_0}.
\end{align*}
We have that
\begin{align*}
\frac{\partial h}{\partial \theta}&=\frac{\partial \Phi(F(\r(\theta)))}{\partial \theta}=\Phi'(F(\r(\theta)))\frac{\partial F(\r(\theta))}{\partial \theta}=\Phi'(F(\r(\theta)))\left(\sum^{d}_{i=1}\sum^{w_i}_{j=1}\frac{\partial F(\r(\theta))}{\partial r_{i,j}}\frac{\partial r_{i,j}(\theta)}{\partial \theta}\right)\\
&=\Phi'(F(\r(\theta)))\left(\sum^{d}_{i=1}\sum^{w_i}_{j=1}\frac{1}{\Phi'(r_{i,j}(\theta))}\frac{\partial F(\r(\theta))}{\partial r_{i,j}}\frac{\partial z_{i,j}(\theta)}{\partial \theta}\right)\\
&=\Phi'(F(\r(\theta)))\left(\sum^{d}_{i=1}\sum^{w_i}_{j=1}\frac{1}{\Phi'(r_{i,j}(\theta))}\frac{\partial
F(\r(\theta))}{\partial
r_{i,j}}\big(z^{(1)}_{i,j}-z^{(2)}_{i,j}\big)\right).
\end{align*}
It follows that
\begin{align}
  |m&(C,x,L)-m(C,x,\infty)|\notag\\
  &=\delta^{D(C)}|\Phi(R(C,x,L))-\Phi(R(C,x,\infty))|\notag\\
  &\leq \max_{\r}\left\{\delta^{D(C)}\left(\sum^{d}_{i=1}\sum^{w_i}_{j=1}\frac{\phi(F(\r))}{\phi(r_{i,j})}
  \Big|\frac{\partial F(\r)}{\partial r_{i,j}}\Big|\,\big|z^{(1)}_{i,j}-z^{(2)}_{i,j}\big|\right)\right\}\notag\\
  &= \max_{\r}\left(\sum^{d}_{i=1}\sum^{w_i}_{j=1}\delta^{D(C)-D(C_{i,j})}\frac{\phi(F(\r))}{\phi(r_{i,j})}
  \Big|\frac{\partial F(\r)}{\partial r_{i,j}}\Big|\,\big|m(C_{i,j},x_{i,j},L-l_{w_i})-m(C_{i,j},x_{i,j},\infty)\big|\right),\label{eq:nnmm}
\end{align}
where the inequality follows by the triangle inequality and the  fact (see Definition~\ref{phi:definition}) that $\phi=\Phi'$ satisfies \eqref{eq:conditionsPhi}. Now note that for all $i\in [d]$ and $j\in[w_i]$, we have by induction that
\begin{equation}\label{eq:inductivebound}
\big|m(C_{i,j},x_{i,j},L-l_{w_i})-m(C_{i,j},x_{i,j},\infty)\big|\leq K^{\mathrm{max}}_\Phi\, \alpha^{L-l_{w_i}}.
\end{equation}
From \eqref{eq:nnmm} and the bounds \eqref{eq:inductivebound}, we obtain
\begin{equation}\label{eq:mmKLxC}
|m(C,x,L)-m(C,x,\infty)|\leq K^{\mathrm{max}}_\Phi\,
\max_{\r}\left(\sum^{d}_{i=1}\sum^{w_i}_{j=1}
\delta^{D(C)-D(C_{i,j})}\frac{\phi(F(\r))}{\phi(r_{i,j})}\Big|\frac{\partial
F(\r)}{\partial r_{i,j}}\Big|\,\alpha^{-l_{w_i}}\right)\alpha^{L}.
\end{equation}

The lower bounds on $D(C)-D(C_{i,j})$, \eqref{eq:deficits}, imply that
\begin{align}
\sum^{d}_{i=1}\sum^{w_i}_{j=1}&\delta^{D(C)-D(C_{i,j})}\frac{\phi(F(\r))}{\phi(r_{i,j})}
\Big|\frac{\partial F(\r)}{\partial
r_{i,j}}\Big|\,\alpha^{-l_{w_i}}\leq
\kappa_*^{d,\wb}(\r)\label{eq:wsxcderfv-general},
\end{align}
where in the inequality we used that $\delta\in (0,1]$.

Combining \eqref{eq:mmKLxC}, \eqref{eq:wsxcderfv-general} and the
assumption \eqref{eq:corrdecay} that $\kappa_*^{d,\wb}(\r)\leq 1$
for $d\leq \Delta-1$ yields \eqref{eq:qqwwqq}, as wanted.

Finally, we prove \eqref{eq:qqwwqqb} with
$\hat{\tau}:=K^{\mathrm{max}}_\Phi \cdot \max\{\MM,1\}$, where
$\MM$ is the constant in assumption \eqref{eq:corrdecay}. For $x$
of degree $\leq \Delta-1$, we get \eqref{eq:qqwwqqb} immediately
from \eqref{eq:qqwwqq}. Now suppose that $(C,x)$ are such that $x$
has degree $\Delta$ in $C$. Note, for a node $(C',x')$ in the
computation tree, $x'$ may have degree $d=\Delta$ in $C'$ only if
$(C',x')$ is  the root of the tree. It follows that the children
of the node $(C,x)$, say $(C_{i,j},x_{i,j})$ with $i\in [d]$ and
$j\in [w_i]$, are such that $x_{i,j}$ has degree at most
$\Delta-1$ in $C_{i,j}$. Hence, by applying  \eqref{eq:qqwwqq}, we
obtain that \eqref{eq:inductivebound} holds for all $i,j$ and
hence (as before) we deduce that \eqref{eq:mmKLxC},
\eqref{eq:wsxcderfv-general} hold as well. Inequality
\eqref{eq:qqwwqqb} now follows since by assumption
\eqref{eq:corrdecay} we have that $\kappa_*^{d,\wb}(\r)\leq
\max\{\MM,1\}$ for $d\leq \Delta$.

This concludes the proof of Lemma~\ref{lem:method}.
\end{proof}

We now state two technical lemmas which will be proved later in the paper.
These lemmas verify the premise of Lemma~\ref{lem:method} in the settings of Theorems~\ref{thm:main} and~\ref{thm:main2}.
Note from Equation~\eqref{eq:kappa-general} that $\kappa_*^{d,\wb}(\r)$ depends on the global quantity~$k$
and on various quantities (depending on $\wb$) which are defined in Definition~\ref{def:suitable}.

\newcommand{\statedecaygeneral}{Let $k$ and $\Delta$ be two integers such that $k\geq \Delta$ and
 $\Delta \geq 200$. There are constants
 $0<\delta<1$  and $U>0$ such that,
  for all
$1\leq d \leq \Delta$, all suitable
$\wb=w_1,\ldots,w_d$,
and all $\r$ satisfying
$ \zeroes < \r \leq \ones$,
it holds that
  \begin{equation*}
    \kappa_*^{d,\wb}(\r)\leq \begin{cases}1,& \mbox{ when } d\leq \Delta-1,\\ \MM,& \mbox{ when } d=\Delta.\end{cases}
  \end{equation*} }
\begin{lemma}\label{lem:decay-general}
 \statedecaygeneral
\end{lemma}

 \newcommand{\statelempotfn}{Let $\Delta=6$ and
 $k=3$.
There are constants
 $0<\delta<1$  and $U>0$ such that,
  for all
$1\leq d \leq \Delta$, all suitable
$\wb=w_1,\ldots,w_d$,
and all $\r$ satisfying
$(1/2)^{\Delta-1}\ones\leq \r \leq \ones$,
it holds that
  \begin{equation*}
    \kappa_*^{d,\wb}(\r)\leq \begin{cases}1,& \mbox{ when } d\leq \Delta-1,\\ \MM,& \mbox{ when } d=\Delta.\end{cases}
  \end{equation*} }
\begin{lemma}\label{lem:potentialfunctionstar}
 \statelempotfn
\end{lemma}

Lemma~\ref{lem:decay-general} will be proved in Section~\ref{sec:largeDelta} and
Lemma~\ref{lem:potentialfunctionstar} will be proved in Section~\ref{sec:decayrate}.
Using these lemmas and Lemma~\ref{lem:method}, we can prove Lemma~\ref{lem:Delta6}  and Lemma \ref{lem:delta-general}.

{\renewcommand{\thetheorem}{\ref{lem:Delta6}}
\begin{lemma}
\statelemDelta
\end{lemma}
\addtocounter{theorem}{-1}
}

\begin{proof} The lemma follows immediately from Lemma~\ref{lem:potentialfunctionstar} and Lemma~\ref{lem:method}.\end{proof}

{\renewcommand{\thetheorem}{\ref{lem:delta-general}}
\begin{lemma}
\statelemgenDelta
\end{lemma}
\addtocounter{theorem}{-1}
}

\begin{proof} The lemma follows immediately from Lemma~\ref{lem:decay-general} and Lemma~\ref{lem:method}.\end{proof}

\subsection{Proof of the main theorems}\label{sec:theoremproofs}
In this section, we give the proofs of Theorems~\ref{thm:main} and~\ref{thm:main2}, which we restate here for convenience.

 {\renewcommand{\thetheorem}{\ref{thm:main}}
\begin{theorem} \statethmmain
\end{theorem}
\addtocounter{theorem}{-1}
}
\begin{proof} First, reformulate $\#\mathsf{HyperIndSet}(3,6)$
as the monotone CNF problem with instances in $\mathcal{C}_{3,6}$,  following the reformulation in
Section~\ref{sec:reformulation}. Then, just invoke Lemmas~\ref{lem:basic} and \ref{lem:Delta6}.\end{proof}

{\renewcommand{\thetheorem}{\ref{thm:main2}}
\begin{theorem} \statethmtwo
\end{theorem}
\addtocounter{theorem}{-1}
}
\begin{proof}
Once again, reformulate $\#\mathsf{HyperIndSet}(k,\Delta)$ as the monotone CNF problem with instances in $\mathcal{C}_{k,\Delta}$.
Then invoke Lemmas~\ref{lem:basic} and~\ref{lem:delta-general}.
\end{proof}

We have now finished the proofs of Theorems~\ref{thm:main} and~\ref{thm:main2} except that we have
not yet proved the Lemmas~\ref{lem:decay-general} and~\ref{lem:potentialfunctionstar}
which we used to bound $\kappa_*^{d,\wb}(\r)$ in the proofs of Lemmas~\ref{lem:Delta6} and~\ref{lem:delta-general}.
Lemma~\ref{lem:decay-general} will be proved in Section~\ref{sec:largeDelta} and
Lemma~\ref{lem:potentialfunctionstar} will be proved in Section~\ref{sec:decayrate}.
 The proofs   bound the multivariate decay rate function $\kappa_*^{d,\wb}(\r)$.
This is an optimisation problem that  is quite complicated to
solve. Moreover, we need to solve it for all possible suitable
vectors  $\wb$. The analysis of bounding $\kappa_*^{d,\wb}(\r)$ is
the technical core of our proof.

\subsection{Application to counting dominating sets}\label{sec:domsetfptas}
In this section, we use Theorems~\ref{thm:main} and~\ref{thm:main2} to obtain Corollary~\ref{cor:domsetfptas}, which we restate here for convenience.
{\renewcommand{\thetheorem}{\ref{cor:domsetfptas}}
\begin{corollary} \statecorfour
\end{corollary}
\addtocounter{theorem}{-1}
}

The corollary follows from the observation that a dominating set
in a $\Delta$-regular graph is defined by a collection of
constraints of arity $\Delta+1$ with each variable occurring in
$\Delta+1$ constraints. The details are spelled out in the
following proof.

\begin{proof}[Proof of Corollary~\ref{cor:domsetfptas}]
Let $\Delta$ be an integer satisfying either $2\leq \Delta\leq 5$
or $\Delta\geq 199$. Let $G$ be a $\Delta$-regular graph and
denote by $\#\mathrm{DomSets}(G)$ the number of dominating sets in
$G$.

Let $k':=\Delta+1,\Delta':=\Delta+1$. We will construct a $k'$-uniform hypergraph $H$ where every vertex of $H$ has degree $\Delta'$ such that the number $Z_H$ of independent sets in $H$ satisfies $Z_H=\#\mathrm{DomSets}(G)$. To conclude the corollary we then only have to check that, for the relevant range of $\Delta$, we can invoke one of the FPTASes of  Theorems~\ref{thm:main} and~\ref{thm:main2} to approximate $Z_H$. Indeed, when $\Delta\geq 199$, we have $\Delta'\geq 200$ and $k'\geq \Delta'$. Thus, by Theorem~\ref{thm:main2}, $\#\mathsf{HyperIndSet}(k',\Delta')$ admits an FPTAS and thus so does $\RegDomSet$. Similarly, when $2\leq\Delta\leq 5$, we have $\Delta'\leq 6$ and $k'\geq 3$. By Theorem~\ref{thm:main}, $\#\mathsf{HyperIndSet}(3,6)$ admits an FPTAS and thus so does $\RegDomSet$.

We conclude the proof by showing the construction of $H$. For a vertex $v\in V$ of $G$, denote by $v_1,\hdots, v_\Delta$ its neighbours in $G$ (note that there are exactly $\Delta$ of those since $G$ is $\Delta$-regular) and let $e_v=\{v,v_1,\hdots,v_\Delta\}$. Then $H$ is the hypergraph with vertex set $V$ and hyperedge set $\mathcal{F}=\cup_{v\in V}e_v$. It is clear from the construction that $H$ is $k'$-uniform with $k'=\Delta+1$. Further, every vertex $v$ of $H$ has degree $\Delta'=\Delta+1$  since it appears in the hyperedges $e_v,e_{v_{1}},\hdots, e_{v_{\Delta}}$. This completes the construction of $H$.

It remains to show that the number of dominating sets in the graph $G$ is equal to the number of independent sets in the hypergraph $H$. It suffices to show that $S\subseteq V$ is an independent set of $H$ iff $V\backslash S$ is a dominating set of $G$. Indeed, if $S$ is an independent set of $H$, then for every vertex $v\in S$,  at least one of $v_1,\hdots,v_\Delta$ is not in $S$ (since $e_v$ is a hyperedge of $H$) and hence $V\backslash S$ is a dominating set of $G$. Similarly, if $V\backslash S$ is a dominating set of $G$, then for every vertex $v\in S$, at least one of $v_1,\hdots,v_\Delta$ is in $V\backslash S$ and hence $S$ is an independent set of $H$ (since each hyperedge $e_v$ of $H$ contains at least one vertex which does not belong to $S$).

This completes the proof.
\end{proof}

\section{Bounding the decay rate for large \texorpdfstring{$\Delta$}{Delta}} \label{sec:largeDelta}

This section is devoted to proving Lemma~\ref{lem:decay-general}.
 Let $k$ and $\Delta$ be two integers such that $k\geq \Delta$ and $\Delta \geq 200$.
 We start by setting up some upper bounds on the function    $\kappa_*^{d,\wb}(\r)$
 which is defined in \eqref{eq:kappa-general} using notation from Definition~\ref{def:suitable}.
 Consider the following definitions, which apply to suitable vectors~$\wb$.

\begin{align}
  \rho^{\wb, i}(\r) &:=
  \begin{cases}
    \alpha^{-l_{w_i}} \delta^{s_i+i-d+b_2} \sum_{j=1}^{w_i} \delta^{-(j-1)(\Delta-1)}
    \frac{1}{\phi(r_{i,j})}\left|\frac{\partial F^{d,\wb}}{\partial r_{i,j}}\right| & \text{if } 1\le i\le d-b_2,\\
    \alpha^{-l_{w_i}} \delta^{s_{d-b_2}}
    \frac{1}{\phi(r_{i,1})}\left|\frac{\partial F^{d,\wb}}{\partial r_{i,1}}\right| & \text{if } d-b_2+1\le i\le d.
  \end{cases}.\\
  \kappa^{d,\wb}(\r)&:=
  \phi(F^{d,\wb}(\r))\delta^{b_2(k-2)} \sum_{i=1}^{d} \rho^{\wb, i}(\r).
\end{align}

We first argue that $\kappa_*^{d,\wb}(\r) \leq \kappa^{d,\wb}(\r)$.
To see this, note that $s_{\min(i,d-b_2)}$ is $s_i$ for $1\leq i \leq d-b_2$ and
is $s_{d-b_2}$ for $d-b_2+1 \leq i \leq d$.
Also, $b'_k \leq d-b_2$, so
$\max(0,b'_k-i) \leq d-b_2-i$ and if $i\geq d-b_2+1$ then
$\max(0,b'_k-i)=0$.
Finally, $(j-1)(\Delta-1)\mathbf{1}_{i\leq d-b_2}$
is equal to $(j-1)(\Delta-1)$ when $1\leq i \leq d-b_2$ and
to $0$ when $d-b_2+1 \leq i$.

Recall from Definition~\ref{def:alpha} that $\alpha=1-10^{-4}$.
For brevity, we will denote $F^{d,\wb}(\r)$ by $F(\r)$.
Finally, we define the constant $\delta$ (which depends on $\Delta$).
\begin{definition}\label{def:cdelta}
Let $\coe=0.7$.
Given $\Delta$,  define $\delta$  by $\delta^{\Delta}=\coe$.
\end{definition}
In order to bound $\kappa^{d,\wb}(\r)$, an important special case is when $b_2=0$.
Indeed, as we will see soon, handling this special case implies an upper bound of $\kappa^{d,\wb}(\r)$ for the general case.
Define the following quantity for suitable vectors~$\wb$.
\begin{align}
  \widehat{\kappa}^{d,\wb}(\r)
  &:=\frac{1}{\psi - F(\r)^\chi}\sum_{i=1}^{d}\delta^{s_i+i-d-(w_i-1)(\Delta-1)}\alpha^{-l_{w_i}}
  \frac{\prod_{j=1}^{w_i}\frac{r_{i,j}}{1+r_{i,j}}}{1-\prod_{j=1}^{w_i}\frac{r_{i,j}}{1+r_{i,j}}}
  \sum_{j=1}^{w_i} \frac{\psi - r_{i,j}^\chi}{1+r_{i,j}}.
  \label{eqn:kappa-hat}
\end{align}
The next lemma gives an upper bound on the quantity $\widehat{\kappa}^{d,\wb}(\r)$.
\newcommand{\statelemkappabounddb} {Let $k$ and $\Delta$ be two integers such that
  $k \ge \Delta$ and $\Delta\geq 200$. Let $d$ be a positive integer such that $d\le \Delta-1$.
  Let $\wb=w_1,\ldots,w_d$ be a suitable vector
  with $b_2=0$.
  Then,
  for all $\r$ satisfying $\zeroes<\r\le\ones$,
  $\widehat{\kappa}^{d,\wb}(\r)\le 1$.

  In the case $d=\Delta$ for $\wb=w_1,\ldots,w_d$, a suitable vector with $b_2=0$
  and all $\r$ satisfying $\zeroes<\r\le\ones$,
  $\widehat{\kappa}^{d,\wb}(\r)\le 1/\delta$.
  }
  \begin{lemma}
   \label{lem:kappabound:d-b2}
\statelemkappabounddb
\end{lemma}

Note that in the statement of Lemma~\ref{lem:kappabound:d-b2}, the
$w_i$'s are positive integers in non-decreasing order, all of which are at least~$2$.
Lemma \ref{lem:kappabound:d-b2} will be proved in the remainder of Section~\ref{sec:largeDelta}.
First we use it to prove Lemma \ref{lem:decay-general}, which we restate for convenience.
{\renewcommand{\thetheorem}{\ref{lem:decay-general}}
\begin{lemma}
\statedecaygeneral\end{lemma}
\addtocounter{theorem}{-1}
}

\begin{proof}

First note that
  \begin{align*}
    \left|\frac{\partial F}{\partial r_{i,j}}\right|
    = \frac{F(\r)}{r_{i,j}(1+r_{i,j})}\cdot\frac{\prod_{t=1}^{w_i}\frac{r_{i,t}}{1+r_{i,t}}}{1-\prod_{t=1}^{w_i}\frac{r_{i,t}}{1+r_{i,t}}}.
     \end{align*}
  Hence, for $1\le i\le d-b_2$,
  \begin{align*}
    \rho^{\wb, i}(\r)
    & = F(\r) \delta^{s_i+i-d+b_2}
    \alpha^{-l_{w_i}} \frac{\prod_{j=1}^{w_i}\frac{r_{i,j}}{1+r_{i,j}}}{1-\prod_{j=1}^{w_i}\frac{r_{i,j}}{1+r_{i,j}}}
    \sum_{j=1}^{w_i} \delta^{-(j-1)(\Delta-1)}\frac{1}{\phi(r_{i,j})} \frac{1}{r_{i,j}(1+r_{i,j})}\\
    & \le F(\r) \delta^{s_i+i-d+b_2-(w_i-1)(\Delta-1)}
    \alpha^{-l_{w_i}} \frac{\prod_{j=1}^{w_i}\frac{r_{i,j}}{1+r_{i,j}}}{1-\prod_{j=1}^{w_i}\frac{r_{i,j}}{1+r_{i,j}}}\sum_{j=1}^{w_i} \frac{\psi - r_{i,j}^\chi}{1+r_{i,j}};
  \end{align*}
  and for $d-b_2 < i \le d$,
  \begin{align*}
    \rho^{\wb, i}(\r) & = \frac{\alpha^{-l_{1}}F(\r)}{\phi(r_{i,1})} \delta^{s_{d-b_2}} \frac{r_{i,1}}{1+r_{i,1}}\\
    & \le F(\r) \delta^{s_{d-b_2}} \alpha^{-l_{1}} \frac{r_{i,1}(\psi - r_{i,1}^\chi)}{1+r_{i,1}}.
  \end{align*}
  Recall from Definition~\ref{eq:lwi} that $l_{w_i}=\left\lceil \log_6 (w_i+1)\right\rceil$ and in particular that $l_1=\hdots=l_5=1$.
  This implies that
  \begin{align}
    \kappa^{d,\wb}(\r)
    \le &  \frac{\delta^{b_2(k-2)}}{\psi - F(\r)^\chi}\Bigg(\sum_{i=1}^{d-b_2}\delta^{s_i+i-d+b_2-(w_i-1)(\Delta-1)}\alpha^{-l_{w_i}}
    \frac{\prod_{j=1}^{w_i}\frac{r_{i,j}}{1+r_{i,j}}}{1-\prod_{j=1}^{w_i}\frac{r_{i,j}}{1+r_{i,j}}}
    \sum_{j=1}^{w_i} \frac{\psi - r_{i,j}^\chi}{1+r_{i,j}}\notag\\
    & \hspace{3cm} +\sum_{i=d-b_2+1}^{d} \delta^{s_{d-b_2}} \frac{r_{i,1}(\psi - r_{i,1}^\chi)}{\alpha(1+r_{i,1})}\Bigg).
    \label{eqn:kappabound:1}
  \end{align}
  We want to move the term in front of the parentheses in \eqref{eqn:kappabound:1} inside.
  Let $\wb'=\{w_1,\dots,w_{d-b2}\}$ be the prefix of $\wb$.
  Then
  \begin{align*}
    F^{d-b_2,\wb'}(\r)=\prod^{d-b_2}_{i=1}\left(1-\prod^{w_i}_{j=1}\frac{r_{i,j}}{1+r_{i,j}}\right).
  \end{align*}
  First, suppose that $d\leq \Delta-1$.
  By \eqref{eqn:kappabound:1},
  \begin{align}
    \kappa^{d,\wb}(\r)
    \le & \delta^{b_2(k-2)}\Bigg(\frac{1}{\psi - F^{d-b_2,\wb'}(\r)^\chi}\sum_{i=1}^{d-b_2}\delta^{s_i+i-d+b_2-(w_i-1)(\Delta-1)}\alpha^{-l_{w_i}}
    \frac{\prod_{j=1}^{w_i}\frac{r_{i,j}}{1+r_{i,j}}}{1-\prod_{j=1}^{w_i}\frac{r_{i,j}}{1+r_{i,j}}}
    \sum_{j=1}^{w_i} \frac{\psi - r_{i,j}^\chi}{1+r_{i,j}}\notag\\
    & \hspace{3cm} +\frac{1}{\psi - \prod^{d}_{i=d-b_2+1}(1+
    r_{i,1})^{-\chi}}
    \sum_{i=d-b_2+1}^{d} \delta^{s_{d-b_2}} \frac{r_{i,1}(\psi - r_{i,1}^\chi)}{\alpha(1+r_{i,1})}\Bigg)\notag\\
    \le & \delta^{b_2(k-2)}\left(\widehat{\kappa}^{d-b_2,\wb'}(\r)
    +\sum_{i=d-b_2+1}^{d} \frac{1}{\psi - (1+r_{i,1})^{-\chi}} \cdot \frac{r_{i,1}(\psi - r_{i,1}^\chi)}{\alpha(1+r_{i,1})}\right) \notag\\
    \le & \delta^{b_2(k-2)}\left(1
    + \sum_{i=d-b_2+1}^d \frac{1}{\psi - (1+r_{i,1})^{-\chi}} \cdot \frac{r_{i,1}(\psi - r_{i,1}^\chi)}{\alpha(1+r_{i,1})}\right),
    \label{eqn:kappabound:2}
  \end{align}
  where we use the definition of $\widehat{\kappa}^{d,\wb}$ from~\eqref{eqn:kappa-hat} and Lemma \ref{lem:kappabound:d-b2} 
  (using the fact that $d\leq \Delta-1$)
  in the last step.
  In Section~\ref{sec:psi1plusr}, we verify using Mathematica's \textsc{Resolve} function that, for any $0\le r\le 1$, it holds that
  \begin{equation}\label{eq:psi1plusr}
    \frac{1}{\psi - (1+r)^{-\chi}} \cdot \frac{r(\psi - r^\chi)}{\alpha(1+r)} \le 0.42.
  \end{equation}
  Therefore \eqref{eqn:kappabound:2} simplifies into
  \begin{align*}
    \kappa^{d,\wb}(\r)\le \delta^{b_2(k-2)}\left(1 + 0.42 b_2 \right).
  \end{align*}
  Since $\delta^{k-2}\le\delta^{\Delta-2}\le c^{1-2/200}$ and $c=0.7$, we have that
  \begin{align*}
    \kappa_*^{d,\wb}(\r) \leq \kappa^{d,\wb}(\r)\le 1,
  \end{align*}
  for any integer $b_2\ge 0$.
  This finishes the proof.

The bound in the case $d=\Delta$ follows the same argument; the
only difference is that in equation~\eqref{eqn:kappabound:2} we
use the   $d=\Delta$ case of Lemma \ref{lem:kappabound:d-b2}
and hence obtain a weaker (constant) bound on
$\kappa^{d,\wb}(\r)$.
\end{proof}

\subsection{Useful lemmas for the proof of Lemma \ref{lem:kappabound:d-b2}}
 \label{sec:kappa:d-b2}

 We have now finished the proof of Theorem~\ref{thm:main2} apart from the proof of
 Lemma \ref{lem:kappabound:d-b2}, and the remainder of
 Section~\ref{sec:largeDelta}
is devoted to the proof of Lemma \ref{lem:kappabound:d-b2}.
First, in this section, we prove some useful lemmas. To make them easier to read, we list
some useful constants and functions in Tables  \ref{tab:d-b2:constants} and    \ref{tab:d-b2:functions}.
The first two lemmas are merely technical.

\newlength\newboxlength
\settowidth\newboxlength{Value or Definition}

\begin{table}[htbp]
  \begin{minipage}{0.5\textwidth}
  \centering
  \begin{tabular}{|c|c|c|}
    \hline
    Name & \makebox[\newboxlength][c]{Value} & \makebox[\newboxlength][c]{Definition} \\
    \hline
    \hline
    $\alpha$ & $1- 10^{-4}$ & Definition~\ref{def:alpha}\\
    \hline
    $\psi$ & $13/10$ & Definition~\ref{phi:definition} \\
    \hline
    $\chi$ & $1/2$ & Definition~\ref{phi:definition} \\
    \hline
    $c$ & $0.7$ & Definition~\ref{def:cdelta}\\
    \hline
    $\delta$ & $\delta=c^{1/\Delta}$ & Definition~\ref{def:cdelta} \\
    \hline
    $c_5$ & $c^{1-6/200}$ & Definition~\ref{def:cfive}\\
    \hline
    $K_2$ & $1.11614$ & Lemma~\ref{lem:ploqaz1}
    \\
    \hline
    $K_6$ & $1$ & Lemma~\ref{lem:ploqaz1}
    \\
    \hline
    $\tau_2$ & $4.5932$ & Lemma~\ref{lemtaus} \\
    \hline
    $\tau_6$ & $2.7805$& Lemma~\ref{lemtaus} \\
    \hline
    $Y_0$ & $\approx 0.933133$ &\eqref{eqn:Y0Y1}\\
    \hline
    $Y_1$ & $\approx 0.988369$ &\eqref{eqn:Y0Y1}\\
    \hline
  \end{tabular}
  \caption{Constants}
  \label{tab:d-b2:constants}
  \end{minipage}\begin{minipage}{0.5\textwidth}
  \centering
  \begin{tabular}{|c|c|}
    \hline
    Name & \makebox[\newboxlength][c]{Defining equation} \\
    \hline
    \hline
    $\widehat{\kappa}^{d,\wb}(\r)$ & \eqref{eqn:kappa-hat} \\
    \hline
    $\xi(w,r)$ & \eqref{eqn:xi}\\
    \hline
    $\zeta(\wb,d)$ & \eqref{eqn:zeta} \\
    \hline
    $h_1(y)$ & \eqref{eqn:h1(y)} \\
    \hline
    $h(y)$ & \eqref{eqn:h(y)} \\
    \hline
    $\sigma_{t,6}(\y)$ & \eqref{eqn:sigma6} \\
    \hline
    $\sigma_{8,2}(\y)$ & \eqref{eqn:sigma2} \\
    \hline
  \end{tabular}
  \caption{Functions}
  \label{tab:d-b2:functions}
  \end{minipage}
\end{table}

\begin{lemma}\label{lem:con23con45}
Let $f(r):=\frac{\psi - r^\chi}{1+r}$ for $r\in (0,1]$ and parameterize $t$ in terms of $r$ via $e^{t}=\frac{r}{1+r}$. Then, $f$, viewed as a function of $t$, is concave.
\end{lemma}
\begin{proof}
Note that since $r$ ranges from 0 to 1, we have that $t$ ranges from $-\infty$ to $\ln(1/2)$. Further, we have that $r=\frac{e^{t}}{1-e^t}$, so we want to show that the function
\[\hat{f}(t):=f\Big(\frac{e^{t}}{1-e^t}\Big)\]
is concave for all $t\leq \ln(1/2)$. We do this by verifying that $\hat{f}''(t)<0$ for $t\leq \ln(1/2)$ using Mathematica's \textsc{Resolve} function, see Section~\ref{sec:con23con45} for details.
\end{proof}

\begin{lemma}\label{lem:funf1f2}
Let $x,y$ be such that $0<y<x<1$. For integers $d,w\geq 1$, consider the functions $f_1(d):=\frac{1-x^d}{1-y^{d}}$ and $f_2(w):=\frac{1-x^w}{1-x^{w+1}}/\frac{1-y^{w}}{1-y^{w+1}}$. Then, for all integer $d,w\geq 1$, it holds that
\[f_1(d)\leq f_1(d+1), \quad f_2(w)\geq f_2(w+1).\]
\end{lemma}
\begin{proof}
To show that $f_1(d)\leq f_1(d+1)$ for integer $d\geq 1$, we only need to show that $\bar{f_1}(y)\leq\bar{f_1}(x)$ where $\bar{f_1}(t):=\frac{1-t^{d+1}}{1-t^{d}}$, i.e., that $\bar{f_1}(t)$ is increasing for $t\in(0,1)$. We calculate
\[\bar{f_1}'(t)=\frac{t^{d-1} \left(t^{d+1}-(d+1)t+d\right)}{\left(1-t^d\right)^2},\]
which is nonnegative for all $t\in(0,1)$ since  the quantity $t^{d+1}-(d+1)t+d$ is nonnegative for $t\in(0,1)$ (it vanishes at $t=1$ and is a decreasing function in $t$ for $t\in(0,1]$ since the first derivative equals $(d+1)(t^d-1)\leq 0$).

To show that $f_2(w)\geq f_2(w+1)$ for integer $w\geq 1$, we only need to show that $\bar{f_2}(x)\leq\bar{f_2}(y)$ where $\bar{f_2}(t):=\frac{(1-t^{w})(1-t^{w+2})}{(1-t^{w+1})^2}$, i.e., that $\bar{f_2}(t)$ is decreasing for $t\in(0,1)$. We calculate
\[\bar{f_2}'(t)=\frac{(1-t) t^{w-1} \left(w t^{w+2}-(w+2) t^{w+1}+t (w+2)-w\right)}{\left(1-t^{w+1}\right)^3},\]
which is at most~$0$ for all $t\in(0,1)$ since  the quantity $w t^{w+2}-(w+2) t^{w+1}+t (w+2)-w$ is at most~$0$  for $t\in[0,1]$ (it vanishes at $t=1$, it is negative at $t=0$ and is a concave function in $t$ for $t\in[0,1]$ since the second derivative equals $w(w+2)(w+1)(t^w-t^{w-1})\leq 0$).
\end{proof}

We now turn to the problem of upper bounding $\widehat{\kappa}^{d,\wb}(\r)$ which is what
we need to do to prove Lemma~\ref{lem:kappabound:d-b2}.
Recall from Definition~\ref{def:suitable} that  $s_i = \sum_{t=1}^i \max(0,k-w_t-1)$.
Thus, for $i\geq 1$ and   $d\le \Delta-1$, we have
\begin{align*}
  s_i+i-d-(w_i-1)(\Delta-1) & \ge s_i+i-w_i(\Delta-1) \\
  & \ge s_{i-1}+k-w_i-1+i-w_i(\Delta-1)\\
  & \ge s_{i-1}+k-w_i\Delta.
\end{align*}
Hence, from the definition of~$\widehat{\kappa}^{d,\wb}(\r)$, we have
\begin{align}
  \widehat{\kappa}^{d,\wb}(\r) \le &
  \frac{\delta^k}{\psi - F(\r)^\chi}\Bigg(\sum_{i=1}^{d}\delta^{s_{i-1}-w_i\Delta}\alpha^{-l_{w_i}}
  \frac{\prod_{j=1}^{w_i}\frac{r_{i,j}}{1+r_{i,j}}}{1-\prod_{j=1}^{w_i}\frac{r_{i,j}}{1+r_{i,j}}}
  \sum_{j=1}^{w_i} \frac{\psi - r_{i,j}^\chi}{1+r_{i,j}}\Bigg).
  \label{eqn:b2=0}
\end{align}

Recall from~\eqref{eq:recursion} that $F(\r)=\prod_{i=1}^d\left( 1-\prod_{j=1}^{w_i}\frac{r_{i,j}}{1+r_{i,j}} \right)$.
Define $t_{i,j}$ by $e^{t_{i,j}} = r_{i,j}/(1+r_{i,j})$ and let $u_i = \sum_{j=1}^{w_i} t_{i,j}$.
Then we can express the right-hand-side of~\eqref{eqn:b2=0}
as a function of the $t_{i,j}$'s as follows.
$$  \frac{\delta^k}{\psi -
{\left(\prod_{i=1}^d (1-e^{u_i})\right)}^\chi}\Bigg(\sum_{i=1}^{d}\delta^{s_{i-1}-w_i\Delta}\alpha^{-l_{w_i}}
  \frac{e^{u_i}
  }{1- e^{u_i}}
  \sum_{j=1}^{w_i}  \hat{f}(t_{i,j})
  \Bigg),$$
  where $\hat{f}$ is defined in Lemma~\ref{lem:con23con45}.
  By Lemma~\ref{lem:con23con45}, $\hat{f}$ is concave, so Jensen's inequality applies,
  showing that the quantity  is at most
$$  \frac{\delta^k}{\psi -
{\left(\prod_{i=1}^d (1-e^{u_i})\right)}^\chi}\Bigg(\sum_{i=1}^{d}\delta^{s_{i-1}-w_i\Delta}\alpha^{-l_{w_i}}
  \frac{e^{u_i}
  }{1- e^{u_i}} w_i
    \hat{f}\left(\frac{u_i}{w_i}\right)
  \Bigg).$$
 So we can replace each $t_{i,j}$ with $u_i/w_i$ without decreasing the right-hand-side.
 Equivalently, we can replace each $r_{i,j}$ with a quantity $r_i$
 so that
$\prod_{j=1}^{w_i}\frac{r_{i,j}}{1+r_{i,j}} =
\left(\frac{r_i}{1+r_i}\right)^{w_i} $.

Define $y_{i}$ such that $y_i=\left( 1- \left( \frac{r_{i}}{1+r_{i}} \right)^{w_i} \right)^{1/2}$.
Also notice that $\delta^k\le\delta^{\Delta} =c=0.7$.
From \eqref{eqn:b2=0}, we thus obtain that :
\begin{align}
  \widehat{\kappa}^{d,\wb}(\r) \le & \frac{\delta^{k}}{\psi - \prod_{i=1}^d{y_i}}\sum_{i=1}^{d}\delta^{s_{i-1}}g(y_i,w_i),
  \label{eqn:y_i}
\end{align}
where
\begin{align}
  g(y,w):=\coe^{-w}\alpha^{-l_{w}}w\frac{1-y^2}{y^2} \left( 1-(1-y^2)^{1/w} \right)
  \left( \psi - \left( \frac{(1-y^2)^{1/w}}{1-(1-y^2)^{1/w}} \right)^\chi \right),
  \label{eqn:g(y,w)}
\end{align}
and the range of $y_i$ is
$y_i\in[(1-2^{-w_i})^{1/2},1)$.
The next lemma gives an upper bound on $g(y,w)$.

\begin{lemma}\label{lem:ploqaz1}
For all $w\geq 2$, for all
$y\in [(1-2^{-w})^{1/2},1)$, it holds that
\begin{equation}\label{eq:ploqaz1}
g(y,w)\le K_{w}\, g((1-2^{-w})^{1/2},w)=0.15K_w \alpha^{-l_{w}} w (2c)^{-w} \left(1-2^{-w}\right)^{-1},
\end{equation}
where $K_2 = 1.11614$, $K_3 = 1.03$, $K_4 = 1.01$, and $K_w=1$ for all $w\ge 5$. In particular, $K_{w+1}\le K_w$ for any $w\ge 2$.
\end{lemma}
\begin{proof}
The equality in \eqref{eq:ploqaz1} is immediate by just substituting $y=(1-2^{-w})^{1/2}$ into $g(y,w)$, so we only need to argue for the inequality in \eqref{eq:ploqaz1}.

We make the change of variable $t=(1-y^2)^{1/w}$ so that the range of $t$ is
$(0,1/2]$.
The inequality can then be written as
\begin{equation}\label{eq:hatghat12}
\frac{\hat{g}(t,w)}{\hat{g}(1/2,w)}\leq K_w \mbox{ where } \hat{g}(t,w):=\frac{t^w}{1-t^w}(1-t)\bigg(\psi-\Big(\frac{t}{1-t}\Big)^{\chi}\bigg).
\end{equation}
We verify \eqref{eq:hatghat12} for $w=2,3,4,5$ and all
$t\in(0,1/2]$
using Mathematica's \textsc{Resolve} function, see Section~\ref{sec:ploqaz1} for details. To obtain the lemma for $w\geq 6$, it then suffices to show that
\[\frac{\hat{g}(t,w+1)}{\hat{g}(1/2,w+1)}\leq \frac{\hat{g}(t,w)}{\hat{g}(1/2,w)} \mbox{ for all
$t\in (0,1/2]$.
}\]
This can be massaged into
\[\overline{g}(t)\leq \overline{g}(1/2), \mbox{ where } \overline{g}(t):=\frac{\hat{g}(t,w+1)}{\hat{g}(t,w)}=\frac{t(1-t^w)}{1-t^{w+1}},
\]
so it suffices to show that $\overline{g}$ is increasing for
$t\in (0,1/2]$.
We calculate
\[\overline{g}'(t)=\frac{w t^{w+1}-(w+1) t^w+1}{\left(1-t^{w+1}\right)^2},\]
which is positive for all $t\in (0,1)$ since the function $f(t):=w t^{w+1}-(w+1) t^w+1$ satisfies $f(1)=0$ and $f'(t)
= w(w+1)t^w(1-\frac1t)
<0$ for all $t\in (0,1)$.
\end{proof}

For every positive integer~$w$, define $\xi(w,0)=0$.
For every positive integer~$d$, define
\begin{equation}\label{eqn:xi}
    \xi(w,d) :=\alpha^{-l_{w}}(2 c)^{-w} w\left(1-2^{-w}\right)^{-1}\sum_{i=0}^{d-1} \delta^{\max\{0,k-w-1\} i}.
\end{equation}
In order to upper bound $\widehat{\kappa}^{d,\wb}(\r)$
it is going to be useful to know when $\xi(w,d)$ is decreasing in~$w$. This is captured by the following lemma.

\begin{lemma}\label{lem:xidecreasing}
Let $k,\Delta$ be integers satisfying $k\geq \Delta$, $\Delta\geq 200$. Then for all integers $d$ such that $1\leq d\leq \Delta$,  for all integers $w$ satisfying either $2\leq w\leq k-3$ or $w\geq k-1$, it holds that
$\xi(w+1,d)\leq \xi(w,d)$.
\end{lemma}

\begin{proof}
Recall  from Definition~\ref{eq:lwi} that
$l_{w}=\left\lceil \log_6 (w+1)\right\rceil$.
Thus,  $l_w-l_{w+1}\in\{0,-1\}$. Using this fact and the fact that $0<\alpha<1$, we obtain that
\begin{align}
    \frac{\xi(w+1,d)}{\xi(w,d)}
    &=\alpha^{l_w-l_{w+1}}(2c)^{-1} (1+w^{-1}) \frac{1-2^{-w}}
    {1-2^{-w-1}}\frac{\sum_{i=0}^{d-1}  \delta^{\max\{0,k-w-2\} i}}{\sum_{i=0}^{d-1} \delta^{\max\{0,k-w-1\} i}}\notag\\
    &\leq (2\alpha c)^{-1}(1+w^{-1}) \frac{1-2^{-w}}{1-2^{-w-1}} D_{w,d}, \label{eq:xscvf123}
\end{align}
where
\begin{equation*}
 D_{w,d}:=\frac{\sum_{i=0}^{d-1} \delta^{\max\{0,k-w-2\} i}}{\sum_{i=0}^{d-1} \delta^{\max\{0,k-w-1\} i}},
    \end{equation*}
Note that when $w\geq k-1$, $D_{w,d}$ is trivially equal to 1, so for this range of $w$ the right hand side of~\eqref{eq:xscvf123} is upper bounded by 1 (since $w\geq k-1\geq 199$ and, from $\alpha=1-10^{-4}$ and $c=0.7$, we have $(2\alpha c)^{-1}<3/4$).

Henceforth, we will thus focus on the range $2\leq w\leq k-3$.  To bound the right-hand side    of~\eqref{eq:xscvf123}, we will use the following bound on $D_{w,d}$:
    \begin{equation}\label{eq:tgbrfvcde}
    D_{w,d}\leq M_w, \mbox{ where } M_w:=\begin{cases} 2/(1+c),& \text{if } k-3\geq w\geq 6,\\ 1.05,& \text{if } 2\leq w\leq 5.\end{cases}
    \end{equation}
We will justify the bound in \eqref{eq:tgbrfvcde} shortly, for now note that the desired inequality
$\xi(w+1,d)/\xi(w,d)\leq 1$
will follow from
    \begin{equation}\label{eq:Mwac23}
    M_w(2ac)^{-1}(1+w^{-1}) \frac{1-2^{-w}}{1-2^{-w-1}}\leq 1.
    \end{equation}
We verify \eqref{eq:Mwac23} for $w=2,3,4,5$ and also show that for
$w=6$ we have $M_w(2ac)^{-1}(1+w^{-1})<1$ (for $w\geq 7$ the
result follows by monotonicity), see
Section~\ref{sec:xidecreasing} for the code.

We finish the proof by justifying the bound in \eqref{eq:tgbrfvcde}. Note that for all $2\leq w\leq k-3$, we have that
\[D_{w,d}=\frac{1-\delta^{(k-w-2)d}}{1-\delta^{(k-w-1)d}}\cdot\frac{1-\delta^{k-w-1}}{1-\delta^{k-w-2}}\]
To prove \eqref{eq:tgbrfvcde} for $2\leq w\leq 5$, we use the fact that $k\geq \Delta\geq 200$,
the definition $\delta^{\Delta}=c$ and the fact that the function $x/(1-x)$ is increasing for $x\in (0,1)$ to obtain
\begin{align}
  D_{w,d}&\leq\frac{1-\delta^{k-w-1}}{1-\delta^{k-w-2}}= 1+\frac{\delta^{k-w-2}}{1-\delta^{k-w-2}}(1-\delta)\leq1+\frac{c^{1-7/\Delta}}{1-c^{1-7/\Delta}}(1-c^{1/\Delta})\notag\\
&\leq 1+\frac{c^{1-7/200}}{1-c^{1-7/200}}(1-c^{1/200})\leq 1.05,\label{eq:cstuffcrude}
\end{align}
see Section~\ref{sec:xidecreasing} for the calculation in the last inequality. To prove \eqref{eq:tgbrfvcde} for $6\leq w\leq k-3$, we will show that
\begin{equation}\label{eq:6tgb23}
D_{w,d}\leq D_{w,\Delta}\leq D_{k-3,\Delta}=\frac{1+\delta}{1+c}<2/(1+c).
\end{equation}
To justify \eqref{eq:6tgb23}, note that the equality is immediate using that $\delta^\Delta=c$. Also, the very last (strict) inequality in \eqref{eq:6tgb23} is immediate using   $0<\delta<1$. In the following, we may thus focus on proving the first two inequalities in \eqref{eq:6tgb23}. To justify the first inequality in \eqref{eq:6tgb23}, i.e., $D_{w,d}\leq D_{w,\Delta}$ for $1\leq d\leq \Delta$, just use Lemma~\ref{lem:funf1f2} for the function $f_1$ with $x=\delta^{k-w-2}$ and $y=\delta^{k-w-1}$. For the second inequality in \eqref{eq:6tgb23}, i.e., $D_{w,\Delta}\leq D_{k-3,\Delta}$ for $6\leq w\leq k-3$, note that, by $\delta^{\Delta}=c$,
\[D_{w,\Delta}=\frac{1-c^{(k-w-2)}}{1-c^{(k-w-1)}}\cdot\frac{1-\delta^{k-w-1}}{1-\delta^{k-w-2}}=\frac{1-c^{w'}}{1-c^{w'+1}}/\frac{1-\delta^{w'}}{1-\delta^{w'+1}}, \mbox{ where } w':=k-w-2.\]
Then $D_{w,\Delta}\leq D_{k-3,\Delta}$ when $w\leq k-3$ (i.e., $w'\geq 1$) follows from Lemma~\ref{lem:funf1f2} for the function $f_2$ with $x=\delta$ and $y=c$.

This concludes the proof of the bound in \eqref{eq:tgbrfvcde}, which completes the proof of Lemma~\ref{lem:xidecreasing}.
\end{proof}

For $\wb = \{w_1,\ldots,w_d\}$ define the quantity
\begin{align}
  \zeta(\wb,d):= \sum_{i=1}^{d} \delta^{s_{i-1}} \alpha^{-l_{w_i}}(2 c)^{-w_i} w_i\left(1-2^{-w_i}\right)^{-1}.
  \label{eqn:zeta}
\end{align}
In order to upper bound $\widehat{\kappa}^{d,\wb}(\r)$
it is going to be useful to have an upper bound on $\zeta(\wb,d)$ and this is done in the following lemma.

\newcommand{\statelemzeta}{
Let $k,\Delta$ be integers satisfying $k\geq \Delta$, $\Delta \geq 200$.
Let $\tau_2:=4.5932$ and $\tau_6:= 2.7805$.
Consider $d\leq \Delta$ and
let $\wb=w_1,\ldots,w_d$ be a vector of integers satisfying $w_1 \leq \dots \leq w_d$.
  \begin{enumerate}
    \item  If $w_1 \geq 2$ then  $\zeta(\wb,d) \le \tau_2$.
    \item  If $w_1 \geq 6$ then  $\zeta(\wb,d) \le \tau_6$.
  \end{enumerate}}
\begin{lemma}\label{lemtaus}
  \label{lem:zeta}
  \statelemzeta
\end{lemma}
\begin{proof}

We wish to find an upper bound for $\zeta(\wb,d)$.
The reason that the task is difficult is that the vector $\wb$ may have up to $d$~different entries.
We say that  a  vector $\wb'=\{w'_1,\ldots,w'_d\}$ of integers \emph{dominates} $\wb$ if
the following are true.
\begin{itemize}
\item $w'_1 = w_1$,
\item $w'_1\leq \dots \leq w'_d$,
\item For all $i$, $w'_i \leq w_i$, and
\item $\zeta(\wb,d) \leq \zeta(\wb',d)$.
\end{itemize}
So a good way to find an upper bound for $\zeta(\wb,d)$ is to find a
``simple'' vector $\wb'$ which dominates $\wb$ and then find an upper bound for $\zeta(\wb',d)$.
To do this, we define several classes of   vectors~$\wb$, depending on
how ``desirable'' they are for proving upper bounds.

\begin{itemize}
\item A vector $\wb$ of integers is ``partly good'' if $2\leq w_1 \leq \dots \leq w_d$.
\item A vector $\wb$ is ``fairly good'' if it is partly good and
 $w_1 = \dots = w_d$ or $w_d \leq k-1$.
\item A vector $\wb$ is ``very good'' if it is fairly good and  every $w_i$ is in $\{w_1,k-1\}$.
 \end{itemize}

 The vector $\wb$ that we start with (in the statement of the lemma) is partly good, but it will be easiest to prove upper bounds
 on $\zeta(\wb,d)$ for very good vectors~$\wb$.
Thus, we will define two transformations.
\begin{enumerate}
\item The first transformation starts with a partly good vector
$\wb$. If $\wb$ is fairly good, then the transformation does
nothing. Otherwise, it produces a partly good vector $\wb'\neq\wb$
which dominates $\wb$. \item The second transformation starts with
a fairly good vector $\wb$. If $\wb$ is very good, then the
transformation does nothing. Otherwise, it produces a fairly good
vector $\wb'\neq\wb$ which dominates $\wb$.
 \end{enumerate}
 Both transformations make progress in the sense that there exists an $i$ such that
 $w'_i < w_i$.

If we start with any partly good vector $\wb$  and repeatedly apply
the first transformation  then, after a finite number of transformations,
we must obtain a fairly good vector $\wb'$ which dominates~$\wb$. (The reason that a finite
number of transformations suffices is that each individual transformation makes progress,
but  the entries stay sorted, and the first coordinate never changes.)  Next, we
apply the second transformation repeatedly, starting from~$\wb'$. Again, after
a finite number of transformations,
we end up with a very good vector $\wb''$ which dominates $\wb'$ and therefore dominates~$\wb$.
An upper bound on $\zeta(\wb'',d)$ gives an upper bound on $\zeta(\wb,d)$.
 So to finish the proof, we must show that the two transformations are possible.
Then we must show that for every very good vector $\wb$ with $w_1\geq 2$,
$\zeta(\wb,d) \leq \tau_2$ and we must show that for every very good vector $\wb$ with
$w_1\geq 6$, $\zeta(\wb,d) \leq \tau_6$.

{\bf Transformation 1:}  Start with a partly good vector $\wb$.
If $\wb$ is fairly good, do nothing. Otherwise, $w_d > \max(w_1,k-1)$.
Choose the integer~$t$
to be as small as possible, subject to the constraint that,
for all $i$ in $t+1,\ldots,d$, we have $w_i=w_d$.
Note that $1\leq t < d$ and $w_t < w_d$.
 Recall
from Definition~\ref{def:suitable}
that
$s_i = \sum_{j=1}^i \max(0,k-w_j-1)$. Since   $w_{t+1}=\cdots = w_d $
 we have for any $j\in \{0,\ldots,d-t\}$
 that $s_{t+j}=s_{t} + \max\{0,k-w_d-1\}j$.
 This means that
 $\sum_{i=t+1}^d \delta^{s_{i-1}} = \delta^{s_t} \sum_{i=0}^{d-t-1}\delta^{\max\{0,k-w_d-1\}i}$.
 So,
  \begin{equation}\label{eq:thy34565}
   \zeta( \wb,d)=
   \sum_{i=1}^{t} \delta^{s_{i-1}} \alpha^{-l_{w_i}}(2 c)^{-w_i} w_i\left(1-2^{-w_i}\right)^{-1} + \delta^{s_{t}}\xi(w_d,d-t).
  \end{equation}
By Lemma~\ref{lem:xidecreasing},
$\xi(w_d,d-t) \leq \xi(\max(w_t,k-1),d-t)$,
so decreasing
$w_{t+1}=\cdots = w_d$ from $w_d$   to
$\max(w_t,k-1)$
does not decrease $\zeta(\wb,d)$.
Thus, the transformation sets $w'_{t+1}=\cdots = w'_d = \max(w_t,k-1)$
and, $w'_1,\ldots,w'_t = w_1,\ldots,w_t$.

{\bf Transformation 2:} Start with a fairly good vector $\wb$.
If $\wb$ is very good, do nothing. Otherwise, $w_1 < w_d \leq k-1$.
Choose the integer~$t$
to be as small as possible, subject to the constraint that,
for all $i$ in $t+1,\ldots,d$, we have $w_i=k-1$.
 Clearly, $t\leq d$.
 We defined $\xi(w,0)$ to be~$0$, so from~\eqref{eq:thy34565}, we have
\begin{equation}\label{eq:edcfvr456}
\zeta( \wb,d)=
   \sum_{i=1}^{t} \delta^{s_{i-1}} \alpha^{-l_{w_i}}(2 c)^{-w_i} w_i\left(1-2^{-w_i}\right)^{-1} + \delta^{s_{t}}\xi(k-1,d-t).
\end{equation}
Now, choose the integer $t'$ to be small as possible,
subject to the constraint that,
for all $i \in \{t'+1,\ldots,t\}$, $w_i=w_{t}$.
Since $\wb$ is not very good, $ 1\leq t' < t$ and $w_{t'}<w_t < k-1$.
Since $w_{t'+1} = \dots = w_t$, we can decompose the right-hand side of \eqref{eq:edcfvr456} as
  \begin{align*}
   \zeta( \wb,d)=
   \sum_{i=1}^{t'} \delta^{s_{i-1}} \alpha^{-l_{w_i}}(2 c)^{-w_i} w_i\left(1-2^{-w_i}\right)^{-1}+\delta^{s_{t'}}\xi(w_{t},t-t') + \delta^{s_{t}}\xi(k-1,d-t).
  \end{align*}
  Now $w_{t'} < w_{t} < k-1$  implies
 $w_{t'} \leq k-3$ and $w_t \leq k-2$ so by Lemma~\ref{lem:xidecreasing},
$\xi(w_t,d-t) \leq \xi(w_{t'},d-t)$,
so decreasing
$w_{t'+1}=\cdots = w_t$ from $w_t$   to
$w_{t'}$
does not decrease $\zeta(\wb,d)$.
Thus, the transformation sets $w'_{t'+1}=\cdots = w'_t = w_{t'}$
and  $w'_1,\ldots,w'_{t'} = w_1,\ldots,w_{t'}$
and $w'_{t+1} = \cdots = w'_d = w_{t+1} = \cdots = w_d = k-1$.

 To finish the proof, we must
 show that for every very good vector $\wb$ with $w_1\geq 2$,
$\zeta(\wb,d) \leq \tau_2$ and we must show that for every very good vector $\wb$ with
$w_1\geq 6$, $\zeta(\wb,d) \leq \tau_6$.
Now if $\wb$ is a very good vector, then
choose $t$ as small as possible so that
$w_{t+1} = \dots = w_d = k-1$.
Note that $0\leq t \leq d$ and $w_1 = \dots = w_t$ and $w_{t+1} = \dots = w_d = k-1$. Thus,
  \begin{align}\label{temptemp}
    \zeta(\wb,d)   \le \xi(w_1,t) + \xi(k-1,d-t).
  \end{align}

To bound the terms in \eqref{temptemp}, note that
$\xi(k-1,0)=0$ and
for any $1\leq d-t\leq \Delta\leq k$ it holds that
\begin{equation}\label{eq:crude12ws}
\xi(k-1,d-t)=\alpha^{-l_{k-1}}(2 c)^{-(k-1)} (k-1)\left(1-2^{-(k-1)}\right)^{-1}(d-t)\leq 2k^{2}(2\alpha c)^{-(k-1)}<10^{-10},
\end{equation}
where the last inequality follows from $(2\alpha c)^{-1}<3/4$ and $k\geq 200$.

For $t=0$, we have $\xi(w_1,t)=0$ and $\xi(k-1,d-t)<10^{-10}$ by \eqref{eq:crude12ws}. Thus assume $t>0$ so that $w_1<k-1$.  Then, by Lemma~\ref{lem:xidecreasing}, for $k-1>w_1\geq 2$, we have $\xi(w_1,t)\leq \xi(2,t)$ and
\begin{align}
\xi(2,t) &=\alpha^{- 1}(2 c)^{-2} 2\left(1-2^{-2}\right)^{-1}\sum_{i=0}^{t-1} \delta^{(k-3) i}. \notag\\
&\leq\alpha^{- 1}(2 c)^{-2} 2\left(1-2^{-2}\right)^{-1}\frac{1}{1-\delta^{k-3}}\notag\\
&=\alpha^{- 1}(2 c)^{-2} 2\left(1-2^{-2}\right)^{-1}\frac{1}{1-c^{(k-3)/\Delta}}\notag\\
& \leq\alpha^{-1}(2 c)^{-2}
2\left(1-2^{-2}\right)^{-1}\frac{1}{1-c^{1-3/200}}.
\label{eq:xi2t}
\end{align}
The numerical calculation in Section~\ref{sec:zeta} shows that this is at most $4.5931$. Since this plus $10^{-10}$ is less than $\tau_2$, we obtain the first part of the lemma.
Similarly,
  \begin{equation}\label{eq:xi6t}
    \xi(6,t) \le \alpha^{-2}(2 c)^{-6} 6\left(1-2^{-6}\right)^{-1}\frac{1}{1-c^{1-7/200}} \le 2.78045,
  \end{equation}
see Section~\ref{sec:zeta} for the calculation in the last inequality. Since $2.78045+10^{-10}<\tau_6$, we obtain the second part of the lemma, and we have finished the proof of the lemma.\end{proof}

\subsection{A quick proof of a weaker Theorem}

Our goal is to prove Lemma~\ref{lem:kappabound:d-b2},  but the proof, which will be
given in Section~\ref{sec:harder},
is a little bit technical.
In order to give the intuition, without getting into technical details, we first state and
prove a weaker version of the lemma.
Lemma~\ref{lem:weak}, below, is identical to Lemma~\ref{lem:kappabound:d-b2}, except that
the condition $k\geq \Delta$ has been strengthened to $k\geq \yuck \Delta$.
Using Lemma~\ref{lem:weak} in place of~Lemma \ref{lem:kappabound:d-b2}
strengthens the condition to $k\geq \yuck \Delta$ in Lemmas~\ref{lem:decay-general}
and~\ref{lem:delta-general}. Thus, Lemma~\ref{lem:weak}
gives immediately a weaker version of Theorem~\ref{thm:main2}
where the condition $k\geq \Delta$ is replaced with $k\geq \yuck \Delta$.

  \begin{lemma}
   \label{lem:weak}
   Let $k$ and $\Delta$ be two integers such that
  $k \ge \yuck \Delta$ and $\Delta\geq 200$. Let $d$ be a positive integer such that $d\le \Delta-1$.
  Let $\wb=w_1,\ldots,w_d$ be a suitable vector
  with $b_2=0$.
  Then,
  for all $\r$ satisfying $\zeroes<\r\le\ones$,
  $\widehat{\kappa}^{d,\wb}(\r)\le 1$.
  
    In the case $d=\Delta$ for $\wb=w_1,\ldots,w_d$, a suitable vector with $b_2=0$
  and all $\r$ satisfying $\zeroes<\r\le\ones$,
  $\widehat{\kappa}^{d,\wb}(\r)\le 1/\delta$.  
\end{lemma}
\begin{proof}
Recall that for $d\leq \Delta-1$,
\begin{equation*}\tag{\ref{eqn:y_i}}
\widehat{\kappa}^{d,\wb}(\r) \le  \frac{\delta^k}{\psi - \prod_{i=1}^d{y_i}}\sum_{i=1}^{d}\delta^{s_{i-1}}g(y_i,w_i),
\end{equation*}
where $y_i=\left( 1- \left( \frac{r_{i}}{1+r_{i}} \right)^{w_i} \right)^{1/2}$
 and  $y_i\in[(1-2^{-w_i})^{1/2},1)$.
 From Lemma~\ref{lem:ploqaz1}
 we have
 $g(y_i,w_i) \leq 0.15K_2 \alpha^{-l_{w_i}} w_i (2c)^{-w_i} \left(1-2^{-w_i}\right)^{-1}$,
 so, from the definition of $\zeta(\wb,d)$ (see \eqref{eqn:zeta}),
 \begin{equation}\label{eq:kapwidewide}
\widehat{\kappa}^{d,\wb}(\r) \le \frac{\delta^k}{\psi -  1} 0.15 K_2 \zeta(\wb,d)\leq \frac{c^{2.64}}{\psi -  1} 0.15 K_2 \tau_2,
\end{equation}
where in the last inequality we used the fact that $\zeta(\wb,d) \leq \tau_2$ from the first part of  Lemma~\ref{lem:zeta} and
the fact that $\delta^{k}\leq \delta^{2.64\Delta}=c^{2.64}$. It is  a matter of a simple numerical calculation to check that the right hand side of \eqref{eq:kapwidewide} is less than 1, see Section \ref{sec:kapwidewide} for details.
Thus, we have shown that $\widehat{\kappa}^{d,\wb}(\r) < 1$.

The case $d=\Delta$ has the same proof, the only difference is
that in  
equations~\eqref{eqn:b2=0} 
and~\eqref{eqn:y_i}
we replace $\delta^k$ by $\delta^{k-1}$,
losing a factor of $1/\delta$ (the upper bound is valid using the
same argument as before, now using inequality
$s_i+i-d-(w_i-1)(\Delta-1)\geq s_{i-1}+(k-1)-w_i\Delta$).

\end{proof}

\subsection{The Proof of  Lemma~\ref{lem:kappabound:d-b2}}\label{sec:harder}

Recall that our actual goal is to prove Lemma~\ref{lem:kappabound:d-b2},
which is stronger than Lemma~\ref{lem:weak} because it only assumes $k\geq \Delta$, not $k\geq \yuck \Delta$. Recall that
\begin{equation*}\tag{\ref{eqn:y_i}}
\widehat{\kappa}^{d,\wb}(\r) \le  \frac{\delta^k}{\psi - \prod_{i=1}^d{y_i}}\sum_{i=1}^{d}\delta^{s_{i-1}}g(y_i,w_i),
\end{equation*}
where $y_i=\left( 1- \left( \frac{r_{i}}{1+r_{i}} \right)^{w_i} \right)^{1/2}$
 so that $y_i\in[(1-2^{-w_i})^{1/2},1)$ and the function $g(y,w)$ is given by \eqref{eqn:g(y,w)}. Using
 the fact that $k\geq \Delta$ and $\delta^{\Delta}=c$ we obtain the inequality
\begin{equation}\label{eqn:y_i2345}
\widehat{\kappa}^{d,\wb}(\r) \le  \frac{c}{\psi - \prod_{i=1}^d{y_i}}\sum_{i=1}^{d}\delta^{s_{i-1}}g(y_i,w_i).
\end{equation}
Lemma~\ref{lem:ploqaz1} gives an upper bound on $g(y,w)$ in terms of the constants~$K_w$.
Since $K_w=1$ for all $w\ge 5$, we want to split the sum in \eqref{eqn:y_i2345} for $w\le 5$ and $w\ge 6$.
More generally, we split the summation in the bound \eqref{eqn:y_i2345} at an index $t\le d$
using Lemma~\ref{lem:ploqaz1}
as follows.
\begin{align}
  \widehat{\kappa}^{d,\wb}(\r)
  & \le \frac{c}{\psi - \prod_{i=1}^d{y_i}}\Bigg(\sum_{i=1}^{t}\delta^{s_{i-1}}g(y_i,w_i)
  +\sum_{i=t+1}^d\delta^{s_{i-1}}g(y_i,w_i)\Bigg)\notag\\
  & \le  \frac{c}{\psi - \prod_{i=1}^t{y_i}}\Bigg(\sum_{i=1}^{t}\delta^{s_{i-1}}g(y_i,w_i)
  +0.15 \delta^{s_{t}} K_{w_{t+1}}\sum_{i=t+1}^d\delta^{s_{i-1}-s_t}\alpha^{-l_{w_i}}(2c)^{-w_i}w_i\left(1-2^{-w_i}\right)^{-1}\Bigg)\notag\\
  & = \frac{c}{\psi - \prod_{i=1}^t{y_i}}\Bigg(\sum_{i=1}^{t}\delta^{s_{i-1}}g(y_i,w_i)
  +0.15 \delta^{s_{t}} K_{w_{t+1}}\sum_{i=t+1}^d\delta^{s_{i-1}'}\alpha^{-l_{w_i}}(2c)^{-w_i}w_i\left(1-2^{-w_i}\right)^{-1}\Bigg),\notag\\
  & = \frac{c}{\psi - \prod_{i=1}^t{y_i}}\Bigg(\sum_{i=1}^{t}\delta^{s_{i-1}}g(y_i,w_i)
  +0.15 \delta^{s_{t}} K_{w_{t+1}}\zeta(\wb',d-t)\Bigg),
  \label{eqn:tail}
\end{align}
where
$s_{i-1}'=s_{i-1}-s_t=\sum_{j=t+1}^{i-1} \max(0,k-w_{j}-1)$,
$\wb'=\{w_{t+1},\dots,w_d\}$,
and the function $\zeta(\wb,d)$ is defined in~\eqref{eqn:zeta}.

Intuitively, the term $\zeta(\wb',d-t)$ bounds a tail sum coming
from the last $d-t$   clauses corresponding to the variables
$w_{t+1},\ldots,w_d$. Recall from the statement of
Lemma~\ref{lem:kappabound:d-b2} and the sentences following its
statement that the $w_j$'s are in increasing order. Preferably, we
want to choose the index~$t$ to split the sum in \eqref{eqn:tail}
so that $w_{t} \le 5$ and $w_{t+1}\ge 6$. However, we also do not
want too many terms to be in the first sum (since each of these
will cause us work), so we insist that $t\leq 8$. When $t=8$, we
will use the bound $w_{t+1}\geq 2$. If $t<8$ we will be able to
use the stronger bound $w_{t+1} \geq 6$.

The  first step in the proof of Lemma~\ref{lem:kappabound:d-b2} is to find a good way to control $g(y,w)$ for  $2\le w\le 5$.
To this end, we use a piecewise linear function to upper bound $g(y,w)$.
In particular, in Section~\ref{sec:gh1yup}, we verify using Mathematica's \textsc{Resolve} function that, for $w=2,3,4,5$ and
$y\in [(1-2^{-w})^{1/2},1)$,
it holds that
\begin{equation}\label{eq:gh1yup}
  g(y,w) \le h_1(y), \mbox{ where }h_1(y):=\min\{0.2279, -1.5 y + 1.6276, -8 y + 8.052\}.
\end{equation}

Moreover, let $t_6+1$ be the first index where $w_{t_6+1}\ge 6$.
Therefore $w_i\le 5$ for all $2\le i\le t_6$.
Then applying \eqref{eqn:tail} with $t=t_6$ and Lemma \ref{lem:zeta},
\begin{align}
  \widehat{\kappa}^{d,\wb}(\r) & \le \frac{c}{\psi - \prod_{i=1}^{t_6}{y_i}}
  \left(\sum_{i=1}^{t_6}\delta^{s_{i-1}}g(y_i,w_i)
  + 0.15 \delta^{s_{t_6}} \tau_6 \right)\notag\\
  & \le \frac{c}{\psi - \prod_{i=1}^{t_6}{y_i}}
  \left(\sum_{i=1}^{t_6}\delta^{(i-1)(k-6)}h_1(y_i)
  + 0.15 \delta^{t_6(k-6)} \tau_6 \right),
  \label{eqn:kappahat:h1}
\end{align}
where, for $i\in\{1,\ldots,t_6\}$, $y_i$
is in the range
$[(1-2^{-w_i})^{1/2} ,1)$.
Since each $w_i$ is at least~$2$,
we have $(1-2^{-w_i})^{1/2} \geq \sqrt{3}/2$, so each
 $y_i\in [\sqrt{3}/2,1]$.

The function $h_1(y)$ is non-increasing, and in fact
\begin{align}
  h_1(y)=
  \begin{cases}
    0.2279 & \textnormal{if $\sqrt{3}/2\le y\le Y_0$};\\
    -1.5 y + 1.6276 & \textnormal{if $Y_0\le y\le Y_1$};\\
    -8 y + 8.052 & \textnormal{if $Y_1\le y\le 1$},
  \end{cases}
  \label{eqn:h1(y)}
\end{align}
where $Y_0\approx 0.933133$ and $Y_1\approx 0.988369$ are two constants such that
\begin{equation}\label{eqn:Y0Y1}
\begin{aligned}
  -1.5 Y_0 + 1.6276 & = 0.2279,\\
  -1.5 Y_1 + 1.6276 & = -8 Y_1 + 8.052.
\end{aligned}
\end{equation}
Notice that $h_1(y_i)$ is a constant for $y_i\in
[\sqrt{3}/2,Y_0]$. Moreover, $\frac{c}{\psi - \prod_{i=1}^t{y_i}}$
is an increasing function of $y_i$ for each $i$. Hence, to upper
bound the right-hand side of \eqref{eqn:kappahat:h1}, we may
restrict the $y_i$'s to be in the interval $[Y_0,1]$. More
precisely, let
\begin{align}
  h(y):=
  \begin{cases}
    -1.5 y + 1.6276 & \textnormal{if $Y_0\le y\le Y_1$};\\
    -8 y + 8.052 & \textnormal{if $Y_1\le y\le 1$}.
  \end{cases}\label{eqn:h(y)}
\end{align}

\begin{definition}
\label{def:cfive}
Let   $c_5:=c^{1-6/200}$.
\end{definition}
Observe that $c_5 = \delta^{\Delta(1-6/200)} \ge \delta^{\Delta-6} \ge \delta^{k-6}$ as $k\ge \Delta\ge 200$ and $\delta<1$.
By \eqref{eqn:kappahat:h1} we have that
\begin{align}
  \widehat{\kappa}^{d,\wb}(\r) & \le \frac{c}{\psi - \prod_{i=1}^{t_6}{y_i}}
  \left(\sum_{i=1}^{t_6}\delta^{(i-1)(k-6)}h(y_i) + 0.15 c_5^{t_6} \tau_6 \right)\notag \\
  & \le \frac{c}{\psi - \prod_{i=1}^{t_6}{y_i}}
  \left(\sum_{i=1}^{t_6}c_5^{i-1}h(y_i) + 0.15 c_5^{t_6} \tau_6 \right),
  \label{eqn:t6}
\end{align}
where each $y_i\in [Y_0,1]$.

What is left is to bound the right hand side of \eqref{eqn:t6}.
For $y_i\in[Y_0,1]$, define
\begin{equation} \label{eqn:sigma6}
  \sigma_{t,6}(\y) := \frac{c}{\psi - \prod_{i=1}^t{y_i}}\Bigg(\sum_{i=1}^{t}c_5^{i-1}h(y_i)
  + 0.15 c_5^{t} \tau_6 \Bigg),
\end{equation}
where $\tau_6$ can be found in Table \ref{tab:d-b2:constants}.
The following lemma, which is proved in Section \ref{sec:lem:zeta-sigma},
  bounds $\sigma_{t,6}(\y)$ when $0\le t\le 7$.

\newcommand{\statelemsigmasixbound}{ Let $t$ be an integer such that $0\le t\le 7$.
For any $(y_1,\ldots,y_t)$ where each
    $y_i$ is in the range $[Y_0,1]$,  we have $\sigma_{t,6}(\y)\le 1$.}
\begin{lemma}
  \label{lem:sigma-6-bound}
  \statelemsigmasixbound
\end{lemma}

 Combining~\eqref{eqn:t6}, \eqref{eqn:sigma6} and Lemma~\ref{lem:sigma-6-bound} we get an upper bound on
 $\widehat{\kappa}^{d,\wb}(\r)$
 provided $\wb$ is such that $t_6 \leq 7$.
If $t_6\ge 8$, we will set $t=8$ in \eqref{eqn:tail}.
Similarly to deriving \eqref{eqn:t6}, we get that
\begin{align}
  \widehat{\kappa}^{d,\wb}(\r)  & \le \frac{c}{\psi - \prod_{i=1}^{8}{y_i}}
  \left(\sum_{i=1}^{8}c_5^{i-1}h(y_i) + 0.15 c_5^{8} K_2 \tau_2 \right),
  \label{eqn:t2}
\end{align}
where each
$y_i \in [Y_0,1]$
and $K_2$ and $\tau_2$ can be found in Table \ref{tab:d-b2:constants}.
Similarly to \eqref{eqn:sigma6}, define
$\sigma_{8,2}(\y)$ to be the right-hand side of~\eqref{eqn:t2}. Namely,
\begin{equation}\label{eqn:sigma2}
  \sigma_{8,2}(\y) := \frac{c}{\psi - \prod_{i=1}^8{y_i}}\Bigg(\sum_{i=1}^{8}c_5^{i-1}h(y_i) + 0.15
  c_5^{8}
  K_2 \tau_2 \Bigg).
\end{equation}
The next lemma bounds $\sigma_{8,2}(\y)$ and is proved in Section \ref{sec:lem:zeta-sigma}.
\newcommand{\statelemsigmatwobound}{For any
$(y_1,\ldots,y_8)$ where each
$y_i$ is in the range $[Y_0,1]$, we have $\sigma_{8,2}(\y)\le 1$.}
\begin{lemma}
  \label{lem:sigma-2-bound}
  \statelemsigmatwobound
\end{lemma}

We can now prove Lemma~\ref{lem:kappabound:d-b2}, which we restate here for convenience.
 {\renewcommand{\thetheorem}{\ref{lem:kappabound:d-b2}}
\begin{lemma}
\statelemkappabounddb
\end{lemma}
\addtocounter{theorem}{-1}
}
\begin{proof}
  Consider $t_6$.
  \begin{enumerate}
    \item If $t_6\le 7$, then by \eqref{eqn:t6}, \eqref{eqn:sigma6}, and Lemma \ref{lem:sigma-6-bound},
      $\widehat{\kappa}^{d,\wb}(\r)  \le \sigma_{t_6,6}(\y) \le 1$.
    \item Otherwise $t_6\ge 8$. By \eqref{eqn:t2}, \eqref{eqn:sigma2} and Lemma \ref{lem:sigma-2-bound},
      $\widehat{\kappa}^{d,\wb}(\r)  \le \sigma_{8,2}(\y) \le 1$. \qedhere
  \end{enumerate}

  The case $d=\Delta$ has the same proof.
  Like the proof of Lemma~\ref{lem:weak},
  the only difference in the $d=\Delta$ case is 
that in  
equations~\eqref{eqn:b2=0} 
and~\eqref{eqn:y_i}
we replace $\delta^k$ by $\delta^{k-1}$,
losing a factor of $1/\delta$ (the upper bound is valid using the
same argument as before, now using inequality
$s_i+i-d-(w_i-1)(\Delta-1)\geq s_{i-1}+(k-1)-w_i\Delta$).
\end{proof}

\subsection{Remaining Proofs}\label{sec:lem:zeta-sigma}

In this section we provide technical details of   Lemma \ref{lem:sigma-6-bound}
and Lemma \ref{lem:sigma-2-bound}, which we restate for convenience.

{\renewcommand{\thetheorem}{\ref{lem:sigma-6-bound}}
\begin{lemma}
\statelemsigmasixbound
\end{lemma}
\addtocounter{theorem}{-1}
}

\begin{proof}
  Recall that
    \begin{equation*} \tag{\ref{eqn:sigma6}}
  \sigma_{t,6}(\y) := \frac{c}{\psi - \prod_{i=1}^t{y_i}}\Bigg(\sum_{i=1}^{t}c_5^{i-1}h(y_i)
  + 0.15 c_5^{t} \tau_6 \Bigg),
\end{equation*}
where $c_5=c^{1-6/200}$, $c=0.7$ and $\tau_6=2.7805$ are as in Table~\ref{tab:d-b2:constants} and the function $h(y)$ is given by \eqref{eqn:h(y)}.
For $t=0$, we have
\begin{equation}\label{eq:sigma06ya}
\sigma_{0,6}(\y)=\frac{0.15c\tau_6}{\psi-1}\approx 0.973044<1,
\end{equation}
see Section~\ref{sec:sigma-6-bound} for the calculation. For $t=1$, we have
\begin{equation}\label{eq:sigma06yb}
\sigma_{1,6}(\y)=\frac{c}{\psi - y_1}\big(h(y_1) + 0.15 c_5 \tau_6 \big)\leq 1,
\end{equation}
see Section~\ref{sec:sigma-6-bound} for the verification using Mathematica's \textsc{Resolve} function. Thus, we may assume that $t\ge 2$ henceforth.

  For the sake of contradiction, suppose that there exists $\y'$ such that $\sigma_{t,6}(\y')> 1$ for some $2\le t\le 7$.
  We will gradually adjust the variables~$y_i'$ without decreasing $\sigma_{t,6}(\y')$ until there is only one variable left,
  in which case we will be able to exclude the possibility that $\sigma_{t,6}(\y')>1$.

We first observe that the partial derivative of $\sigma_{t,6}(\y)$ with respect to $y_i$ is
  \begin{equation}\label{eq:partialt6a}
    \frac{\partial \sigma_{t,6}(\y)}{\partial y_i} = \frac{1}{\psi - \prod_{j=1}^t{y_j}}\Bigg( \frac{\sigma_{t,6}(\y)\prod_{j=1}^t{y_j}}{y_i} -  1.5 c c_5^{i-1}\Bigg),
  \end{equation}
  if $Y_0\le y_i
  \leq Y_1$,
  and
  \begin{equation}\label{eq:partialt6b}
    \frac{\partial \sigma_{t,6}(\y)}{\partial y_i} = \frac{1}{\psi - \prod_{j=1}^t{y_j}}\Bigg( \frac{\sigma_{t,6}(\y)\prod_{j=1}^t{y_j}}{y_i} -  8 c c_5^{i-1}\Bigg),
  \end{equation}
  if $Y_1
  \leq y_i \le 1$.

  Suppose that there exists an index $3\leq i\leq t$ such that $y_i'\le Y_1$.
  Using our initial assumption that $\sigma_{t,6}(\y')>1$, we then have (from \eqref{eq:partialt6a} and $y_j'\geq Y_0$)  that
  \begin{equation}\label{eq:sigmat6a}
    \frac{\partial \sigma_{t,6}(\y)}{\partial y_i}\bigg|_{\y=\y'}
     > \frac{1}{\psi - \prod_{j=1}^t{y_j'}}\Bigg( Y_0^{t-1} -  1.5 c c_5^{i-1}\Bigg) >0,
  \end{equation}
  for any $2\le t\le 7$ and $3\le i\le t$, see Section~\ref{sec:sigma-6-bound} for the verification of the last inequality.
  Hence $\sigma_{t,6}(\y)$ is increasing in this $y_i'$ and we may thus assume $y_i'\ge Y_1$ for all $3\le i\le t$.

  Suppose that $y_2'\le Y_1$. Then, using that $y_i'\ge Y_1$ for all $3\le i\le t$ and $y_1'\geq Y_0$, together with our  assumption that $\sigma_{t,6}(\y')>1$, we have (from \eqref{eq:partialt6a}) that
  \begin{equation}\label{eq:sigmat6b}
    \frac{\partial \sigma_{t,6}(\y)}{\partial y_2}\bigg|_{\y=\y'}
    > \frac{1}{\psi - \prod_{j=1}^t{y_j'}}\Bigg( Y_0 Y_1^{t-2} -  1.5 c c_5\Bigg) >0
  \end{equation}
     for any $2\le t\le 7$, see Section~\ref{sec:sigma-6-bound} for the verification of the last inequality. Arguing as before, we may therefore assume $y_i'\ge Y_1$ for all $2\le i\le t$.

Suppose that there exists an index $2\leq i\leq \min\{5,t\}$ such that $y_i'> Y_1$. Since $y_i'\ge Y_1$ for all $2\le i\le t$ and $y_1'\geq Y_0$, we obtain
(from~\eqref{eqn:sigma6})
the following upper bound on $\sigma_{t,6}(\y')$ (using also the fact that $h$ is decreasing):
  \begin{equation}\label{eqn:sigma6-upper}
    \sigma_{t,6}(\y') \le \frac{c}{\psi - 1}\left(h(Y_0)+\sum_{i=2}^{t}c_5^{i-1}h(Y_1) + 0.15 c_5^{t} \tau_6
    \right).
  \end{equation}
  Plugging \eqref{eqn:sigma6-upper} into \eqref{eq:partialt6b} and using the fact that the $y_j$'s are  at most~$1$, we obtain that
  \begin{equation}\label{eq:sigmat6c}
    \frac{\partial \sigma_{t,6}(\y)}{\partial y_i}\bigg|_{\y=\y'}
     \leq
     \frac{1}{\psi - \prod_{j=1}^t{y_j'}}\Bigg( \frac{c}{\psi-1}\left(h(Y_0)+c_5 h(Y_1)\frac{1-c_5^{t-1}}{1-c_5} + 0.15 c_5^{t} \tau_6\right) -  8 c c_5^{i-1}\Bigg) <0,
  \end{equation}
    see Section~\ref{sec:sigma-6-bound} for the verification of the last inequality.
  Therefore $\sigma_{t,6}(\y)$ is decreasing in $y_i'$ for any $2\le i\le \min\{5,t\}$
  and we may therefore assume that $y_i'=Y_1$ for any $2\le i\le \min\{5,t\}$.

  For $2\leq t\leq 5$, using the fact that $y_2'=\hdots=y_t'=Y_1$, we thus have that
  \begin{equation}\label{eq:falsifyt6a}
    \sigma_{t,6}(\y') = \frac{c}{\psi - y_1' Y_1^{t-1}}\Bigg(h(y_1') + c_5 h(Y_1)\frac{1-c_5^{t-1}}{1-c_5} + 0.15 c_5^{t} \tau_6 \Bigg) > 1,
  \end{equation}
    which is false for all $y_1'\in[Y_0,1]$, see Section~\ref{sec:sigma-6-bound} for the proof using Mathematica's \textsc{Resolve} function. Similarly, for $t=6$, using that $y_2'=\hdots=y_5'=Y_1$, $Y_1\leq y_6'\leq 1$ and that $h$ is decreasing, we have that
  \begin{equation}\label{eq:falsifyt6b}
    1<\sigma_{6,6}(\y') \le \frac{c}{\psi - y_1' Y_1^{4}}\Bigg(h(y_1') + c_5 h(Y_1)\frac{1-c_5^{5}}{1-c_5}+ 0.15 c_5^{6} \tau_6 \Bigg),
  \end{equation}
  which is false for all $y_1'\in[Y_0,1]$, see Section~\ref{sec:sigma-6-bound} for the proof using Mathematica's \textsc{Resolve} function.    Finally, for $t = 7$, we have that $y_2'=y_3'=y_4'=y_5'=Y_1$. Using this and $Y_1\leq y_6',y_7'\leq 1$, we obtain that
\begin{equation}\label{eq:falsifyt6c}
    1<\sigma_{7,6}(\y') \le \frac{c}{\psi - y_1' Y_1^{4}}\Bigg(h(y_1') + c_5 h(Y_1)\frac{1-c_5^{6}}{1-c_5}+ 0.15 c_5^{7} \tau_6 \Bigg),
  \end{equation}
which is false for all $y_1'\in[Y_0,1]$, see Section~\ref{sec:sigma-6-bound} for the proof using Mathematica's \textsc{Resolve} function.

Thus, for all $2\leq t\leq 7$, the assumption that there exists $\y'$ such that $\sigma_{t,6}(\y')>1$ has lead to a contradiction, thus completing the proof of Lemma~\ref{lem:sigma-6-bound} for all $0\leq t\leq 7$.
\end{proof}

The proof of Lemma \ref{lem:sigma-2-bound} is very similar.
{\renewcommand{\thetheorem}{\ref{lem:sigma-2-bound}}
\begin{lemma}
\statelemsigmatwobound
\end{lemma}
\addtocounter{theorem}{-1}
}

\begin{proof}
Recall that
\begin{equation}\tag{\ref{eqn:sigma2}}
  \sigma_{8,2}(\y) := \frac{c}{\psi - \prod_{i=1}^8{y_i}}\Bigg(\sum_{i=1}^{8}c_5^{i-1}h(y_i) + 0.15 c_5^{t}  K_2 \tau_2 \Bigg),
\end{equation}
where $c_5=c^{1-6/200}$, $c=0.7$,
$K_2=1.11614$
and $\tau_2=2.7805$ are as in Table~\ref{tab:d-b2:constants} and the function $h(y)$ is given by \eqref{eqn:h(y)}.
For the sake of contradiction, suppose  that there exists $\y'$ such that $\sigma_{8,2}(\y')> 1$.
  We will gradually adjust the variables~$y_i'$ without decreasing $\sigma_{8,2}(\y')$ until there is only one variable left,  in which case we can directly verify that $\sigma_{8,2}(\y')>1$ is impossible.

Identically to Lemma~\ref{lem:sigma-6-bound},
  the partial derivative of $\sigma_{8,2}(\y)$ with respect to $y_i$ is
   \begin{equation*}\tag{\ref{eq:partialt6a}}
    \frac{\partial \sigma_{8,2}(\y)}{\partial y_i} = \frac{1}{\psi - \prod_{j=1}^8{y_j}}\Bigg( \frac{\sigma_{8,2}(\y)\prod_{j=1}^8{y_j}}{y_i} -  1.5 c c_5^{i-1}\Bigg),
  \end{equation*}
  if $Y_0\le y_i    \leq Y_1$,
  and
  \begin{equation*}\tag{\ref{eq:partialt6b}}
    \frac{\partial \sigma_{8,2}(\y)}{\partial y_i} = \frac{1}{\psi - \prod_{j=1}^8{y_j}}\Bigg( \frac{\sigma_{8,2}(\y)\prod_{j=1}^8{y_j}}{y_i} -  8 c c_5^{i-1}\Bigg),
  \end{equation*}
  if $Y_1
  \leq y_i \le 1$.  We may thus use the same line of arguments as in Lemma~\ref{lem:sigma-6-bound} to conclude that we may assume that $y_i'\geq Y_1$ for $3\leq i\leq 8$ (by verifying \eqref{eq:sigmat6a} for $t=8$ and $3\leq i\leq 8$), and then bootstrap  that to $y_i'\geq Y_i$ for $2\leq i\leq 8$ (by verifying \eqref{eq:sigmat6b} for $t=8$ and $i=2$), see Section~\ref{sec:sigma-2-bound} for the details of the verification.

We thus obtain the following upper bound for $\sigma_{8,2}(\y')$ (this is an
analogue of \eqref{eqn:sigma6-upper} which is obtained using the fact that $h$ is decreasing):
\begin{equation}\label{eqn:sigma2-upper}
    \sigma_{8,2}(\y') \le \frac{c}{\psi - 1}\left(h(Y_0)+\sum_{i=2}^{8}c_5^{i-1}h(Y_1) + 0.15c_5^{8} K_2\tau_2
    \right).
  \end{equation}
Now suppose that there exists an index $2\leq i \leq 5$ such that $y'_i> Y_1$. We plug~\eqref{eqn:sigma2-upper}
into \eqref{eq:partialt6b} and obtain the following analogue of \eqref{eq:sigmat6c}:
  \begin{equation}\label{eq:rferferfe}
    \frac{\partial \sigma_{8,2}(\y)}{\partial y_i}\bigg|_{\y=\y'}
     \leq \frac{1}{\psi - \prod_{j=1}^8{y_j'}}\left( \frac{c}{\psi-1}\left(h(Y_0)+c_5 h(Y_1)\frac{1-c_5^{7}}{1-c_5} + 0.15 c_5^{8} K_2 \tau_2 \right) - 8 c c_5^{i-1}\right) <0,
  \end{equation}
see Section~\ref{sec:sigma-2-bound} for the details of the verification of the last inequality.
Thus, we may assume that $y_i'=Y_1$ for $2\le i\le 5$.

Then we can bootstrap our bound on $\sigma_{8,2}(\y')$ in \eqref{eqn:sigma2-upper} to
\begin{equation}\label{eqn:sigma2-upper-b}
\sigma_{8,2}(\y') \le\frac{c}{\psi- y_1' Y_1^4}\left(h(y_1')+\sum_{i=2}^{8}c_5^{i-1}h(Y_1) + 0.15c_5^{8} K_2\tau_2 \right),
\end{equation}
which gives that
  \begin{equation}\label{eq:tvbgfdb3ed3}
    \frac{\partial \sigma_{8,2}(\y)}{\partial y_6}\bigg|_{\y=\y'}
     \leq \frac{1}{\psi - \prod_{j=1}^t{y_j'}}\left( \frac{cY_1^4}{\psi- y_1' Y_1^4}
    \left(h(y_1')+c_5 h(Y_1)\frac{1-c_5^{7}}{1-c_5}+ 0.15 c_5^{8} K_2 \tau_2 \right) - 8 c c_5^{5}\right)<0.
  \end{equation}
  where the last inequality holds for all $y_1'\in[Y_0,1]$, see Section~\ref{sec:sigma-2-bound} for the verification using Mathematica's \textsc{Resolve} function. This  implies that we can set $y_6'=Y_1$ as well.

Using $y_2'=\cdots=y_5'=y'_6=Y_1$, $Y_1\leq y_7',y_8'\leq 1$ and the fact that $h$ is decreasing, we obtain that
\begin{equation}\label{eq:falsifyt2a}
    1<\sigma_{8,2}(\y') \le \frac{c}{\psi - y_1' Y_1^{5}}\Bigg(h(y_1') + c_5 h(Y_1)\frac{1-c_5^{7}}{1-c_5}+ 0.15 c_5^{8} K_2 \tau_2  \Bigg),
  \end{equation}
which is false for all $y_1'\in[Y_0,1]$, see Section~\ref{sec:sigma-2-bound} for the proof using Mathematica's \textsc{Resolve} function.

Thus, the assumption that there exists $\y'$ such that $\sigma_{8,2}(\y')> 1$
has led to a contradiction. This completes the proof of the lemma.
\end{proof}

\section{Bounding the decay rate for \texorpdfstring{$k=3$, $\Delta=6$}{k=3, Delta=6}}\label{sec:decayrate}

In this section, we fix $k=3$ and $\Delta=6$.
This section is devoted to proving Lemma~\ref{lem:potentialfunctionstar}, which we restate here for convenience.

{\renewcommand{\thetheorem}{\ref{lem:potentialfunctionstar}}
\begin{lemma}
\statelempotfn\end{lemma}
\addtocounter{theorem}{-1}
}

In the statement of Lemma~\ref{lem:potentialfunctionstar},
$d$ is an integer between~$1$ and~$6$.
The vector $\wb$ is a suitable vector, which is defined in
Definition~\ref{def:suitable}.
This means that the entries $w_1,\ldots,w_d$ are in non-decreasing order, except
that any ``$1$'' entries are left to the end.
The definition of ``suitable'' also includes some global notation which depends implicitly on
$\wb$:  For all positive integers~$\ell$, $b_\ell$ is the number of entries  amongst $w_1,\ldots,w_d$
 which are equal to $\ell-1$.
 Hence, $\sum_{\ell = 2}^{\infty} b_\ell = d$.
Also, $w_1,\ldots, w_{b_3}$ are all equal to $2$, whereas for
$i>b_3$, $w_i$ is either~$1$ or it is at least~$3$.
 We have $k=3$ so $b'_3 = b_3$.
 Finally,
 $s_i = \sum_{t=1}^i \max(0,2-w_t)$.
 Recall  the definition of $\kappa_*^{d,\wb}(\r)$
from~\eqref{eq:kappa-general}, which we have specialised here to $k=3$:
\begin{equation*}
\kappa_*^{d,\wb}(\r) :=
\sum_{i=1}^{d}
\sum_{j=1}^{w_i}  \alpha^{-l_{w_i}}
\delta^{(b_2  + s_{\min(i,d-b_2)} - \max(0,b_3-i) - (j-1)(\Delta-1)\mathbf{1}_{i\leq d-b_2})}
\frac{\phi(F^{d,\wb}(\r))}{\phi(r_{i,j})}
    \left|\frac{\partial F^{d,\wb}}{\partial r_{i,j}}\right|  .
  \end{equation*}

Consider the following definitions, which apply to all suitable $\wb$.
\begin{align}\label{def:rho}
\rho^{\wb, i}(\r) &:=\alpha^{-l_{w_i}}\sum^{w_i}_{j=1}\left(\frac{1}{\delta}\right)^{(j-1)(\Delta-1)}\frac{1}{\phi(r_{i,j})}\left|\frac{\partial F^{d,\wb}}{\partial r_{i,j}}\right|.\\
 \label{def:kappa}
  \kappa^{d,\wb}(\r)&:=\phi(F(\r))\,\delta^{b_2}\, \bigg(\sum^{b_3}_{i=1}\Big(\frac{1}{\delta}\Big)^{b_3-i}\rho^{\wb, i}(\r)+\sum^{d}_{i=b_3+1}\rho^{\wb, i}(\r)\bigg).
\end{align}

We first argue that $ \kappa_*^{d,\wb}(\r)\leq   \kappa^{d,\wb}(\r)$.
To see this, note that  $s_{\min(i,d-b_2)}\geq 0$. Also,
the term $\max(0,b_3-i)$ is $b_3-i$ if $1\leq i \leq b_3$ and is~$0$ if $i>b_3$.
Finally, $\mathbf{1}_{i\leq d-b_2}\leq 1$.

\begin{definition}Let $\delta=9789/10000$.\end{definition}

We can now state a lemma which immediately implies Lemma~\ref{lem:potentialfunctionstar}
since $ \kappa_*^{d,\wb}(\r)\leq   \kappa^{d,\wb}(\r)$.

\newcommand{\statelemnewpotfn}{Let $\Delta=6$ and $k=3$. There is constant $\MM>0$ such that,   for all
$1\leq d \leq \Delta$, all suitable $\wb = w_1,\ldots,w_d$,
and all  $\r$ satisfying $(1/2)^{\Delta-1}\ones\leq \r \leq \ones$,
          it holds that
  \begin{equation*}
    \kappa^{d,\wb}(\r)\leq \begin{cases}1,& \mbox{ when } d\leq \Delta-1,\\ \MM,& \mbox{ when } d=\Delta.\end{cases}
    \end{equation*}
    }
\begin{lemma}\label{lem:potentialfunction}
 \statelemnewpotfn
\end{lemma}

 The rest of the section contains  the proof of Lemma~\ref{lem:potentialfunction}.
 This   is a more involved optimisation problem than
 the one that arose in the proof of Lemma \ref{lem:decay-general}.
Before delving into the details, we  set up some convenient notation and then give a roadmap of the argument.

\subsection{Outline of the proof}\label{sec:outline123}

For convenience, let $F:=F^{d,\wb}$, $\kappa(\r):=\kappa^{d,\wb}(\r)$ which is defined in \eqref{def:kappa}
and for $i\in [d]$, let $\rho_i(\r):=\rho^{\wb, i}(\r)$ which is defined in \eqref{def:rho}.
Our goal is to show that, for all $\r$ such that $\frac{1}{2^{\Delta-1}}\mathbf{1}\leq \r\leq \mathbf{1}$, it holds that $\kappa(\r)\leq 1$ when $d\leq \Delta-1$ and that $\kappa(\r)$ is bounded by a constant when $d=\Delta$.

Here is a rough outline of our analysis.
\begin{enumerate}
\item The first part of the proof will be to bound the quantities $\rho_i(\r)$ appropriately for each $i\in [d]$. Namely, the main goal here will be to replace the $\{r_{i,j}\}_{j\in [w_i]}$ by a suitable quantity $\hat{r}_i$. In fact, for this part of the proof, rather than working with the $r_{i,j}$'s, it will be easier to work with $t_{i,j}=\frac{r_{i,j}}{1+r_{i,j}}$
See Lemma~\ref{lem:100assym}.
\item \label{it:rfvqwerdf} After the first part, we will have reduced significantly the dimensionality of the optimisation problem: from the initial number of $\sum_{i\in [d]} w_i$ variables $\{r_{i,j}\}_{i\in[d],j\in[w_i]}$, we will be left with just $d$ ``representative'' variables
(one for each~$i$).
\item \label{it:rfvqwerdf2} Despite having reduced the number of variables quite a lot, the $w_i$'s so far can be arbitrarily large integers. It will be convenient for us  to restrict the range of the $w_i$'s. Using a rather crude argument, we will be able to restrict our attention to $i$'s such that $1\leq w_i\leq 5$.
Intuitively, the reason is that
large $w_i$'s make $\kappa$ smaller. We quantify this effect  in an appropriate way for our analysis.
(See Lemma~\ref{lem:thnmi}.)
\item The next step is a further reduction of the number of variables. In particular, recall from Item~\ref{it:rfvqwerdf} that we have reduced to the case where the number of variables is $d$
(one for each~$i$)
and, from Item~\ref{it:rfvqwerdf2}, for each
$i\in [d]$ it holds that $1\leq w_i\leq 5$. We will further group together  these variables according to
their values.
That is,
for integer $1\leq w\leq 5$ we will be able to use a single variable (indexed by $w$) to capture the aggregate contribution of
the variables~$w_i$
with $w_i=w$. (See Equation \eqref{eq:q1} and
the inequality in Lemma~\ref{lem:assym2}
  and the relevant Lemmas~\ref{lem:concav2} and~\ref{lem:assym2}.)
\item At this point, we are able to do the final steps of the optimisation. The most important case turns out to be when all of the $w_i$'s are either 1 or 2 (Lemma~\ref{lem:bootphase1}). In this case, we obtain quite sharp bounds for
$\kappa$
for each value of~$d$
(i.e, $d=1,2,\hdots,\Delta-1=5$). This facilitates the application of cruder arguments to handle the cases where there exist $i\in [d]$ with $w_i\neq 1,2$ (Lemmas~\ref{lem:bootphase2},~\ref{lem:bootphase3} and~\ref{lem:bootphase4}).
\end{enumerate}

\subsection{The details of the argument}
In this section, we expand in detail the outline of Section~\ref{sec:outline123} and give the necessary technical ingredients needed to complete the proof of Lemma~\ref{lem:potentialfunction}.  Later subsections   contain the left-over technical proofs which would significantly interrupt the flow.

The first part of the proof will be to bound the quantities
$\rho_i(\r)$ appropriately. For $i\in [d]$ and $j\in [w_i]$, we
have (by differentiating $\ln F(\r)$ as in the proof of
Lemma~\ref{lem:decay-general}):
\[\frac{\partial F}{\partial r_{i,j}}=-F(\r)\frac{\prod^{w_i}_{j'=1}\frac{r_{i,j'}}{1+r_{i,j'}}}{1-\prod^{w_i}_{j'=1}\frac{r_{i,j'}}{1+r_{i,j'}}}\frac{1}{r_{i,j}(1+r_{i,j})}.\]

Let $g(z):=\frac{1}{\phi(z)}\frac{1}{z(1+z)}=\frac{\psi - z^\chi}{1+z}$. For $i\in [d]$, the quantity $\rho_i(\r)$ can then be written as
\begin{equation}\label{new:rho}
\rho_i(\r):=F(\r)\frac{\prod^{w_i}_{j=1}\frac{r_{i,j}}{1+r_{i,j}}}{1-\prod^{w_i}_{j=1}\frac{r_{i,j}}{1+r_{i,j}}} \alpha^{-l_{w_i}}\sum^{w_i}_{j=1}\Big(\frac{1}{\delta}\Big)^{(j-1)(\Delta-1)}g(r_{i,j}).
\end{equation}

Let $t_{i,j}:= \frac{r_{i,j}}{1+r_{i,j}}$ and let $\hat{t}_i$ be the geometric mean of the $t_{i,j}$'s, i.e., $(\hat{t}_i)^{w_i}:=\prod^{w_i}_{j=1}t_{i,j}=\prod^{w_i}_{j=1}\frac{r_{i,j}}{1+r_{i,j}}$. As we shall see soon, $\hat{t}_i$ will be used to capture the ``aggregate" effect of~$w_i$.
Let $\t$ be the vector whose entries are given by $t_{i,j}$ with $i\in[d]$ and $j\in [w_i]$. Note that $\frac{1}{2^{\Delta-1}+1}\ones\leq\t\leq (1/2)\ones$.

We will view the quantities $\kappa(\r)$ and $\rho_i(\r)$ as a function of $\t$. For that, it will be convenient to  consider the function
\begin{equation}\label{eq:functionh}
h(t):=g\left(\frac{t}{1-t}\right)=(1-t)\left[\psi-\left(\frac{t}{1-t}\right)^{\chi}\right]
\end{equation}
for $t\in[1/(2^{\Delta-1}+1),1/2]$. With this preprocessing, for $i\in [d]$, the quantities $\rho_i(\r),\ \kappa(\r)$ as a function of $\t$ become
\begin{equation}\label{new:rho1}
\begin{gathered}
\rho_i(\t)=\overline{F}(\t)\frac{(\hat{t}_i)^{w_i}}{1-(\hat{t}_i)^{w_i}}\alpha^{-l_{w_i}}\sum^{w_i}_{j=1}\Big(\frac{1}{\delta}\Big)^{(j-1)(\Delta-1)}h(t_{i,j}),\\
\bar{\kappa}(\t)=\phi\big(\overline{F}(\t)\big)\,\delta^{b_2}\, \bigg(\sum^{b_3}_{i=1}\Big(\frac{1}{\delta}\Big)^{b_3-i}\rho_{i}(\t)+\sum^{d}_{i=b_3+1}\rho_{i}(\t)\bigg).
\end{gathered}
\end{equation}
where
\begin{equation*}
\overline{F}(\t):=\prod^{d}_{i=1}\Big(1-\prod^{w_i}_{j=1}t_{i,j}\Big).
\end{equation*}
After this preliminary step, for $i\in [d]$, we will now pursue the task of substituting the variables $t_{i,j}$ with $j\in [w_i]$ with a single variable $\hat{t}_i$. Let $\that=\{\hat{t_i}\}_{i=1,\hdots,d}$ and note that $\frac{1}{2^{\Delta-1}+1}\ones\leq\that\leq (1/2)\ones$. As a starting point, observe that
\begin{equation}\label{eq:convertconvert}
\overline{F}(\t)=\widehat{F}(\that), \mbox{ where } \widehat{F}(\that)=\prod^{d}_{i=1}\big(1-(\hat{t}_{i})^{w_i}\big).
\end{equation}
Recall that
$\Delta=6$, $\delta=9789/10000$, $\chi=1/2$, $\psi=13/10$.
The following technical lemma, proved in Section~\ref{sec:assym100}, will be crucial in reducing the number of the variables $t_{i,j}$.

\newcommand{\statelemassym}{
Define the following constants.
\begin{equation*}
\begin{gathered}
K^{(1)}_\delta=1,\quad K^{(2)}_\delta=1069/1000,\quad K^{(3)}_\delta=1160/1000,\quad
K^{(4)}_\delta=1225/1000,\text{ and}\\
K^{(w)}_\delta=\Big(\frac{1}{\delta}\Big)^{(w-1)(\Delta-1)} \mbox{ for } w\geq 5.
\end{gathered}
\end{equation*}
Then for all positive integers~$w$, the following inequality holds for all $t_1,\hdots,t_w\in[0,1/2]$:
\begin{equation*}
\sum^{w}_{j=1}\frac{1}{\delta^{(j-1)(\Delta-1)}}h(t_j)\leq w\, K^{(w)}_\delta\, h(t),
\end{equation*}
where $t$ is the geometric mean of the $t_i$'s, i.e., $t=(t_1\cdots t_w)^{1/w}$.
}

\begin{lemma}\label{lem:100assym}
 \statelemassym
\end{lemma}
 Applying Lemma~\ref{lem:100assym} to the quantity $\rho_i(\t)$ in \eqref{new:rho1} yields that for all $i\in[d]$ it holds that
\begin{equation}\label{eq:rhoi1t}
\rho_i(\t)\leq \rho^{(1)}_i\big(\that\big), \mbox{ where } \rho^{(1)}_i\big(\that\big)=\widehat{F}\big(\that\big)w_i\, K^{(w_i)}_\delta\alpha^{-l_{w_i}}\frac{(\hat{t}_i)^{w_i}}{1-(\hat{t}_i)^{w_i}}h(\hat{t}_i),
\end{equation}
where $\widehat{F}(\that)$ is given by \eqref{eq:convertconvert}.

The next part of the proof will be to bound the contribution of
$w_i$'s
with $w_i\geq 6$ by small quantities so that  we will eliminate those $i$ with $w_i\geq 6$ (and hence the respective variables $\hat{t}_i$) from consideration. This will be accomplished by the following lemma (proved in Section~\ref{sec:thnmi}).

\newcommand{\statelemthnmi}{
Let $M= 25/1000$.
Recall that $\widehat{F}(\that)$ is given by   \eqref{eq:convertconvert}.
Let $i$ be such that $w_i\geq 6$. Then, for all $\that$ such that $\mathbf{0}\leq\that\leq (1/2)\ones$, it holds that
\begin{equation*}
\rho^{(1)}_i\big(\that\big)\leq \frac{1}{\alpha} \widehat{F}(\that) M.
\end{equation*}
}
\begin{lemma}\label{lem:thnmi}
\statelemthnmi
\end{lemma}

Recall
from the beginning of Section~\ref{sec:decayrate} (based on Definition~\ref{def:suitable}) that
$b_\ell$ is the number of entries amongst $w_1,\ldots,w_d$ which are equal to $\ell-1$.
 Using Lemma~\ref{lem:thnmi} we will now be able to eliminate those $i$ such that $w_i\geq 6$.
 In order to do that easily, we will first re-order the entries in $\wb$.
 Let  $B=b_2+b_3+b_4+b_5+b_6$ and note that $0\leq B\leq d$.
 Note that,
 in the context of \eqref{new:rho1}, the ordering of  the $w_i$'s
 with $w_i\neq 2$ does not matter as long as we maintain
 the invariant that their index $i$ satisfies $i\geq b_3+1$.
 Thus, from now on, without loss of generality, we will assume that $w_i\geq 6$ implies that $i\geq B+1$.
 That is, $w_i$'s with $w_i\geq 6$ have indices larger than $w_i$'s with $w_i \leq 5$.

Using \eqref{new:rho1}, \eqref{eq:rhoi1t},
Lemma~\ref{lem:thnmi}
and $l_1=\cdots=l_5=1$, we  thus obtain
\begin{equation}\label{eq:kappakappa1}
\bar{\kappa}(\t)\leq \frac{1}{\alpha}\kappa^{(1)}\big(\that\big),
\end{equation}
where
\begin{equation}\label{eq:kappa1defdef}
\begin{aligned}
\kappa^{(1)}&(\that):=\phi\big(\widehat{F}\big(\that\big)\big)\widehat{F}\big(\that\big)\delta^{b_2}\\
&\bigg(2 K^{(2)}_\delta\sum^{b_3}_{i=1} \left(\frac{1}{\delta}\right)^{b_3-i}\frac{(\hat{t}_i)^{2}}{1-(\hat{t}_i)^{2}}h(\hat{t}_i)+\sum^{B}_{i=b_3+1}w_i\, K^{(w_i)}_\delta\frac{(\hat{t}_i)^{w_i}}{1-(\hat{t}_i)^{w_i}}h(\hat{t}_i)+(d-B)M \bigg),
\end{aligned}
\end{equation}
where $\widehat{F}(\that)$ is given by   \eqref{eq:convertconvert} and $M=25/1000$ (cf. Lemma~\ref{lem:thnmi}).

To complete the program of  eliminating  those variables $\hat{t}_i$ where $i$ is such that $w_i\geq 6$, observe that $z\phi(z)=1/(\psi-z^\chi)$, so
\begin{align}
\phi\big(\widehat{F}\big(\that\big)\big)\widehat{F}\big(\that\big)&=\frac{1}{\psi-\big(\widehat{F}\big(\that\big)\big)^{\chi}}=\frac{1}{\psi-\prod^{d}_{i=1}\big(1-(\hat{t}_{i})^{w_i}\big)^{\chi}}\leq \frac{1}{\psi-\prod^{B}_{i=1}\big(1-(\hat{t}_{i})^{w_i}\big)^{\chi}}.\label{eq:interinterinter}
\end{align}

Using \eqref{eq:interinterinter}, we thus obtain that
\begin{equation}\label{eq:kappa1kappa2}
\kappa^{(1)}(\that)\leq \kappa^{(2)}(\that),
\end{equation}
where
\begin{equation}\label{eq:kappa2defdef}
\kappa^{(2)}(\that)= \delta^{b_2}\frac{2 K^{(2)}_\delta\sum^{b_3}_{i=1} \left(\frac{1}{\delta}\right)^{b_3-i}\frac{(\hat{t}_i)^{2}}{1-(\hat{t}_i)^{2}}h(\hat{t}_i)+\sum^{B}_{i=b_3+1}w_i\, K^{(w_i)}_\delta\frac{(\hat{t}_i)^{w_i}}{1-(\hat{t}_i)^{w_i}}h(\hat{t}_i)+(d-B)M}{\psi-\prod^{B}_{i=1}\big(1-(\hat{t}_{i})^{w_i}\big)^{\chi}}.
\end{equation}

The following quantity $\kappa^{(3)}(\that)$ is similar to $\kappa^{(2)}(\that)$:
\begin{equation}\label{eq:kappa3defdef}
\kappa^{(3)}(\that)= \delta^{b_2}\frac{2 K^{(2)}_\delta\sum^{b_3}_{i=1} \left(\frac{1}{\delta}\right)^{b_3-i}\frac{(\hat{t}_i)^{2}}{1-(\hat{t}_i)^{2}}h(\hat{t}_i)+\sum^{B}_{i=b_3+1}w_i\, K^{(w_i)}_\delta\frac{(\hat{t}_i)^{w_i}}{1-(\hat{t}_i)^{w_i}}h(\hat{t}_i)}{\psi-\prod^{B}_{i=1}\big(1-(\hat{t}_{i})^{w_i}\big)^{\chi}}.
\end{equation}
The only difference between $\kappa^{(2)}(\that)$ and $\kappa^{(3)}(\that)$ is that the term $(d-B)M$ is not present in the numerator of the latter. We therefore have
\begin{equation}\label{eq:kappa2kappa3}
\kappa^{(2)}(\that)= \kappa^{(3)}(\that)+\frac{\delta^{b_2}(d-B)M}{\psi-\prod^{B}_{i=1}\big(1-(\hat{t}_{i})^{w_i}\big)^{\chi}}\leq \kappa^{(3)}(\that)+\frac{(d-B)M}{\psi-1},
\end{equation}
where the last inequality follows from $b_2\geq 0$, $\delta\in(0,1]$ and the fact that the $\hat{t}_i$'s are positive.

The following lemma, proved later in this section, will allow us to conclude Lemma~\ref{lem:potentialfunction}.
\begin{lemma}\label{lem:proofBbounds}
Let $\Delta=6$ and $B$ be a non-negative integer less than or equal to $\Delta-1=5$. Recall that $\alpha=1-10^{-4}$. There exists a constant $\epsilon_B\leq \alpha$ such that for all non-negative integers $b_2,b_3,b_4,b_5,b_6$ with $b_2+b_3+b_4+b_5+b_6=B$, it holds that
$\kappa^{(3)}(\that)\leq \epsilon_B$.

In particular, we will show that
\begin{equation}\label{eq:boundsonepsilons}
\epsilon_0=0,\quad \epsilon_1=6/10, \quad \epsilon_2=7/10, \quad \epsilon_3=83/100,\quad  \epsilon_4=91/100, \quad \epsilon_5=\alpha=1-10^{-4}.
\end{equation}
\end{lemma}

Assuming Lemma~\ref{lem:proofBbounds} for the moment, we give the proof of Lemma~\ref{lem:potentialfunction}, which we restate here for convenience.
{
\renewcommand{\thetheorem}{\ref{lem:potentialfunction}}
\begin{lemma}
\statelemnewpotfn
 \end{lemma}
\addtocounter{theorem}{-1}
}
\begin{proof}[Proof of Lemma~\ref{lem:potentialfunction}]
We first derive the bound $\kappa(\r)\leq 1$ when $d\leq \Delta-1=5$. Recall that the quantity $\kappa(\r)$ (as given in \eqref{def:kappa}) is equal to the quantity $\bar{\kappa}(\t)$ (as given in \eqref{new:rho1}). Also, we have shown that
\begin{equation*}\tag{\ref{eq:kappakappa1}}
\bar{\kappa}(\t)\leq \frac{1}{\alpha}\kappa^{(1)}\big(\that\big),
\end{equation*}
where $\kappa^{(1)}\big(\that\big)$ is as in \eqref{eq:kappa1defdef}. We have also shown that
\begin{equation*}\tag{\ref{eq:kappa1kappa2}}
\kappa^{(1)}\big(\that\big)\leq \kappa^{(2)}\big(\that\big),
\end{equation*}
where $\kappa^{(2)}\big(\that\big)$ is as in \eqref{eq:kappa2defdef}. Moreover, we showed that
\begin{equation*}\tag{\ref{eq:kappa2kappa3}}
\kappa^{(2)}(\that)\leq  \kappa^{(3)}(\that)+\frac{(d-B)M}{\psi-1},
\end{equation*}
where $\kappa^{(3)}(\that)$ is as in \eqref{eq:kappa3defdef}, $B$ is a non-negative integer less than or equal to $\Delta-1=5$ satisfying $B=b_2+b_3+b_4+b_5+b_6$ and $M=25/1000$ is as in Lemma~\ref{lem:thnmi}. Lastly, by Lemma~\ref{lem:proofBbounds}, we have
\begin{equation*}
\kappa^{(3)}(\that)\leq \epsilon_B
\end{equation*}
where the constants $\epsilon_B$ are as in \eqref{eq:boundsonepsilons}. Combining all the above we obtain that
\begin{equation}\label{eq:vghb123}
\kappa(\r)\leq \frac{1}{\alpha}\Big(\epsilon_B+\frac{(d-B)M}{\psi-1}\Big).
\end{equation}
It is a matter of numerical calculations to check that
\begin{equation}\label{eq:numerical123}
\frac{1}{\alpha}\Big(\epsilon_B+\frac{(d-B)M}{\psi-1}\Big)\leq 1
\end{equation}
for all $B=0,1,\hdots,5$ and $d= \Delta-1=5$, see Section~\ref{sec:potentialfunction} for the explicit calculations. This completes the proof of the lemma for $d\leq \Delta-1=5$.

We next consider the case $d=\Delta$, i.e., we show that there exists a constant $\MM>0$ such that $\kappa(\r)\leq \MM$. This will follow by continuity arguments. More precisely, first note that inequalities \eqref{eq:kappakappa1}, \eqref{eq:kappa1kappa2} and \eqref{eq:kappa2kappa3} still hold in the case where $d=\Delta$, with the minor modification that in \eqref{eq:kappa2kappa3}, the integer $B$ (which, recall,
is equal to $b_2 + b_3 + b_4 + b_5 + b_6$)
can be as large as (but not bigger than) $\Delta$. Let us fix $B$ to be a non-negative integer
which is at most $\Delta=6$.  Observe that there are finitely many possibilities for the non-negative integers $b_2,b_3,b_4,b_5,b_6$ such that $b_2+b_3+b_4+b_5+b_6=B$. For each such choice, the quantity $\kappa^{(3)}(\that)$ is a continuous function of the (finitely many) variables $\hat{t}_1,\hdots,\hat{t}_B$ and hence is bounded above by an absolute constant when $\mathbf{0}\leq\that\leq (1/2)\ones$. It follows that for every non-negative integer $B\leq \Delta=6$, there exists an absolute constant $\MM_B>0$ such that $\kappa^{(3)}(\that)\leq \MM_B$. Thus, analogously to \eqref{eq:vghb123}, we obtain the bound
\begin{equation}\label{eq:vghb124}
\kappa(\r)\leq \max_{B=0,\hdots,6}\left\{\frac{1}{\alpha}\Big(\MM_B+\frac{(\Delta-B)M}{\psi-1}\Big)\right\}=:\MM.
\end{equation}
Note that $\MM$, as defined in \eqref{eq:vghb124}, is a constant. The desired bound on $\kappa(\r)$ when $d=\Delta$ follows.

This concludes the proof of Lemma~\ref{lem:potentialfunction}.
\end{proof}

The remainder of this section will focus on the proof of Lemma~\ref{lem:proofBbounds}. We begin our considerations by reducing the number of variables. We first need the following transformation. Namely, for $w=1,2,\hdots$ and all $i\in[B]$ such that $w_i=w$, we set $y_i:=1-t_i^w$ (note that $y_i\in[1-(1/2)^w, 1-1/(2^{\Delta-1}+1)^w]$). For $w=1,2,\hdots$, consider the functions $g_w(y)$ defined for $y\in[1-(1/2)^w, 1-1/(2^{\Delta-1}+1)^w]$.
\begin{align}
g_w(y):=\frac{(1-y)}{y}h\big((1-y)^{1/w}\big).\label{eq:gwfunction}
\end{align}
The quantity $\kappa^{(3)}(\that)$ as a function of $\y=\{y_i\}^{B}_{i=1}$ and the functions $g_w$ then becomes
\begin{equation}\label{eq:firstform}
\begin{aligned}
\kappa^{(4)}(\y)&=\delta^{b_2}\frac{2K^{(2)}_{\delta}\sum^{b_3}_{i=1}\left(\frac{1}{\delta}\right)^{b_3-i}g_2(y_i)+\sum^{B}_{i=b_3+1}w_i\,
K^{(w_i)}_\delta\,
g_{w_i}(y_i)}{\psi-\big(\prod^{B}_{i=1}y_i\big)^{\chi}}.
\end{aligned}
\end{equation}
Let $\hat{y}_w$ be the geometric mean of those $y_i$'s with $w_i=w$ (note that the number of such $i$'s is equal to $b_{w+1}$). More precisely, for $b_{w+1}>0$, let
\[\quad \hat{y}_w:=\Big(\prod_{i; w_i=w}y_i\Big)^{1/b_{w+1}},\]
and when $b_{w+1}=0$, let $\hat{y}_w=1$. Note that
\begin{equation}\label{eq:Fq1q2}
\prod^{B}_{i=1}y_i=\prod^5_{w=1}(\hat{y}_w)^{b_{w+1}}.
\end{equation}
Let $\yhat=\{\hat{y}_i\}_{i=1,\hdots,5}$. Our goal will be to bound $\kappa^{(4)}(\y)$ by a function of $\yhat$.
\begin{lemma}\label{lem:concav2}
For $w=1,2,\hdots,5$, the function $g_w(e^z)$ is a concave function of $z$ in the interval $\big[\ln\big(1-(1/2)^{w}\big),0\big]$.
\end{lemma}
\begin{proof}[Proof]
For $z\in \big[\ln\big(1-(1/2)^w\big),0\big]$, let $f(z)=g_w(e^z)$. Our goal is to show that for $w=1,\hdots,5$, it holds that
\begin{equation}\label{eq:fppz}
f''(z)\leq0 \mbox{ for all } z\in \big[\ln\big(1-(1/2)^w\big),0\big].
\end{equation}
For convenience, we use Mathematica's \textsc{Resolve} function, see Section~\ref{sec:concav2} for details.
\end{proof}

Lemma~\ref{lem:concav2} and Jensen's inequality yield that
\begin{gather}
\sum_{i; w_i=w}g_w(y_i)= \sum_{i; w_i=w}g_w(e^{\ln y_i})=b_{w+1}\, g_w\big(e^{\frac{1}{b_{w+1}}\sum_{i; w_i=w}\ln y_i}\big)\leq b_{w+1}\, g_w\big(\hat{y}_w\big)\label{eq:q1}.
\end{gather}
To bound $\sum^{b_3}_{i=1}\left(\frac{1}{\delta}\right)^{b_3-i}g_2(y_i)$ by a function of $\hat{y}_2$, we will use the following lemma (proved in Section~\ref{sec:assym2}).
\newcommand{\statelassymtwo}{
Let $\Delta=6$ and $b_3$ be a non-negative integer less than or equal to $\Delta-1=5$. There exists a constant $C^{(b_3)}_{\delta}\geq  0$ so that for $y_1,\hdots,y_{b_3}\in [3/4,1-1/(2^d+1)^2]$ it holds that
\begin{equation*}
\sum^{b_3}_{i=1}\left(\frac{1}{\delta}\right)^{b_3-i}g_2(y_i)\leq b_3\, C^{(b_3)}_{\delta}\, g_2(\sqrt[b_3]{y_1\cdots y_{b_3}}).
\end{equation*}
In particular, we will show that the
inequality  holds with $C^{(b_3)}_{\delta}=\frac{1}{\delta^{b_3-1}}\Big(\frac{1-\delta^{b_3p}}{b_3(1-\delta^p)}\Big)^{1/p}$, where $p=27/2$. For $\delta=9789/10000$, we have the following upper bounds on the values of $C^{(b_3)}_{\delta}$:
\begin{equation*}
C^{(0)}_{\delta}=0,\ \,  C^{(1)}_{\delta}=1,\ \,  C^{(2)}_{\delta}=102/100,\ \,  C^{(3)}_{\delta}=103/100,\ \, C^{(4)}_{\delta}=104/100,\ \,C^{(5)}_{\delta}=105/100.
\end{equation*}}
\begin{lemma} \label{lem:assym2}
\statelassymtwo
\end{lemma}

Using \eqref{eq:Fq1q2}, \eqref{eq:q1} and
the inequality in Lemma~\ref{lem:assym2}  we obtain that
\[\kappa^{(4)}(\y)\leq \kappa^{(5)}(\yhat),\]
where
\[\kappa^{(5)}(\yhat):=\delta^{b_2}\cdot \frac{2b_3 K^{(2)}_\delta\, C^{(b_3)}_{\delta}\, g_2(\hat{y}_2)+\sum_{w; w\in [5], w\neq 2} w\, b_{w+1} K^{(w)}_\delta g_w(\hat{y}_w)}{\psi-\prod^5_{w=1}(\hat{y}_w)^{b_{w+1}\chi}},\]
with the values of $K^{(w)}_\delta$ as in Lemma~\ref{lem:100assym}   and the values of $C^{(b_3)}_{\delta}$ as in Lemma~\ref{lem:assym2}.

We next define the following constants $\tau_{b_2,b_3}$ for non-negative integers $b_2,b_3$ satisfying $b_2+b_3\leq \Delta-1=5$ (they are all at most $\alpha$ and we will show that they bound $\kappa^{(5)}(\yhat)$):
\begin{equation}\label{eq:constantstau}
\begin{gathered}
\tau_{0,0}=0,\quad
\tau_{0,1}=\tau_{1,0}=42/100,\\
\tau_{0,2}=54/100,\quad \tau_{1,1}=59/100,\quad \tau_{2,0}=63/100,\\
\tau_{0,3}=72/100, \ \, \tau_{1,2}=74/100,\ \, \tau_{2,1}=76/100,\ \, \tau_{3,0}=79/100,\\
\tau_{0,4}=864/1000,\ \, \tau_{1,3}=868/1000,\ \, \tau_{2,2}=876/1000,\ \, \tau_{3,1}=886/1000,\ \, \tau_{4,0}=901/1000,\\
\tau_{b_2,b_3}=\alpha \mbox{ when } b_2+b_3=5.
\end{gathered}
\end{equation}

\begin{lemma}\label{lem:bootphase1}
Let $b_4=b_5=b_6=0$. For all non-negative integers $b_2,b_3$ such that $b_2+b_3\leq \Delta-1=5$, it holds that $\kappa^{(5)}(\yhat)\leq \tau_{b_2,b_3}$.
\end{lemma}
\begin{proof}[Proof of Lemma~\ref{lem:bootphase1}]
For $b_4=b_5=b_6=0$, the quantity $\kappa^{(5)}(\yhat)$ simplifies into
\begin{equation}\label{eq:kappa4y1y2}
\kappa^{(5)}(\yhat):=\delta^{b_2}\cdot \frac{2b_3 K^{(2)}_\delta\, C^{(b_3)}_{\delta}\, g_2(\hat{y}_2)+b_{2}\, g_1(\hat{y}_1)}{\psi-(\hat{y}_1)^{b_{2}\chi}(\hat{y}_2)^{b_{3}\chi}},
\end{equation}
where we used that $K^{(1)}_\delta=1$ (note that the values of the variables $\hat{y}_3,\hat{y}_4, \hat{y}_5$ do not affect the value of $\kappa^{(5)}$ when $b_4=b_5=b_6=0$).

To bound $\kappa^{(5)}(\yhat)$, we need a couple of transformations. The first one is simple:  set $v_1:=1-\hat{y}_1$ and $v_2:=(1-\hat{y}_2)^{1/2}$ and note that $v_1,v_2\in [1/(2^{\Delta-1}+1),1/2]$. From the definition of the functions $g_1,g_2$ (cf. equation \eqref{eq:gwfunction}), we have that
\begin{equation}\label{eq:g1g2subs}
\begin{aligned}
g_1(\hat{y}_1)&=\frac{1-\hat{y}_1}{\hat{y}_1}h(1-\hat{y}_1)=\frac{v_1}{1-v_1}h(v_1)=v_1\Big(\psi-\big(\frac{v_1}{1-v_1}\big)^{\chi}\Big),\\
g_2(\hat{y}_2)&=\frac{1-\hat{y}_2}{\hat{y}_2}h\big((1-\hat{y}_2)^{1/2}\big)=\frac{v_2^2}{1-v_2^2}h(v_2)=\frac{v_2^2}{1+v_2}\Big(\psi-\big(\frac{v_2}{1-v_2}\big)^{\chi}\Big).
\end{aligned}
\end{equation}
It follows that the quantity $\kappa^{(5)}(\yhat)$ as a function of $v_1,v_2$ becomes
\begin{align*}
\kappa^{(6)}(v_1,v_2):=\delta^{b_2}\frac{2b_3 K^{(2)}_{\delta}\, C^{(b_3)}_{\delta}\, \frac{v^2_2}{1+v_2}\Big(\psi-\big(\frac{v_2}{1-v_2}\big)^{\chi}\Big)+b_2v_1\Big(\psi-\big(\frac{v_1}{1-v_1}\big)^{\chi}\Big)}{\psi-(1-v_1)^{b_2\chi}(1-v_2^2)^{b_3\chi}}.
\end{align*}
We will show that for all $v_1,v_2\in[0,1/2]$, it holds that $\kappa^{(6)}(v_1,v_2)\leq \tau_{b_2,b_3}$. To do this, we will use Mathematica's \textsc{Resolve} function. We first need to rationalize the expressions to keep the computations feasible (this can be achieved for $\chi=1/2$). This brings us to the second (and final) transformation. In particular, set
\[v_1=\frac{4z_1^2}{(1+z_1^2)^2},\quad v_2=\frac{4z_2^2}{1+4z_2^4}\]
for $0\leq z_1\leq \sqrt{2}-1$ and  $0\leq z_2\leq \frac{1}{2}(\sqrt{3}-1)$. Under these transformations, for $\chi=1/2$, we obtain that
\[\Big(\frac{v_1}{1-v_1}\Big)^{1/2}=\frac{2z_1}{1-z_1^2}, \quad (1-v_1)^{1/2}=\frac{1-z_1^2}{1+z_1^2},\quad \Big(\frac{v_2}{1-v_2}\Big)^{1/2}=\frac{2z_2}{1-2z_2^2}, \quad (1-v_2^2)^{1/2}=\frac{1-4z_2^4}{1+4z_2^4}.\]
The quantity $\kappa^{(6)}(v_1,v_2)$ in terms of $z_1,z_2$ thus becomes
\begin{align*}
\kappa^{(7)}(z_1,z_2):=\delta^{b_2}\frac{
b_2\frac{4z_1^2}{(1+z_1^2)^2}\big(\psi-\frac{2z_1}{1-z_1^2}\big)+
2b_3K^{(2)}_\delta\, C^{(b_3)}_{\delta}\, \frac{16 z_2^4}{(1 + 2
z_2^2)^2 (1 + 4
z_2^4)}\big(\psi-\frac{2z_2}{1-2z_2^2}\big)}{\psi-\Big(\frac{1-z_1^2}{1+z_1^2}\Big)^{b_2}\Big(\frac{1-4z_2^4}{1+4z_2^4}\Big)^{b_3}}.
\end{align*}
Our goal is to show that, for $b_2,b_3\geq 0$ satisfying $b_2+b_3\leq \Delta-1=5$, there do not exist $z_1,z_2$ in the range $0\leq z_1\leq \sqrt{2}-1$ and  $0\leq z_2\leq \frac{1}{2}(\sqrt{3}-1)$ such that $\kappa^{(7)}(z_1,z_2)> \tau_{b_2,b_3}$ where the constants $\tau$ are as in \eqref{eq:constantstau}. This can be done symbolically using Mathematica. We give the code in Section~\ref{sec:bootphase1}.
\end{proof}

We use Lemma~\ref{lem:bootphase1} to show the following.
\begin{lemma}\label{lem:bootphase2}
Let $b_5=b_6=0$ and $B$ be a non-negative integer less than or equal to $\Delta-1=5$. For all non-negative integers $b_2,b_3,b_4$ such that $b_2+b_3+b_4=B$, it holds that $\kappa^{(5)}(\yhat)\leq \tau_{B,0}$ where the constants $\tau_{B,0}$ are given by \eqref{eq:constantstau}.
\end{lemma}
\begin{proof}[Proof of Lemma~\ref{lem:bootphase2}]
We may assume that $b_4\geq 1$ (when $b_4=0$, the bounds in the lemma follow immediately from Lemma~\ref{lem:bootphase1}). For $b_5=b_6=0$, the quantity $\kappa^{(5)}(\yhat)$ simplifies into
\begin{equation*}
\kappa^{(5)}(\yhat):=\delta^{b_2}\cdot \frac{b_{2}\, g_1(\hat{y}_1)+2b_3 K^{(2)}_\delta\, C^{(b_3)}_{\delta}\, g_2(\hat{y}_2)+3b_4 K^{(3)}_\delta\,  g_3(\hat{y}_3)}{\psi-(\hat{y}_1)^{b_{2}\chi}(\hat{y}_2)^{b_{3}\chi}(\hat{y}_3)^{b_{4}\chi}},
\end{equation*}
where we used that $K^{(1)}_\delta=1$ (note that the values of the variables $\hat{y}_4, \hat{y}_5$ do not affect the value of $\kappa^{(4)}$ when $b_5=b_6=0$). The proof splits into two cases depending on whether $b_3$ is zero.

\textbf{Case I}: $b_3\geq 1$. Let $A:=(\hat{y}_1)^{b_{2}\chi}(\hat{y}_2)^{b_{3}\chi}$. Since $\hat{y}_1,\hat{y}_2 \in[0,1]$, we have the crude bound $0\leq A\leq 1$. By Lemma~\ref{lem:bootphase1} (see also \eqref{eq:kappa4y1y2}), we have that
\[\delta^{b_2}\cdot \frac{b_{2}\, g_1(\hat{y}_1)+2b_3 K^{(2)}_\delta\, C^{(b_3)}_{\delta}\, g_2(\hat{y}_2)}{\psi-A}\leq \tau_{b_2,b_3},\]
where the values of the constants $\tau_{b_2,b_3}$ are as in Lemma~\ref{lem:bootphase1} (cf. equation \eqref{eq:constantstau}).
It follows that
\[\kappa^{(5)}(\yhat)\leq \kappa^{(6)}(A,\hat{y}_3), \mbox{ where } \kappa^{(6)}(A,\hat{y}_3):= \frac{\tau_{b_2,b_3}(\psi-A)+3\delta^{b_2}\, b_4\, K^{(3)}_\delta\,  g_3(\hat{y}_3)}{\psi-A(\hat{y}_3)^{b_{4}\chi}}.\]
We next perform a transformation for the variable $\hat{y}_3$ (similar to the one used in the proof of Lemma~\ref{lem:bootphase1}), namely, we set $v_3:=(1-\hat{y}_3)^{1/3}$ so that $v_3\in [0,1/2]$. From the definition of the function $g_3$ (cf. equation \eqref{eq:gwfunction}), we have that
\begin{equation}\label{eq:g3subs}
g_3(\hat{y}_3)=\frac{1-\hat{y}_3}{\hat{y}_3}h\big((1-\hat{y}_3)^{1/3}\big)=\frac{v_3^3}{1-v_3^3}h(v_3)=\frac{v_3^3}{1+v_3+v_3^2}\Big(\psi-\big(\frac{v_3}{1-v_3}\big)^{\chi}\Big).
\end{equation}
It follows that the quantity $\kappa^{(6)}(A,\hat{y}_3)$ as a function of $A,v_3$ can now be written as
\[\kappa^{(7)}(A,v_3):= \frac{\tau_{b_2,b_3}(\psi-A)+3\delta^{b_2}\, b_4\, K^{(3)}_\delta\,  \frac{v_3^3}{1+v_3+v_3^2}\Big(\psi-\big(\frac{v_3}{1-v_3}\big)^{\chi}\Big)}{\psi-A(1-v_3^3)^{b_{4}\chi}}.\]
We will show that
\begin{equation}\label{eq:bootphase2a}
\kappa^{(7)}(A,v_3)\leq \tau_{B,0}\mbox{ for all } 0\leq A\leq 1,\ 0\leq v_3\leq 1/2 \mbox{ (recall, $B=b_2+b_3+b_4$)}.
\end{equation}
We use Mathematica's \textsc{Resolve} function, see Section~\ref{sec:bootphase2} for details.

\textbf{Case II:} $b_3=0$. For $b_3=b_5=b_6=0$, the quantity $\kappa^{(4)}(\yhat)$ simplifies into
\begin{equation*}
\kappa^{(5)}(\yhat):=\delta^{b_2}\cdot \frac{b_{2}\, g_1(\hat{y}_1)+3b_4 K^{(3)}_\delta\,  g_3(\hat{y}_3)}{\psi-(\hat{y}_1)^{b_{2}\chi}(\hat{y}_3)^{b_{4}\chi}},
\end{equation*}
We next perform a transformation on the variables $\hat{y}_1,\hat{y}_3$ (similar to the one used in the proof of Lemma~\ref{lem:bootphase1}), namely, we set $v_1=1-\hat{y}_1$ and $v_3:=(1-\hat{y}_3)^{1/3}$ so that $v_1,v_3\in [0,1/2]$. Using \eqref{eq:g1g2subs} and \eqref{eq:g3subs}, we obtain the following expression for $\kappa^{(5)}$ in terms of $v_1,v_3$:
\begin{equation}
\kappa^{(8)}(v_1,v_3):=\delta^{b_2}\frac{b_2v_1\Big(\psi-\big(\frac{v_1}{1-v_1}\big)^{\chi}\Big)+3b_4K^{(3)}_{\delta}\, \frac{v^3_3}{1+v_3+v_3^2}\Big(\psi-\big(\frac{v_3}{1-v_3}\big)^{\chi}\Big)}{\psi-(1-v_1)^{b_2\chi}(1-v_3^3)^{b_4\chi}}.
\end{equation}
This quantity is still too complicated for Mathematica to resolve efficiently, so we will need one more transformation. In particular, let  $u_1,u_3$ be positive reals defined by
\[v_1=\frac{u^2_1}{1+u^2_1},\qquad v_3=\frac{u^2_3}{1+u^2_3},\]
and note that $0\leq u_1,u_3\leq 1$.
The quantity $\kappa^{(8)}$ in terms of $u_1,u_3$ becomes:
\begin{equation*}
\kappa^{(9)}(u_1,u_3):=\delta^{b_2}\frac{b_2\frac{u^2_1}{1+u^2_1}(\psi-u_1)+3b_4K^{(3)}_{\delta}\, \frac{u_3^6}{3 u_3^6+6 u_3^4+4 u_3^2+1}(\psi-u_3)}{\psi-\big(\frac{1}{1+u_1^2}\big)^{b_2\chi}\Big(1-\big(\frac{u^2_3}{1+u^2_3}\big)^3\Big)^{b_4\chi}}.
\end{equation*}
We will show that
\begin{equation}\label{eq:bootphase2b}
\kappa^{(9)}(u_1,u_3)\leq \tau_{B,0}\mbox{ for all } 0\leq u_1,u_3\leq 1 \mbox{ (note, $B=b_2+b_4$)}.
\end{equation}
We use Mathematica's \textsc{Resolve} function, see Section~\ref{sec:bootphase2} for details.

This completes the case analysis and therefore the proof of Lemma~\ref{lem:bootphase2}.
\end{proof}

\begin{lemma}\label{lem:bootphase3}
Let $b_6=0$ and $B$ be a non-negative integer less than or equal to $\Delta-1=5$. For all non-negative integers $b_2,b_3,b_4,b_5$ such that $b_2+b_3+b_4+b_5=B$, it holds that $\kappa^{(5)}(\yhat)\leq \tau_{B,0}$ where the constants $\tau_{B,0}$ are given by \eqref{eq:constantstau}.
\end{lemma}
\begin{proof}[Proof of Lemma~\ref{lem:bootphase3}]
We may assume that $b_5\geq 1$ (when $b_5=0$, the bounds in the lemma follow immediately from Lemma~\ref{lem:bootphase2}). For $b_6=0$ (note that the value of the variable $\hat{y}_5$ does not affect the value of $\kappa^{(5)}$), the quantity $\kappa^{(5)}(\yhat)$ becomes
\begin{equation*}
\kappa^{(5)}(\yhat):=\delta^{b_2}\cdot \frac{b_{2}\, g_1(\hat{y}_1)+2b_3 K^{(2)}_\delta\, C^{(b_3)}_{\delta}\, g_2(\hat{y}_2)+3b_4 K^{(3)}_\delta\,  g_3(\hat{y}_3)+4b_5 K^{(4)}_\delta\,  g_4(\hat{y}_4)}{\psi-(\hat{y}_1)^{b_{2}\chi}(\hat{y}_2)^{b_{3}\chi}(\hat{y}_3)^{b_{4}\chi}(\hat{y}_4)^{b_{5}\chi}}.
\end{equation*}

Let $A:=(\hat{y}_1)^{b_{2}\chi}(\hat{y}_2)^{b_{3}\chi}(\hat{y}_3)^{b_{4}\chi}$. Since $\hat{y}_1,\hat{y}_2,\hat{y}_3 \in[0,1]$, we have the crude bound $0\leq A\leq 1$. By Lemma~\ref{lem:bootphase2}, we have that
\[\delta^{b_2}\cdot \frac{b_{2}\, g_1(\hat{y}_1)+2b_3 K^{(2)}_\delta\, C^{(b_3)}_{\delta}\, g_2(\hat{y}_2)+3b_4 K^{(3)}_\delta\,  g_3(\hat{y}_3)}{\psi-A}\leq \tau_{B',0}, \mbox{ where $B'=b_2+b_3+b_4$}.\]
where  the values of the constants $\tau_{B',0}$ are given by equation \eqref{eq:constantstau}.
Using that $\delta^{b_2}\leq 1$, it follows that
\[\kappa^{(5)}(\yhat)\leq \kappa^{(6)}(A,\hat{y}_4), \mbox{ where } \kappa^{(6)}(A,\hat{y}_4):= \frac{\tau_{B',0}(\psi-A)+4\, b_5\, K^{(4)}_\delta\,  g_4(\hat{y}_4)}{\psi-A(\hat{y}_4)^{b_{5}\chi}}.\]
We next perform a transformation on the variable $\hat{y}_4$ (similar to the one used in the proof of Lemma~\ref{lem:bootphase1}), namely, we set $v_4:=(1-\hat{y}_4)^{1/4}$ so that $v_4\in [0,1/2]$. From the definition of the function $g_4$ (cf. equation \eqref{eq:gwfunction}), we have that
\begin{equation*}
g_4(\hat{y}_4)=\frac{1-\hat{y}_4}{\hat{y}_4}h\big((1-\hat{y}_4)^{1/4}\big)=\frac{v_4^4}{1-v_4^4}h(v_4)=\frac{v_4^4}{1+v_4+v_4^2+v_4^3}\Big(\psi-\big(\frac{v_4}{1-v_4}\big)^{\chi}\Big).
\end{equation*}
It follows that the quantity $\kappa^{(6)}(A,\hat{y}_4)$ as a function of $A,v_4$ can now be written as
\[\kappa^{(7)}(A,v_4):= \frac{\tau_{B',0}(\psi-A)+4 b_5\, K^{(4)}_\delta\,  \frac{v_4^4}{1+v_4+v_4^2+v_4^3}\Big(\psi-\big(\frac{v_4}{1-v_4}\big)^{\chi}\Big)}{\psi-A(1-v_4^4)^{b_{5}\chi}}.\]
We will show that
\begin{equation}\label{eq:bootphase3}
\kappa^{(7)}(A,v_4)\leq \tau_{B,0}\mbox{ for all } 0\leq A\leq 1,\ 0\leq v_4\leq 1/2 \mbox{ (recall, $B=b_2+b_3+b_4+b_5=B'+b_5$)}.
\end{equation}
We use Mathematica's \textsc{Resolve} function, see Section~\ref{sec:bootphase3} for details.
\end{proof}

\begin{lemma}\label{lem:bootphase4}
Let $B$ be a non-negative integer less than or equal to $\Delta-1=5$. For all non-negative integers $b_2,b_3,b_4,b_5,b_6$ such that $b_2+b_3+b_4+b_5+b_6=B$, it holds that $\kappa^{(5)}(\yhat)\leq \tau_{B,0}$ where the constants $\tau_{B,0}$ are given by \eqref{eq:constantstau}.
\end{lemma}
\begin{proof}[Proof of Lemma~\ref{lem:bootphase4}]
We may assume that $b_6\geq 1$ (when $b_6=0$, the bounds in the lemma follow immediately from Lemma~\ref{lem:bootphase3}). Recall that the quantity $\kappa^{(5)}(\yhat)$ is given by
\begin{equation*}
\kappa^{(5)}(\yhat)=\delta^{b_2}\cdot \frac{b_{2}\, g_1(\hat{y}_1)+2b_3 K^{(2)}_\delta\, C^{(b_3)}_{\delta}\, g_2(\hat{y}_2)+3b_4 K^{(3)}_\delta\,  g_3(\hat{y}_3)+4b_5 K^{(4)}_\delta\,  g_4(\hat{y}_4)+5b_6 K^{(5)}_\delta\,  g_5(\hat{y}_5)}{\psi-(\hat{y}_1)^{b_{2}\chi}(\hat{y}_2)^{b_{3}\chi}(\hat{y}_3)^{b_{4}\chi}(\hat{y}_4)^{b_{5}\chi}(\hat{y}_5)^{b_{6}\chi}}.
\end{equation*}

Let $A:=(\hat{y}_1)^{b_{2}\chi}(\hat{y}_2)^{b_{3}\chi}(\hat{y}_3)^{b_{4}\chi}(\hat{y}_4)^{b_{5}\chi}$. Since $\hat{y}_1,\hat{y}_2,\hat{y}_3,\hat{y}_4 \in[0,1]$, we have the crude bound $0\leq A\leq 1$. By Lemma~\ref{lem:bootphase3}, we have that
\[\delta^{b_2}\cdot \frac{b_{2}\, g_1(\hat{y}_1)+2b_3 K^{(2)}_\delta\, C^{(b_3)}_{\delta}\, g_2(\hat{y}_2)+3b_4 K^{(3)}_\delta\,  g_3(\hat{y}_3)+4b_5 K^{(4)}_\delta\,  g_4(\hat{y}_4)}{\psi-A}\leq \tau_{B',0},\]
where $B'=b_2+b_3+b_4+b_5$ and the values of the constants $\tau_{B',0}$ are given by  equation \eqref{eq:constantstau}.
Using that $\delta^{b_2}\leq 1$, it follows that
\[\kappa^{(5)}(\yhat)\leq \kappa^{(6)}(A,\hat{y}_5), \mbox{ where } \kappa^{(6)}(A,\hat{y}_5):= \frac{\tau_{B',0}(\psi-A)+5\, b_6\, K^{(5)}_\delta\,  g_5(\hat{y}_5)}{\psi-A(\hat{y}_5)^{b_{6}\chi}}.\]
We next perform a transformation on the variable $\hat{y}_5$ (similar to the one used in the proof of Lemma~\ref{lem:bootphase1}), namely, we set $v_5:=(1-\hat{y}_5)^{1/5}$ so that $v_5\in [0,1/2]$. From the definition of the function $g_5$ (cf. equation \eqref{eq:gwfunction}), we have that
\begin{equation*}
g_5(\hat{y}_5)=\frac{1-\hat{y}_5}{\hat{y}_5}h\big((1-\hat{y}_5)^{1/5}\big)=\frac{v_5^5}{1-v_5^5}h(v_5)=\frac{v_5^5}{1+v_5+v_5^2+v_5^3+v_5^4}\Big(\psi-\big(\frac{v_5}{1-v_5}\big)^{\chi}\Big).
\end{equation*}
It follows that the quantity $\kappa^{(6)}(A,\hat{y}_5)$ as a function of $A,v_5$ can now be written as
\[\kappa^{(7)}(A,v_5):= \frac{\tau_{B',0}(\psi-A)+5 b_6\, K^{(5)}_\delta\,  \frac{v_5^5}{1+v_5+v_5^2+v_5^3+v_5^4}\Big(\psi-\big(\frac{v_5}{1-v_5}\big)^{\chi}\Big)}{\psi-A(1-v_5^5)^{b_{6}\chi}}.\]
We will show that
\begin{equation}\label{eq:bootphase4}
\kappa^{(7)}(A,v_5)\leq \tau_{B,0}\mbox{ for all } 0\leq A\leq 1,\ 0\leq v_5\leq 1/2 \mbox{ (recall, $B=B'+b_6$)}.
\end{equation}
We use Mathematica's \textsc{Resolve} function, see Section~\ref{sec:bootphase4} for details.
\end{proof}
The proof of Lemma~\ref{lem:proofBbounds}, which was important in proving Lemma~\ref{lem:potentialfunction}, is now immediate.
{
\renewcommand{\thetheorem}{\ref{lem:proofBbounds}}
\begin{lemma}
Let $\Delta=6$ and $B$ be a non-negative integer less than or equal to $\Delta-1=5$. Recall that $\alpha=1-10^{-4}$. There exists a constant $\epsilon_B\leq \alpha$ such that for all non-negative integers $b_2,b_3,b_4,b_5,b_6$ with $b_2+b_3+b_4+b_5+b_6=B$, it holds that
$\kappa^{(3)}(\that)\leq \epsilon_B$.

In particular, we will show that
\begin{equation*}\tag{\ref{eq:boundsonepsilons}}
\epsilon_0=0,\quad \epsilon_1=6/10, \quad \epsilon_2=7/10, \quad \epsilon_3=83/100,\quad  \epsilon_4=91/100, \quad \epsilon_5=\alpha=1-10^{-4}.
\end{equation*}
\end{lemma}
\addtocounter{theorem}{-1}
}
\begin{proof}[Proof of Lemma~\ref{lem:proofBbounds}]
By the definition of $\yhat$, we have that $\kappa^{(3)}(\that)=\kappa^{(5)}(\yhat)$. Now, just use Lemma~\ref{lem:bootphase4} and observe that  the $\epsilon_B$ in \eqref{eq:boundsonepsilons} and the constants $\tau_{B,0}$ in  \eqref{eq:constantstau} satisfy $\tau_{B,0}\leq \epsilon_B$.
\end{proof}

\subsection{Simplifying the optimisation using geometric means}\label{sec:assym100}
In this section, we prove Lemma~\ref{lem:100assym}, which we restate here for convenience. Roughly, the lemma bounds the contribution of  a $w_i$ with $w_i=w$ to $\kappa$.
The main accomplishment here is the significant reduction of the number of variables; initially the contribution
is a function of $w$ variables $t_1,\hdots,t_w$.
The lemma shows that we can reduce the number of variables to 1 by considering the geometric mean of the $t_j$'s. The challenge here is to deal with the asymmetry caused by the $\delta$ terms  without introducing too much slack in the argument, especially for
small values of $w$ (say $w\leq 4$).
{
\renewcommand{\thetheorem}{\ref{lem:100assym}}
\begin{lemma}
 \statelemassym
 \end{lemma}
\addtocounter{theorem}{-1}
}

Recall that the function $h(t)$ is given by
\begin{equation*}\tag{\ref{eq:functionh}}
h(t)=(1-t)\left[\psi-\left(\frac{t}{1-t}\right)^{\chi}\right] \mbox{ for } t\in[0,1/2],
\end{equation*}
where $\chi=1/2$, $\psi=13/10$. We begin with the following lemma.
\begin{lemma}\label{lem:concav1}
The function $h(e^y)$ is a concave function of $y$ in the interval $(-\infty,\ln(1/2)]$.
\end{lemma}
\begin{proof}[Proof of Lemma~\ref{lem:concav1}]
Let $f(y):=h(e^{y})$ for $y\in (-\infty,\ln(1/2)]$. We will show that
\begin{equation}\label{eq:fppy2}
f''(y)\leq0 \mbox{ for all } y\in (-\infty,\ln(1/2)].
\end{equation}
For convenience, we use Mathematica's \textsc{Resolve} function, see Section~\ref{sec:concav1} for the code.
\end{proof}

As an immediate consequence of Lemma~\ref{lem:concav1} and Jensen's inequality, we obtain the following inequality for  $w=2,3,\hdots$,  for all $t_1,\hdots,t_w\in[0,1/2]$:

\begin{equation}\label{eq:assymyhn}
\sum^{w}_{j=1}h(t_j)=\sum^{w}_{j=1}h(e^{\ln t_j })\leq w\, h\big(e^{\frac{1}{w}\sum^w_{j=1}\ln t_j}\big)=w\, h(t),
\end{equation}
where $t$ is the geometric mean of the $t_i$'s, i.e., $t=(t_1\cdots t_w)^{1/w}$. Using that $\delta\in (0,1]$ and inequality \eqref{eq:assymyhn}, it follows that
\begin{equation*}
\sum^{w}_{j=1}\frac{1}{\delta^{(j-1)(\Delta-1)}}h(t_j)\leq \frac{1}{\delta^{(w-1)(\Delta-1)}}\sum^{w}_{j=1}h(t_j)\leq w\, \frac{1}{\delta^{(w-1)(\Delta-1)}}\, h\big(t\big).
\end{equation*}
This proves the bounds on $K^{(w)}_\delta$ stated in the lemma for all integer $w\geq 5$.

For the bounds on $K^{(w)}_\delta$ stated in the lemma for $w=2,3,4$ we will have to work harder. Our goal is to prove the following inequalities for $t_i\in [0,1/2]$ ($i=1,2,\hdots$):
\begin{align}
h(t_1)+\frac{1}{\delta^5}\,h(t_2)&\leq 2K^{(2)}_\delta h(\sqrt{t_1t_2}),\label{eq:w2aa}\\
h(t_1)+\frac{1}{\delta^5}\,h(t_2)+\frac{1}{\delta^{10}}\,h(t_3)&\leq 3K^{(3)}_\delta h(\sqrt[3]{t_1t_2t_3}),\label{eq:w3aa}\\
h(t_1)+\frac{1}{\delta^5}\,h(t_2)+\frac{1}{\delta^{10}}\,h(t_3)+\frac{1}{\delta^{15}}\,h(t_4)&\leq
4K^{(4)}_\delta h(\sqrt[4]{t_1t_2t_3t_4}).\label{eq:w4aa}
\end{align}
To prove these, we will need the following inequalities.
\begin{lemma}\label{lem:hgeneralproof}
Let $A_1,A_2>0$ be real numbers. There exists $A>0$ such that for $t_1,t_2\in [0,1/2]$, it holds that
\begin{equation}\label{eq:ht1ht2}
A_1 h(t_1)+A_2 h(t_2)\leq A\, h(\sqrt{t_1 t_2}).
\end{equation}
In particular, inequality \eqref{eq:ht1ht2} holds for the following values of $A_1,A_2,A$:
\begin{align}
A_1&=1,& A_2&=\frac{1}{\delta^5},& A&=2K^{(2)}_\delta,\label{eq:567rty}\\
A_1&=2,& A_2&=\frac{2}{\delta^5}K^{(2)}_\delta,& A&=4\cdot \frac{1120}{1000},\label{eq:568rty}\\
A_1&=1,& A_2&=\frac{1}{\delta^{15}},& A&=\frac{5}{2},\label{eq:569rty}\\
A_1&=\frac{2}{\delta^5}K^{(2)}_\delta,& A_2&=\frac{5}{2},& A&=4\cdot K^{(4)}_\delta,\label{eq:570rty}
\end{align}
where $K^{(2)}_\delta=1069/1000$ and $K^{(4)}_\delta=1225/1000$ are as in Lemma~\ref{lem:100assym}.
\end{lemma}
\begin{proof}
The existence of such an $A$ follows by standard continuity and compactness arguments. The positivity of $A$ is also easy to prove. We thus focus on the more intricate task of verifying \eqref{eq:ht1ht2} for the values of $A_1,A_2,A$ given in the statement of the lemma.

We will use  Mathematica's \textsc{Resolve} function. To do this, we first need to rationalize the expressions which can be achieved for $\chi=1/2$. In particular, we will use the transformations
\begin{equation}\label{eq:xsw}
t_1=\frac{4x_1^2}{(1+x_1^2)^2},\quad t_2=\frac{4x_2^2}{(1+x_2^2)^2}
\end{equation}
for $0\leq x_1,x_2\leq \sqrt{2}-1$. Under these transformations, for $\chi=1/2$, we obtain that
\begin{equation}\label{eq:xsw2}
\Big(\frac{t_1}{1-t_1}\Big)^{\chi}=\frac{2x_1}{1-x_1^2}, \quad \Big(\frac{t_2}{1-t_2}\Big)^{\chi}=\frac{2x_2}{1-x_2^2},\quad \sqrt{t_1 t_2}=\frac{4x_1 x_2}{(1+x_1^2)(1+x_2^2)}.
\end{equation}
We are quite close to rationalizing the desired inequality, we only have to address the rationalization of $\big(\frac{\sqrt{t_1 t_2}}{1-\sqrt{t_1 t_2}}\big)^\chi$. Unfortunately, we will have to explicitly eradicate the radical for this expression.

In particular, inequality \eqref{eq:ht1ht2} is equivalent to
\begin{equation}\label{eq:cvbnm12}
\frac{A_1(1-t_1)\left[\psi-\left(\frac{t_1}{1-t_1}\right)^{\chi}\right]+A_2(1-t_2)\left[\psi-\left(\frac{t_2}{1-t_2}\right)^{\chi}\right]}{A(1-\sqrt{t_1t_2})}\leq \psi-\Big(\frac{\sqrt{t_1 t_2}}{1-\sqrt{t_1 t_2}}\Big)^\chi.
\end{equation}
Inequality \eqref{eq:cvbnm12} will follow from the following inequalities:
\begin{equation}\label{eq:cvbnm13}
\frac{A_1(1-t_1)\left[\psi-\left(\frac{t_1}{1-t_1}\right)^{\chi}\right]+A_2(1-t_2)\left[\psi-\left(\frac{t_2}{1-t_2}\right)^{\chi}\right]}{A(1-\sqrt{t_1t_2})}\leq \psi,
\end{equation}
and
\begin{equation}\label{eq:cvbnm14}
\frac{\sqrt{t_1 t_2}}{1-\sqrt{t_1 t_2}}\leq \bigg(\psi-\frac{A_1(1-t_1)\left[\psi-\left(\frac{t_1}{1-t_1}\right)^{\chi}\right]+A_2(1-t_2)\left[\psi-\left(\frac{t_2}{1-t_2}\right)^{\chi}\right]}{A(1-\sqrt{t_1t_2})}\bigg)^2.
\end{equation}
Note that \eqref{eq:cvbnm13} allows us to take square roots in \eqref{eq:cvbnm14}, and thus \eqref{eq:cvbnm12} follows.

It remains to prove \eqref{eq:cvbnm13} and \eqref{eq:cvbnm14}. Using the substitutions \eqref{eq:xsw} and \eqref{eq:xsw2}, inequalities \eqref{eq:cvbnm13} and \eqref{eq:cvbnm14} are equivalent to
\begin{equation}\label{eq:cvbnm15}
0 \leq \psi-\frac{A_1\frac{(1-x_1^2)^2}{(1+x_1^2)^2}\left(\psi-\frac{2x_1}{1-x_1^2}\right)+A_2\frac{(1-x_2^2)^2}{(1+x_2^2)^2}\left(\psi-\frac{2x_2}{1-x_2^2}\right)}{A\Big(1-\frac{4x_1 x_2}{(1+x_1^2)(1+x_2^2)}\Big)}
\end{equation}
and
\begin{equation}\label{eq:toverifysecond}
\frac{4x_1x_2}{(1+x_1^2)(1+x_2^2)-4x_1x_2}\leq \left(\psi-\frac{A_1\frac{(1-x_1^2)^2}{(1+x_1^2)^2}\left(\psi-\frac{2x_1}{1-x_1^2}\right)+A_2\frac{(1-x_2^2)^2}{(1+x_2^2)^2}\left(\psi-\frac{2x_2}{1-x_2^2}\right)}{A\Big(1-\frac{4x_1 x_2}{(1+x_1^2)(1+x_2^2)}\Big)}\right)^2,
\end{equation}
respectively. The last inequalities involve rational expressions and can be resolved using Mathematica for the values of $A_1,A_2,A$ given in the statement of the lemma, the code can be found in Section~\ref{sec:hgeneralproof}.
\end{proof}

We now return to the task of proving the inequalities \eqref{eq:w2aa}, \eqref{eq:w3aa} and \eqref{eq:w4aa}. Inequality \eqref{eq:w2aa} is an immediate consequence of Lemma~\ref{lem:hgeneralproof} (cf. the values in \eqref{eq:567rty}).

To prove  \eqref{eq:w3aa}, we will use the following inequality
\begin{equation}\label{eq:inini1}
h(t_1)+\frac{1}{\delta^5}\,h(t_2)+\frac{1}{\delta^{10}}\,h(t_3)+h(t_4)\leq 4 \cdot \frac{1120}{1000} h(\sqrt[4]{t_1t_2t_3t_4}).
\end{equation}
Applying \eqref{eq:inini1} with $t_4=\sqrt[3]{t_1t_2t_3}$ (note that with this value of $t_4$ it holds that $\sqrt[4]{t_1t_2t_3t_4}=\sqrt[3]{t_1t_2t_3}$) yields
\[h(t_1)+\frac{1}{\delta^5}\,h(t_2)+\frac{1}{\delta^{10}}\,h(t_3)\leq 3\left(\frac{4 \cdot \frac{1120}{1000}-1}{3}\right) h(\sqrt[3]{t_1t_2t_3})=3K^{(3)}_\delta h(\sqrt[3]{t_1t_2t_3}),\]
which proves \eqref{eq:w3aa}. It remains to prove \eqref{eq:inini1}, which follows by adding the following inequalities:
\begin{align}
h(t_1)+h(t_4)&\leq 2h(\sqrt{t_1t_4}),\label{eq:w3bb}\\
\frac{1}{\delta^5}\, h(t_2)+\frac{1}{\delta^{10}}\, h(t_3)&\leq \frac{2}{\delta^5}K^{(2)}_\delta h(\sqrt{t_2t_3}),\label{eq:w3bb1}\\
2h(\sqrt{t_1t_4})+\frac{2}{\delta^5}K^{(2)}_\delta
h(\sqrt{t_2t_3})&\leq 4 \cdot \frac{1120}{1000}
h(\sqrt[4]{t_1t_2t_3t_4}). \label{eq:w3bb2}
\end{align}
Inequality \eqref{eq:w3bb} is an immediate consequence of Lemma~\ref{lem:concav1}. Inequality \eqref{eq:w3bb1} is an immediate consequence of inequality \eqref{eq:w2aa} (multiplied by $1/\delta^5$). For inequality \eqref{eq:w3bb2}, we use the transformations $u_1=\sqrt{t_1t_4}$ and $u_2=\sqrt{t_2t_3}$, so that we  need to show
\begin{equation}\label{eq:w3bb3}
2h(u_1)+\frac{2}{\delta^5}K^{(2)}_\delta h(u_2)\leq 4 \cdot \frac{1120}{1000} h(\sqrt{u_1u_2})
\end{equation}
for $u_1,u_2\in [0,1/2]$, which follows from Lemma~\ref{lem:hgeneralproof} (cf. the values \eqref{eq:568rty}).

Finally, we conclude with the proof of inequality \eqref{eq:w4aa}. This is obtained by adding the following three inequalities:
\begin{align}
\frac{1}{\delta^5}\, h(t_2)+\frac{1}{\delta^{10}}\, h(t_3)&\leq \frac{2}{\delta^5}K^{(2)}_\delta h(\sqrt{t_2t_3}),\label{eq:w4bb1}\\
h(t_1)+\frac{1}{\delta^{15}}h(t_4)&\leq \frac{5}{2} h(\sqrt{t_1t_4}),\label{eq:w4bb2}\\
\frac{2}{\delta^5}\,K^{(2)}_\delta h(\sqrt{t_2t_3})+\frac{5}{2} h(\sqrt{t_1t_4})&\leq  4K^{(4)}_\delta h(\sqrt[4]{t_1t_2t_3t_4}).\label{eq:w4bb3}
\end{align}
Inequality \eqref{eq:w4bb1} is an immediate consequence of \eqref{eq:w2aa} (again, multiplied by $1/\delta^5$). Inequality \eqref{eq:w4bb2} follows from Lemma~\ref{lem:hgeneralproof} (cf. the values \eqref{eq:569rty}). Finally, inequality \eqref{eq:w4bb3} can be proved using an analogous transformation as the one used to prove \eqref{eq:w3bb2}; the required analogue of inequality \eqref{eq:w3bb3} has been proved in Lemma~\ref{lem:hgeneralproof} (cf. the values \eqref{eq:570rty}).

This concludes the proof of Lemma~\ref{lem:100assym}.

 \subsection{Eliminating large arity clauses
from consideration}\label{sec:thnmi}
In this section, we prove Lemma~\ref{lem:thnmi}, which we restate here for convenience.
Recall that in the construction of the optimisation problem from the original
correlation-decay argument,
 $w_i$ is the arity of the $i$-th clause containing $x$ minus one.  Intuitively, clauses with large arity should not affect significantly the correlation decay.
 The following lemma captures this in a quantitative way which is sufficient for our needs (for clauses with $w_i\geq 6$).
{\renewcommand{\thetheorem}{\ref{lem:thnmi}}
\begin{lemma}
\statelemthnmi
\end{lemma}
\addtocounter{theorem}{-1}
}
\begin{proof}
We will show that for all integers $w\geq 6$ and all $t\in[0,1/2]$, it holds that
\begin{equation}\label{eq:tedious}
\frac{t^{w}}{1-t^{w}}\leq \frac{63}{2^{w}-1}\frac{t^{6}}{1-t^{6}}.
\end{equation}
We will also show that for $\chi=1/2, \psi=13/10$, it holds that
\begin{equation}\label{eq:v3max}
\max_{t\in[0,1/2]}\frac{t^6}{1-t^6}h(t)\leq M_1,\mbox{ where } M_1=\frac{1}{410}.
\end{equation}
Finally, we will show that  for integer $w\geq 6$, $l_{w}=\left\lceil \log_6(w+1)\right\rceil$, $\delta=9789/10000$, it holds that
\begin{equation}\label{eq:Kwlarge}
\frac{63M_1}{2^{w}-1}w\, K^{(w)}_{\delta}\, \alpha^{-l_{w}}\leq \frac{M}{\alpha},
\end{equation}
where  $K^{(w)}_{\delta}$ for $w\geq 6$ is given by Lemma~\ref{lem:100assym} and $M=25/1000$ is as in the statement of the lemma. The lemma then follows by multiplying  \eqref{eq:tedious}, \eqref{eq:v3max} and \eqref{eq:Kwlarge}.

We start with the verification of \eqref{eq:tedious}. Note that \eqref{eq:tedious} holds at equality for $t=1/2$, so it suffices to show that for all integer $w\geq 6$, the function
\[f(t):=\frac{t^{w}(1-t^6)}{t^6(1-t^{w})}=\frac{t^{w-6}(1-t^6)}{1-t^{w}}\]
is increasing for $t\in [0,1/2]$. For $w=6$, there is nothing to show, so we may assume that $w\geq 7$. We then calculate (see Section~\ref{sec:fghc} for the calculation)
\begin{equation}\label{eq:somederivative}
f'(t)=\frac{t^{w-7} \, p(t)}{\left(t^w-1\right)^2}, \mbox{ where } p(t):=6 t^w-t^6w+w-6.
\end{equation}
so we only need to show that $p(t)\geq 0$ for $t\in[0,1/2]$. Note that $p'(t)=6 w( t^{w-1}-t^5)\leq 0$ for all $t\in[0,1/2]$ since  $w\geq 7$. It follows that
\[p(t)\geq p(1/2)=6(1/2)^w+(63/64)w-6\geq (63/64)w-6\geq 1/2,\]
where in the last inequality we again used that $w\geq 7$. This completes the verification of \eqref{eq:tedious}.

We next verify \eqref{eq:v3max}. For convenience, we use  Mathematica's \textsc{Resolve} function for that, see Section~\ref{sec:fghc}.

Finally, we verify \eqref{eq:Kwlarge}. Since $l_{w}\leq \log_6(w+1)+1$ and $\alpha=1-10^{-4}<1$, it suffices to show that for all $w\geq 6$ it holds that
\begin{equation}\label{eq:Kwlargeb}
\frac{63M_1}{2^{w}-1}w\, \Big(\frac{1}{\delta}\Big)^{(w-1)(\Delta-1)}\, \alpha^{-\log_6(w+1)}\leq M.
\end{equation}
It is a matter of numerical calculations to show that $\alpha^{-1}\leq \exp(11\cdot 10^{-5})$. Thus, to show \eqref{eq:Kwlargeb}, it suffices to show that
\begin{equation}\label{eq:Kwlargeb2}
\frac{63M_1}{2^{w}-1}w\, \Big(\frac{1}{\delta}\Big)^{(w-1)(\Delta-1)}\, (w+1)^{11\cdot 10^{-5}/\ln6}\leq M.
\end{equation}
We view the lhs in \eqref{eq:Kwlargeb2} as a function of $w$, say $f(w)$. We will prove that
\begin{equation}\label{eq:Kwlargeb3}
f(6)\leq M\mbox{ and } f(w+1)/f(w)\leq 1\mbox{ for all } w\geq 6,
\end{equation}
from which inequality \eqref{eq:Kwlargeb2} follows. The first inequality in \eqref{eq:Kwlargeb3} follows by a numerical calculation, see Section~\ref{sec:fghc} for details. For the second inequality in \eqref{eq:Kwlargeb3}, we have
\[\frac{f(w+1)}{f(w)}=\Big(\frac{1}{\delta}\Big)^{\Delta-1}\frac{2^w-1}{2^{w+1}-1}(1+w^{-1})\big(1+(w+1)^{-1}\big)^{11\cdot 10^{-5}/\ln6}\]
Note that $\frac{2^w-1}{2^{w+1}-1}\leq 1/2$ and hence for $w\geq 6$, we have the bound
\begin{equation}\label{eq:Kwlargeb303030}
\frac{f(w+1)}{f(w)}\leq \frac{1}{2}\Big(\frac{1}{\delta}\Big)^{\Delta-1}(1+6^{-1})(1+7^{-1})^{11\cdot 10^{-5}/\ln6}<1,
\end{equation}
where the last inequality follows by a numerical calculation, see Section~\ref{sec:fghc} for details. This concludes the proof of \eqref{eq:Kwlargeb3} and thus the proof of Lemma~\ref{lem:thnmi}.
\end{proof}

\subsection{The contribution of arity 3 clauses}\label{sec:assym2}
In this section, we give the proof of Lemma~\ref{lem:assym2}. Roughly, the lemma bounds the aggregate contribution of arity 3 clauses along with the effect of the creation of arity 2 clauses (due to the pinnings when processing arity 3 clauses). This was used to further reduce the number of variables.

To prove Lemma~\ref{lem:assym2}, we will use H\"{o}lder's inequality. Let $p>1$. For $q=p/(p-1)$ and positive real numbers $\alpha_i,\beta_i$ for $i\in [b_3]$, it holds that
\[\sum^{b_3}_{i=1}\alpha_i \beta_i\leq \bigg(\sum^{b_3}_{i=1}\alpha_i^p\bigg)^{1/p}\bigg(\sum^{b_3}_{i=1}\beta_i^q\bigg)^{1/q}.\]
This yields
\begin{equation}\label{eq:holderholder}
\sum^{b_3}_{i=1}\left(\frac{1}{\delta}\right)^{b_3-i}g_2(y_i)\leq \bigg(\sum^{b_3}_{i=1}\Big(\frac{1}{\delta}\Big)^{(b_3-i)p}\bigg)^{1/p}\bigg(\sum^{b_3}_{i=1}(g_2(y_i))^q\bigg)^{1/q}.
\end{equation}
For $p=27/2$, we have $q=27/25$. Recall from equation \eqref{eq:gwfunction} that the function $g_2(y)$ is given by
\begin{equation}\tag{\ref{eq:gwfunction}}
g_2(y):=\frac{(1-y)}{y}h\big((1-y)^{1/2}\big)=\frac{(1-y)\big(1-(1-y)^{1/2}\big)}{y}\bigg(\psi-\Big(\frac{(1-y)+(1-y)^{1/2}}{y}\Big)^{\chi}\bigg).
\end{equation}
\begin{lemma}\label{lem:g2q}
For $q=27/25$ and $\Delta=6$, the function $(g_2(e^t)))^q$ is a concave function of $t$ in the interval $[\ln(3/4),\ln(1-\frac{1}{(2^{\Delta-1}+1)^2})]$.
\end{lemma}
\begin{proof}[Proof of Lemma~\ref{lem:g2q}]
Let $\bar{g}_2(t):=(g_2(e^t))^q$. Our goal is to show that $\bar{g}_2''(t)\leq 0$ for all $t\in [\ln(3/4),\ln(1-\frac{1}{(2^{\Delta-1}+1)^2})]$. We have
\begin{equation*}
\begin{aligned}
\bar{g}_2'(t)&=q\, \big(g_2(e^t)\big)^{q-1} g_2'(e^t) e^t,\\
\bar{g}_2''(t)&=q(q-1)\, \big(g_2(e^t)\big)^{q-2}\big(g_2'(e^t)\big)^2 e^{2t}+q \big(g_2(e^t)\big)^{q-1} g_2''(e^t)e^{2t}+q \big(g_2(e^t)\big)^{q-1} g_2'(e^t) e^t.
\end{aligned}
\end{equation*}
Observe that $g_2(e^t) e^t>0$ for all $t\in [\ln(3/4),\ln(1-\frac{1}{(2^{\Delta-1}+1)^2})]$. Thus, using the transformation $y=e^t$, it suffices to show that
\begin{equation}\label{eq:xswxsw34}
(q-1) y\, \big(g_2'(y)\big)^2 +g_2(y) \big(g_2'(y)+y\, g_2''(y)\big)  \leq 0,
\end{equation}
for all $y\in [3/4,1-\frac{1}{(2^{\Delta-1}+1)^2}]$. For convenience, we verify \eqref{eq:xswxsw34} using Mathematica's \textsc{Resolve} function, see Section~\ref{sec:rfv} for the code.
\end{proof}
We are now ready to prove Lemma~\ref{lem:assym2}, which we restate here for convenience.
{\renewcommand{\thetheorem}{\ref{lem:assym2}}
\begin{lemma}
\statelassymtwo\end{lemma}
\addtocounter{theorem}{-1}
}
\begin{proof}
For $p=27/2$ and $q=27/25$, inequality \eqref{eq:holderholder} and Lemma~\ref{lem:g2q} yield that
\[\sum^{b_3}_{i=1}\Big(\frac{1}{\delta}\Big)^{b_3-i}g_2(y_i)\leq b_3^{1/q}\bigg(\sum^{b_3}_{i=1}\Big(\frac{1}{\delta}\Big)^{(b_3-i)p}\bigg)^{1/p}\, g_2(\sqrt[b_3]{y_1\cdots y_{b_3}}).\]
Note that for $b_3=1,\hdots,5$ it holds that
\[\sum^{b_3}_{i=1}\Big(\frac{1}{\delta}\Big)^{(b_3-i)p}=\frac{\frac{1}{\delta^{b_3p}}-1}{\frac{1}{\delta^{p}}-1}=\frac{1}{\delta^{(b_3-1)p}}\cdot \frac{1-\delta^{b_3p}}{1-\delta^{p}}.\]
Using this and $q=p/(p-1)$, we obtain that
\[\sum^{b_3}_{i=1}\left(\frac{1}{\delta}\right)^{b_3-i}g_2(y_i)\leq b_3\, C^{(b_3)}_{\delta} g_2(\sqrt[b_3]{y_1\cdots y_{b_3}}).\]
where $C^{(b_3)}_{\delta}=\frac{1}{\delta^{b_3-1}}\Big(\frac{1-\delta^{b_3p}}{b_3(1-\delta^p)}\Big)^{1/p}$, as desired. The numerical bounds on the values of $C^{(b_3)}_{\delta}$
given in the statement of Lemma~\ref{lem:assym2}  for $b_3=1,\hdots,5$ can be verified by a direct calculation using Mathematica, see Section~\ref{sec:rfvbbbb}. This concludes the proof of Lemma~\ref{lem:assym2}.
\end{proof}

\section{Hardness for Approximate Counting}\label{sec:hardnesstotal}
In this section, we prove the hardness results stated in the introduction for the problem of counting independent sets in hypergraphs and the problem of counting dominating sets in graphs.

\subsection{Counting independent sets in hypergraphs}\label{sec:hardness}
In this section, we prove inapproximability results for the $\HyperIndSet$ problem. For this section, it will be convenient to return to the original hypergraph independent set formulation of the problem (instead of the monotone CNF formulation).
The proof is via a  reduction to the independent set model on graphs which was used by Bordewich \emph{et al.} \cite{BDK}. The precise inapproximability results for the hard-core model had not yet been proved at the time \cite{BDK} was written, so we carry out the details explicitly to obtain the bound that their reduction gives.

Namely, we will use the inapproximability result of Sly and Sun \cite{SlySun} for the hard-core model. We first remind the reader the relevant definitions. Let $\lambda>0$. For a graph $G=(V,E)$, the hard-core model with parameter $\lambda$ is a probability distribution over the set of independent sets of $G$; each independent set $I$ of $G$ has weight proportional to $\lambda^{|I|}$. The normalizing factor of this distribution is the partition function $Z_G(\lambda)$, formally defined as $Z_G(\lambda):=\sum_{I} \lambda^{|I|}$ where the sum ranges over all independent sets $I$ of $G$.
\begin{theorem}[\cite{SlySun}]\label{thm:slysunhardcore}
For $\Delta\geq 3$, let $\lambda_c(\Delta):=(\Delta-1)^{\Delta-1}/(\Delta-2)^{\Delta}$. For all $\lambda>\lambda_c(\Delta)$, it is  $\mathrm{NP}$-hard to approximate $Z_G(\lambda)$ on $\Delta$-regular graphs $G$,  even within an exponential factor.
\end{theorem}
\begin{theorem}\label{thm:inapprox}
Let $k\geq 2,\ \Delta\geq 3$ be integers. Suppose that $2^{\left\lceil k/2\right\rceil}-1<\frac{(\Delta-2)^{\Delta}}{(\Delta-1)^{\Delta-1}}$. Then, it is $\mathrm{NP}$-hard to approximate $\HyperIndSet$, even within an exponential factor.
\end{theorem}

\begin{proof}
Let $k\geq 2,\ \Delta\geq 3$ be integers satisfying  $2^{\left\lceil k/2\right\rceil}-1<\frac{(\Delta-2)^{\Delta}}{(\Delta-1)^{\Delta-1}}$ and let
\[\lambda:=1/(2^{\left\lceil k/2\right\rceil}-1).\]
Note that $\lambda>\lambda_c(\Delta)$ where $\lambda_c(\Delta)$ is as in Theorem~\ref{thm:slysunhardcore}. For convenience, let $k':=\left\lceil k/2\right\rceil$ in what follows.

Let $G=(V,E)$ be a $\Delta$-regular graph and set $n:=|V|$. We will construct a $(2k')$-uniform hypergraph $H=(U,\mathcal{F})$ with maximum degree $\Delta$  such that $|U|=k'|V|$, $|\mathcal{F}|=|E|$ and
\begin{equation}\label{eq:2c2onnection}
Z_H=(2^{k'}-1)^n\, Z_G(\lambda).
\end{equation}
Note that the size of $H$ is larger than the size of $G$ only by a constant factor. It thus follows that if we could approximate $\HyperIndSet$ within an arbitrarily small exponential factor, we could also approximate $Z_G(\lambda)$ within an (arbitrarily small) exponential factor for all $\Delta$-regular graphs $G$, contradicting Theorem~\ref{thm:slysunhardcore}.

It remains to construct the hypergraph $H=(U,\mathcal{F})$. Let
\[U=\bigcup_{v\in V}\{u_{v,1},\hdots,u_{v,k'}\},\quad \mathcal{F}=\bigcup_{(v,w)\in E}\{\{u_{v,1},\hdots,u_{v,k'},u_{w,1},\hdots,u_{w,k'}\}\}.\]
In words, every vertex $v$ of $G$ maps to a (distinct) set of $k'$ vertices in $H$, the set $\{u_{v,1},\hdots,u_{v,k'}\}$, which we will henceforth denote as $S_v$. Further, each edge $(v,w)$ in $G$ maps to a hyperedge in $H$ which is given by $S_v\cup S_w$.  It is clear from the construction that every vertex of $H$ has degree $\Delta$ (since $G$ is a $\Delta$-regular graph) and, further, that every hyperedge of $H$ has arity $2k'\geq k$. Also, note that $|U|=k'|V|$ and $|\mathcal{F}|=|E|$.

We complete the proof by showing \eqref{eq:2c2onnection}. To do this, we will map independent sets of the hypergraph $H$ to independent sets of the graph $G$ as follows. Let $I_H$ be an independent set of $H$. Define $I_G$ to be the subset of vertices of $G$ such that $v\in I_G$ iff $S_v\cap I_H=S_v$. It is immediate that $I_G$ is an independent set of $G$. In fact, it is not hard to see that for every independent set $I_G$ of $G$ there are exactly $(2^{k'}-1)^{n-|I_G|}$ independent sets of $H$ that map to $I_G$. From this, \eqref{eq:2c2onnection} follows, thus completing the proof.
\end{proof}

\begin{corollary}\label{cor:inapprox1}
Let $k=6$, $\Delta=22$. It is $\mathrm{NP}$-hard to approximate $\HyperIndSet$, even within an exponential factor.
\end{corollary}
\begin{proof}
Just plug the values of $k,\Delta$ to check that the inequality $2^{\left\lceil k/2\right\rceil}-1<\frac{(\Delta-2)^{\Delta}}{(\Delta-1)^{\Delta-1}}$ holds. Then, apply Theorem~\ref{thm:inapprox}.
\end{proof}

The following corollary is a crude estimate of the range of $\Delta$ in which $\HyperIndSet$ is hard to approximate (by applying Theorem~\ref{thm:inapprox}).
\begin{corollary}\label{cor:inapprox2}
Let $k\geq 2$. For all  integer $\Delta\geq 5 \cdot 2^{k/2}$, it is $\mathrm{NP}$-hard to approximate $\HyperIndSet$, even within an exponential factor.
\end{corollary}
\begin{proof}
For $\Delta\geq 5 \cdot 2^{k/2}$, we have that
\[\frac{(\Delta-2)^{\Delta}}{(\Delta-1)^{\Delta-1}}=(\Delta-2)\Big(1-\frac{1}{\Delta-1}\Big)^{\Delta-1}\geq \bigg(\frac{4}{5}\bigg)^5(\Delta-2)>\sqrt{2}\cdot 2^{k/2}-1\geq 2^{\left\lceil k/2\right\rceil}-1.\]
The first inequality follows from the fact that $\Big(1-\frac{1}{\Delta-1}\Big)^{\Delta-1}$ is an increasing function of $\Delta$ and the (trivial) absolute bound $\Delta\geq 6$ (since $k\geq 2$). The second inequality follows from the fact that $\Delta\geq 5 \cdot 2^{k/2}$ and $k\geq 2$. Finally, the last inequality is trivial.

Now apply Theorem~\ref{thm:inapprox}.
\end{proof}

\subsection{Counting dominating sets in graphs}\label{sec:domsets}

In this section, we prove inapproximability results for the problem of counting dominating sets in graphs of maximum degree $\Delta$. In contrast to Corollary~\ref{cor:domsetfptas} where we showed  algorithmic results for $\Delta$-regular graphs, here we consider graphs which are not necessarily regular but only have bounded maximum degree. Formally, we are interested in the following problem.

\prob{ $\DomSet$.}
{A graph $G$ with maximum degree at most $\Delta$.}
{The number of dominating sets in $G$.}

For unbounded degree graphs, it was shown by Goldberg, Gysel and Lapinskas \cite[Theorem 4]{GGL} that it is \#SAT-hard to approximate the number of dominating sets. We refine this result in a bounded degree setting. More precisely, we show the following.
\begin{theorem}\label{thm:domsethardness}
For all integers $\Delta\geq 18$, it is $\mathrm{NP}$-hard to approximate $\DomSet$, even within an exponential factor.
\end{theorem}

To prove Theorem~\ref{thm:domsethardness}, we will utilise inapproximability results for the partition function of antiferromagnetic 2-spin system on graphs. We give a quick overview of the relevant definitions and results, following \cite{GG,LLY}. A 2-spin system on a graph is specified by three parameters $\beta,\gamma\geq 0$ and $\lambda>0$. For a graph $G=(V,E)$, configurations of the system are all possible assignments $\sigma:V\rightarrow \{0,1\}$ and the partition function is given by
\[Z_{G}(\beta,\gamma,\lambda)=\sum_{\sigma:V\rightarrow\{0,1\}}\lambda^{|\sigma^{-1}(0)|}\prod_{(u,v)\in E}\beta^{\mathbf{1}\{\sigma(u)=\sigma(v)=0\}}\gamma^{\mathbf{1}\{\sigma(u)=\sigma(v)=1\}},\]
with the convention that $0^0\equiv 1$ when one of the parameters $\beta,\gamma$ is equal to zero. The case $\beta=\gamma$ corresponds to the Ising model, while the case $\beta=0$ and $\gamma=1$ corresponds to the hard-core model (which we already encountered in Section~\ref{sec:hardness}).

The 2-spin system with parameters $\beta,\gamma, \lambda$ is called \textit{antiferromagnetic} if $\beta \gamma<1$.  In \cite{SlySun}, the following inapproximability result was shown for antiferromagnetic 2-spin systems on $\Delta$-regular graphs.\footnote{\cite[Theorems 2 \& 3]{SlySun} are about the hard-core and the antiferromagnetic Ising model on $\Delta$-regular graphs. It is standard to derive from those Theorem~\ref{thm:ising12} (which applies to general antiferromagnetic 2-spin systems), since it is well-known (see, e.g., \cite{SlySun}) that antiferromagnetic 2-spin systems on $\Delta$-regular graphs can be expressed in terms of either the Ising model  or the hard-core model. The detailed derivation can be found in \cite[Corollary 21]{GG}.}

\begin{theorem}[{\cite[Theorems 2 \& 3]{SlySun}}]\label{thm:ising12}
Let $\beta,\gamma\geq 0$ with $\beta \gamma<1$, $\gamma>0$, $\lambda>0$ and $\Delta\geq 3$. If the 2-spin system specified by the parameters $\beta,\gamma,\lambda$ is in the non-uniqueness regime of the infinite $\Delta$-regular tree, then
there is a $c>1$ such that  it is $\mathrm{NP}$-hard to approximate $Z_{\beta,\gamma,\lambda;G}$ within  a factor of $c^n$ on the class of
 $\Delta$-regular graphs $G$.
\end{theorem}

To apply Theorem~\ref{thm:ising12}, we will need the following characterisation of the uniqueness regime on the infinite  $\Delta$-regular tree (see, e.g., \cite[Lemma 21]{LLY} or \cite[Section 3]{GG} for more details). For a 2-spin system with parameters $\beta,\gamma,\lambda$, non-uniqueness on the infinite $\Delta$-regular tree holds iff the system of equations
\begin{equation}\label{eq:wed456456}
x=\lambda\Big(\frac{\beta y+1}{y+\gamma}\Big)^{\Delta-1},\quad y=\lambda\Big(\frac{\beta x+1}{x+\gamma}\Big)^{\Delta-1}
\end{equation}
has multiple (i.e., more than one) positive solutions $(x,y)$.

We are now ready to give the proof of Theorem~\ref{thm:domsethardness}.

\begin{proof}[Proof of Theorem~\ref{thm:domsethardness}]
Let $\Delta:=17$, $\Delta':=\Delta+1=18$. Consider the 2-spin system with
\[\lambda=1/2,\beta=1/2, \gamma=1.\]
In Section~\ref{sec:domsethardness}, we use Mathematica to find that \eqref{eq:wed456456} has multiple  positive solutions $(x,y)$. Thus, we have that the 2-spin system is in the non-uniqueness regime of the infinite $\Delta$-regular tree and hence Theorem~\ref{thm:ising12} applies.

Let $G=(V,E)$ be a $\Delta$-regular graph for which we want to compute $Z_{G}(\beta,\gamma,\lambda)$. Set $n:=|V|$ and $m:=|E|$. We will construct a graph $G'=(V',E')$ with maximum degree $\Delta'=18$ such that $|V'|=2n+m$, $|E'|=|2m+n|$ and
\begin{equation}\label{eq:2c2onnection2345}
\#\mathrm{DomSets}(G')=2^n 2^m\, Z_G(\beta,\gamma,\lambda),
\end{equation}
where $\#\mathrm{DomSets}(G')$ denotes the number of dominating sets of $G'$. Note that the size of $G'$ is larger than the size of $G$ only by a constant factor. It thus follows that if we could approximate $\#\mathsf{DomSet}(\Delta')$ within an arbitrarily small exponential factor, we could also approximate $Z_G(\beta,\gamma,\lambda)$ within an (arbitrarily small) exponential factor for all $\Delta$-regular graphs $G$, contradicting Theorem~\ref{thm:ising12}.

It remains to construct the graph $G'=(V',E')$. The graph $G'$ is obtained as follows from $G$. Denote $V=\{v_1,\hdots,v_n\}$ and $E=\{e_1,\hdots, e_m\}$. For each vertex $v_i\in V$, add a new vertex $u_i$ and connect it to $v_i$. Further, for each edge $e_t=(v_i,v_j)$ add a new vertex $w_t$, connect it to both $v_i$ and $v_j$ and delete the edge $e_t$. In particular, we have that
\begin{align*}
V'&=\{v_1,\hdots,v_n\}\bigcup \{u_1,\hdots,u_n\}\bigcup \{w_1,\hdots,w_m\},\\
E'&=\{(v_1,u_1),\hdots, (v_n,u_n)\} \bigcup \bigcup_{e_t=(v_i,v_j)\in E}\{(v_i,w_t), (w_t,v_j)\}.
\end{align*}
It is clear from the construction that every vertex of $H$ has maximum degree at most $\Delta'=\Delta+1$ (since $G$ is a $\Delta$-regular graph, in $G'$ each of $v_1,\hdots, v_n$ has degree $\Delta'$, each of $u_1,\hdots,u_n$ has degree 1 and each of $w_1,\hdots,w_m$ has degree 2). Further, it is clear that $|V'|=2n+m$ and $|E'|=2m+n$.

We complete the proof by showing \eqref{eq:2c2onnection2345}. To do this, we will map dominating sets of $G'$ to configurations $\sigma:V\rightarrow \{0,1\}$. In particular, let $S$ be a dominating set of $G'$. For a vertex $v\in V$, we set $\sigma(v)=1$ iff $v\in S$. Then, it remains to observe that for every $\sigma:V\rightarrow \{0,1\}$, there are exactly
\begin{equation}\label{eq:rfvgtbyhnjum}
2^{|\sigma^{-1}(1)|}\prod_{e=(u,v)\in E} 2^{\mathbf{1}\{\sigma(u)=1\, \vee\, \sigma(v)=1\}}=2^{n}2^m\,(1/2)^{|\sigma^{-1}(0)|}\prod_{e=(u,v)\in E} (1/2)^{\mathbf{1}\{\sigma(u)=\sigma(v)=0\}}
\end{equation}
dominating sets of the graph $G'$ that map to $\sigma$. To see this, fix $\sigma:V\rightarrow \{0,1\}$. We will consider the possibilities for a dominating set $S$ of $G'$ that maps to $\sigma$. For every vertex $v_i\in V$, if $\sigma(v_i)=1$, we have that $v_i\in S$ and hence the vertex $u_i$ can either belong to $S$ or not belong to $S$ (2 choices). In contrast, if $\sigma(v_i)=0$, we have that $v_i\notin S$ and hence the vertex $u_i$ must  belong to $S$ in order to be dominated since its only neighbour is vertex $v_i$ (1 choice). Similarly, for every $e_t=(v_i,v_j)\in E$, if $\sigma(v_i)=\sigma(v_j)=0$, then  $v_i,v_j\notin S$ and hence the vertex $w_t$ must  belong to $S$ in order to be dominated since its only neighbours are the vertices $v_i,v_j$ (1 choice). In all other cases, $w_t$ can either belong to $S$ or not belong to $S$ (2 choices). This justifies \eqref{eq:rfvgtbyhnjum}, thus completing the justification of \eqref{eq:2c2onnection2345}.

Note that the purpose of the ``bristle'' vertices $u_1,\ldots,u_n$ is to make the interactions of the edges of~$G'$
independent of each other. If $G$ has edges $e_t=(v_i,v_j)$ and $e_{t'}=(v_i,v_{j'})$ and
amongst $\{v_i,w_t,v_j\}$, only $v_j\in S$
then $u_i$ has to be in~$S$, so $w_{t'}$ can either be in $S$ or not, independently of $w_t$ and $v_j$.

This concludes the proof.
\end{proof}

\section{The Uniqueness Threshold on the Infinite Hypertree}\label{sec:uniqueness}
We denote by $\Tree$ the infinite $(\Delta-1)$-ary $k$-uniform hypertree with root vertex $\rho$. Also, for $n=0,1,2,\hdots$, denote by $\Tree(n)$ the subtree of $\Tree$ obtained by the first $n$ levels, i.e., $\Tree(n)$ is the tree induced by the set of vertices at distance $\leq n$ from $\rho$ in $\Tree$. We denote by $V_n$ the vertex set of $\Tree(n)$ and by $L_n$ the leaves of the tree, i.e., vertices with degree 1 in $\Tree(n)$.

Denote by $\mu_n$ the Gibbs distribution of the independent set model on $\Tree(n)$ (see Section~\ref{sec:SSM}). For a configuration $\sigma:V_n\rightarrow \{0,1\}$, we denote by $\sigma_{L_n}$ the restriction of $\sigma$ to the set $L_n$ and by $\sigma_\rho$ the spin of the root $\rho$.
\begin{definition}
Let $k\geq 2,\Delta\geq 2$ be integers. The independent set model has uniqueness on $\Tree$ iff
\begin{equation}\label{eq:uniq1}
\limsup_{n\rightarrow\infty}\max_{\eta,\eta':L_n\rightarrow \{0,1\}}\big|\mu_n(\sigma_\rho=1\mid \sigma_{L_n}=\eta)-\mu_n(\sigma_\rho=1\mid \sigma_{L_n}=\eta')\big|=0.
\end{equation}
\end{definition}

We will use $\sigma_{L_n}=1$ to denote that, in the configuration $\sigma$, all vertices in $L_n$ are assigned the spin 1. For $n=0,1,2,\hdots$, define
\begin{equation}
p_{n}=\mu_n(\sigma_\rho=1\mid \sigma_{L_n}=1)
\end{equation}
When $\sigma_\rho=1$, in each of the $\Delta-1$ hyperedges that include $\rho$, at least one of the $k-1$ vertices (other than $\rho$) must have spin 0. When $\sigma_\rho=0$, any configuration on the neighbours of $\rho$ is allowed. By considering the (normalised) weight of such configurations on $\Tree(n+1)\backslash \rho$, it is not hard to see that the sequence $p_n$ satisfies the following recursion for every integer $n\geq 0$:
\begin{equation}\label{eq:pppmrecursion}
p_{n+1}=f(p_{n}), \mbox{ where } f(z):=\frac{(1-z^{k-1})^{\Delta-1}}{1+(1-z^{k-1})^{\Delta-1}}.
\end{equation}
For any configuration $\eta:L_n\rightarrow \{0,1\}$, we will see that  $\mu_n(\sigma_\rho=1\mid \sigma_{L_n}=\eta)$ is sandwiched between $p_n$ and $p_{n+1}$. This yields the following.
\begin{lemma}\label{lem:uniq2}
Let $k\geq 2,\Delta\geq 2$ be integers. The independent set model has uniqueness on $\Tree$ iff
\begin{equation}\label{eq:uniq2}
\limsup_{n\rightarrow\infty}|p_{n+1}-p_n|=0.
\end{equation}
\end{lemma}
\begin{proof}
Let $n$ be a non-negative integer and let $\eta$ be an arbitrary configuration on $L_n$, i.e., $\eta:L_n\rightarrow \{0,1\}$. We will show that
\begin{equation}\label{eq:sandwich}
\begin{gathered}
p_{n+1}\leq \mu_n(\sigma_\rho=1\mid \sigma_{L_n}=\eta)\leq p_n \mbox{ for even integers } n,\\
p_n\leq \mu_n(\sigma_\rho=1\mid \sigma_{L_n}=\eta)\leq p_{n+1} \mbox{ for odd integers } n.\\
\end{gathered}
\end{equation}
Let us first conclude the lemma assuming \eqref{eq:sandwich}. From \eqref{eq:sandwich}, we obtain that for all $n$, it holds that
\[\max_{\eta,\eta':L_n\rightarrow \{0,1\}}\big|\mu_n(\sigma_\rho=1\mid \sigma_{L_n}=\eta)-\mu_n(\sigma_\rho=1\mid \sigma_{L_n}=\eta')\big|= |p_{n+1}-p_n|.\]
It follows that the conditions in \eqref{eq:uniq1} and \eqref{eq:uniq2} are equivalent, which yields the statement in the lemma.

We next show \eqref{eq:sandwich}. The proof is by induction on $n$. The claim is trivial for $n=0$ since $p_0=1$ and $p_1=0$. So assume that the claim holds for all non-negative integers less than $n$, we will show it for $n$.

Set $d:=\Delta-1$. Let $e_1,\hdots,e_d$ be the $d$ hyperedges containing $\rho$ and for $i\in[d]$ denote by $v_{i,1},\hdots,v_{i,k-1}$ the vertices in $e_i$ other than $\rho$, i.e.,
\[e_1=\{\rho,v_{1,1},\hdots,v_{1,k-1}\}, \hdots, e_d=\{\rho,v_{d,1},\hdots,v_{d,k}\},\]
For $i\in[d]$ and $j\in [k-1]$, let $T_{i,j}$ be the subtree of $\Tree(n)$ rooted at $v_{i,j}$. Denote by $S_{i,j}$ the leaves of $T_{i,j}$ and by $\eta_{i,j}$ the restriction of $\eta$ on $S_{i,j}$. Let $\mu_{T_{i,j}}$ be the Gibbs distribution of $T_{i,j}$ in the independent set model. Note that $\cup_{i\in[d],j\in [k-1]}S_{i,j}=L_n$.  Finally, let
\begin{equation*}
\begin{aligned}
q_{i,j}&:=\mu_{T_{i,j}}(\sigma_{v_{i,j}}=1\mid \sigma_{S_{i,j}}=\eta_{i,j}),\\
q&:=\mu_{n}(\sigma_{\rho}=1\mid \sigma_{L_n}=\eta).
\end{aligned}
\end{equation*}
It is simple to see that
\begin{equation}\label{eq:qrecurse}
q=\frac{\prod^{d}_{i=1}(1-\prod^{k-1}_{j=1}q_{i,j})}{1+\prod^{d}_{i=1}(1-\prod^{k-1}_{j=1}q_{i,j})}, \mbox{ or equivalently that } \frac{q}{1-q}=\prod^{d}_{i=1}\bigg(1-\prod^{k-1}_{j=1}q_{i,j}\bigg).
\end{equation}
For $i\in[d]$ and $j\in[k-1]$, note that $T_{i,j}$ is isomorphic to $\Tree(n-1)$ and hence we can use the induction hypothesis to bound $q_{i,j}$. Let us consider first the case where $n$ is odd. Then, we have that
\begin{equation}\label{eq:qijbounds}
p_n\leq q_{i,j}\leq p_{n-1}.
\end{equation}
It follows from \eqref{eq:qrecurse} and \eqref{eq:qijbounds} that
\begin{equation}\label{eq:qij2bounds}
\big(1-(p_{n-1})^{k-1}\big)^d\leq \frac{q}{1-q}\leq \big(1-(p_{n})^{k-1}\big)^d, \mbox{ so that } p_{n}\leq q\leq p_{n+1},
\end{equation}
where to derive the last inequality we used that the sequence $p_n$ satisfies the recursion in \eqref{eq:pppmrecursion} and that the function $\frac{x}{1-x}$ is increasing in $x$. The proof for odd $n$ is completely analogous, modulo that the inequalities in \eqref{eq:qijbounds} and \eqref{eq:qij2bounds} hold in the opposite direction.

This concludes the induction step and hence the proof of \eqref{eq:sandwich}. The proof of the lemma is thus complete.
\end{proof}
\begin{lemma}\label{lem:xfxunique}
Let $k\geq 2$, $\Delta\geq 2$ be integers. Let $f(z)=\frac{(1-z^{k-1})^{\Delta-1}}{1+(1-z^{k-1})^{\Delta-1}}$ be as in \eqref{eq:pppmrecursion}. The function $f$ is strictly decreasing in the interval $[0,1]$. Also, there is unique $x\in[0,1]$ such that
\begin{equation}\label{eq:uniquenessequation}
f(x)=x,
\end{equation}
which further satisfies $|f'(x)|=\frac{(\Delta-1)(k-1)x^{k-1}(1-x)}{1-x^{k-1}}$.

Finally, if $|f'(x)|< 1$, the equation
\begin{equation}\label{eq:uniqueness2equation}
f(f(z))=z \mbox{ for } z\in[0,1]
\end{equation}
is \emph{uniquely} satisfied  by $z=x$.
\end{lemma}
\begin{proof}
To see that the function $f$ is decreasing, we calculate
\begin{equation}\label{eq:derivative123off}
f'(z)=-\frac{(\Delta-1)(k-1)\big(1-z^{k-1}\big)^{\Delta-2}z^{k-2}}{\big(1+(1-z^{k-1})^{\Delta-1}\big)^2},
\end{equation}
which clearly shows that $f$ is decreasing for $z\in[0,1]$. Note that $f'(z)=0$ iff $z=0$ or $z=1$, so in fact $f$ is strictly decreasing over the interval [0,1].

We next show the second part of  the lemma. We can rewrite $z=f(z)$ as
\[g(z)=0 \mbox{ where } g(z):=z-(1-z)(1-z^{k-1})^{\Delta-1}.\]
For the function $g$, we have that $g(0)=-1$, $g(1)=1$ and $g$ is continuous on $[0,1]$. It thus follows that there exists $x$ such that $g(x)=0$ which implies that  $f(x)=x$. We next prove that $x$ is unique, i.e., for all $z\neq x$ it holds that $g(z)\neq 0$.   For this, it suffices to show that $g$ is increasing on $[0,1]$ or that $g'(z)>0$ for all $z\in[0,1]$. We calculate
\[g'(z)=1+(1-z^{k-1})^{\Delta-1}+(k-1)(\Delta-1)(1-z)(1-z^{k-1})^{\Delta-2}z^{k-2},\]
which clearly shows that $g'(z)\geq 1$ for all $z\in[0,1]$. Finally, to see the expression for $|f'(x)|$, just use \eqref{eq:derivative123off} and use that
\[x=\frac{(1-x^{k-1})^{\Delta-1}}{1+(1-x^{k-1})^{\Delta-1}},\quad 1-x=\frac{1}{1+(1-x^{k-1})^{\Delta-1}}\]
to simplify. This proves the second assertion in the lemma.

Finally, we  show the last part of the lemma. Namely, suppose that $|f'(x)|< 1$. Consider the function
\[h(z):=f(f(z))-z, \mbox{ for $z\in [0,1]$}.\]
Clearly, we have that $h(x)=0$. By the assumption $|f'(x)|< 1$, we have that $h'(x)<0$ and hence for small $\epsilon>0$ it holds that $h(x-\epsilon)>0$ and $h(x+\epsilon)<0$. Note also that $h(0)>0$ and $h(1)<0$.

For the sake of contradiction, assume that $h$ has a zero other than $x$, say at $z=\rho_1$ where $\rho_1\neq x$. Let $\rho_2:=f(\rho_1)$ and note that $f(\rho_2)=\rho_1$. Observe also that $h$ has another zero at $z=\rho_2$. Also, using that $f$ is decreasing and $f(x)=x$, we have $\rho_2\neq \rho_1,x$. In fact, we have that $x$ is between $\rho_1$ and $\rho_2$. Without loss of generality we may thus assume that $\rho_1>x>\rho_2$ (otherwise we may swap $\rho_1$ and $\rho_2$).

We may also assume that $h'(\rho_1)\geq 0$. Otherwise, we claim that there exists $\rho\in(x,\rho_1)$ such that $h(\rho)=0$ and $h'(\rho)\geq 0$, so that if $h'(\rho_1)< 0$  we could swap our focus from  $\rho_1$ to $\rho$ (the role of $\rho_2=f(\rho_1)$ would be played by $f(\rho)$). To see the claim, assume that  $h'(\rho_1)<0$ and let $\epsilon>0$ be sufficiently small. Then, it holds that $h(\rho_1-\epsilon)>0$. Recall also that $h(x+\epsilon)<0$. It follows that there must exist a \emph{crossing} of $h$ in the interval $(x,\rho_1)$, i.e., a real number $\rho\in (x,\rho_1)$ such that $h(\rho)=0$ and, for all sufficiently small $\epsilon'$, it holds that $h(\rho+\epsilon')<0$ and $h(\rho-\epsilon')>0$. It must thus be the case that $h'(\rho)\geq 0$, as claimed.

Thus, if $|f'(x)|<1$, we have concluded the existence of $\rho_1,\rho_2$ such that $1>\rho_1>\rho_2>0$, $\rho_1=f(\rho_2)$, $\rho_2=f(\rho_1)$ and $h'(\rho_1)\geq 0$. We obtain a contradiction by showing that $h'(\rho_1)<0$.  We have
\begin{equation}\label{eq:hrho1der}
h'(\rho_1)=f'(\rho_1)f'(\rho_2)-1=\frac{(\Delta-1)^2(k-1)^2\rho_1^{k-2}\rho_2^{k-2}\big(1-\rho_1^{k-1}\big)^{\Delta-2}\big(1-\rho_2^{k-1}\big)^{\Delta-2}}{\big(1+(1-\rho_1^{k-1})^{\Delta-1}\big)^2\big(1+(1-\rho_2^{k-1})^{\Delta-1}\big)^2}-1.
\end{equation}
For $k=2,\Delta=2$, \eqref{eq:hrho1der} is trivial, so henceforth we may assume that at least one of $k,\Delta$ is greater 
than or equal to~$3$.

We can rewrite $\rho_1=f(\rho_2)$ and $\rho_2=f(\rho_1)$ as
\begin{equation}\label{eq:rho1rho2fix}
1-\rho_1=\frac{1}{1+\big(1-\rho_2^{k-1}\big)^{\Delta-1}}, \quad 1-\rho_2=\frac{1}{1+\big(1-\rho_1^{k-1}\big)^{\Delta-1}}.
\end{equation}
Subtracting the equations in \eqref{eq:rho1rho2fix} gives
\begin{equation}\label{eq:rho1rho2dif}
\rho_1-\rho_2=\frac{\big(1-\rho_2^{k-1}\big)^{\Delta-1}-\big(1-\rho_1^{k-1}\big)^{\Delta-1}}{\big(1+(1-\rho_1^{k-1})^{\Delta-1}\big)\big(1+(1-\rho_2^{k-1})^{\Delta-1}\big)}.
\end{equation}
Let $E:=\big(1-\rho_2^{k-1}\big)^{\Delta-1}-\big(1-\rho_1^{k-1}\big)^{\Delta-1}$ and assume for now that $k\geq 3$, $\Delta\geq 3$. We use the inequality $x^d-y^d> d(x-y)(xy)^{(d-1)/2}$, which holds for all  $d\geq 2$ and real numbers $x> y>0$ (see, for example, \cite[Claim 35]{GSV:arxiv}), to obtain the following bound for $E$:
\begin{align}
E&> (\Delta-1)(\rho_1^{k-1}-\rho_2^{k-1})(1-\rho_1^{k-1})^{(\Delta-2)/2}(1-\rho_2^{k-1})^{(\Delta-2)/2}\notag\\
&> (\Delta-1)(k-1)(\rho_1-\rho_2)\rho_1^{(k-2)/2}\rho_2^{(k-2)/2}(1-\rho_1^{k-1})^{(\Delta-2)/2}(1-\rho_2^{k-1})^{(\Delta-2)/2}.\label{eq:rho1rho2done}
\end{align}
Note that the strict inequality in \eqref{eq:rho1rho2done} holds even in the cases where $k=2$, $\Delta\geq 3$ or $\Delta=2,k\geq 3$.

Combining \eqref{eq:rho1rho2dif} and \eqref{eq:rho1rho2done}, we obtain that
\[1>\frac{(\Delta-1)(k-1)\rho_1^{(k-2)/2}\rho_2^{(k-2)/2}(1-\rho_1^{k-1})^{(\Delta-2)/2}(1-\rho_2^{k-1})^{(\Delta-2)/2}}{(1+(1-\rho_1^{k-1})^{\Delta-1}\big)\big(1+(1-\rho_2^{k-1})^{\Delta-1}\big)}.\]
Squaring the last inequality and using \eqref{eq:hrho1der}, we obtain $h'(\rho_1)<0$, as desired.

This concludes the proof of Lemma~\ref{lem:xfxunique}.
\end{proof}

\begin{lemma}\label{lem:technical1criterion}
Let $k\geq 2,\Delta\geq 2$ be integers. Let $f(z)=\frac{(1-z^{k-1})^{\Delta-1}}{1+(1-z^{k-1})^{\Delta-1}}$ be as in \eqref{eq:pppmrecursion} and  $x$ be as in Lemma~\ref{lem:xfxunique}.

If $|f'(x)|<1$, the independent set model has uniqueness on $\Tree$.

If $|f'(x)|>1$, the independent set model has non-uniqueness on $\Tree$.
\end{lemma}
\begin{proof}
Recall the sequence $p_n$ defined in \eqref{eq:pppmrecursion}. Let $p^+_n=p_{2n}$ and $p^-_n=p_{2n+1}$. As a consequence of the fact that $f$ is decreasing (cf.  Lemma~\ref{lem:xfxunique}) and $p^+_0=1$, $p^-_0=0$, we have that
\begin{equation}\label{eq:limitsexistence}
p^+_{n}\downarrow p^+,\quad  p^-_{n}\uparrow p^-
\end{equation}
where $p^+,p^-$ are real numbers in $[0,1]$. To see the existence of these limits, note that  $p^+_0=1\geq p^+_1$ and $p^-_0=0\leq p^-_1$. Since $p^{\pm}_{n+1}=f(f(p^\pm_n))$, a simple induction shows that $p^+_{n}$ is a decreasing sequence  and $p^-_n$ increasing. Since both sequences are bounded, we obtain the existence of the limits in \eqref{eq:limitsexistence}. For later use, we remark here that the continuity of $f$ and the recursions $p^{\pm}_{n+1}=f(p^\mp_n)$ and $p^{\pm}_{n+1}=f(f(p^\pm_n))$ imply that $p^+,p^-$  satisfy the equalities
\begin{gather}
p^+=f(p^-) \mbox{ and } p^-=f(p^+), \label{eq:semitranslation1}\\
p^+=f(f(p^+))  \mbox{ and } p^-=f(f(p^-)).\label{eq:semitranslation2}
\end{gather}

As a consequence of the existence of the limits in \eqref{eq:limitsexistence}, we can conclude  that the condition $\limsup_{n\rightarrow \infty}|p_{n+1}-p_n|=0$ is equivalent to
\begin{equation}\label{eq:4rf5tg}
p^+=p^-.
\end{equation}

We are now ready to show the equivalence in the lemma. For the first part, assume that $|f'(x)|< 1$. To show that uniqueness holds on $\Tree$, it suffices to show that \eqref{eq:4rf5tg} holds, i.e., $p^+=p^-$. We have that $p^+,p^-$ satisfy \eqref{eq:semitranslation2}. From the second part of Lemma~\ref{lem:xfxunique} and the assumption  $|f'(x)|< 1$, we thus obtain that $p^+=x=p^-$, as wanted.

For the second part, it suffices to show the contrapositive. So, assume that uniqueness holds on $\Tree$, we will show that $|f'(x)|\leq 1$ where recall that $x$ is specified by the relation $x=f(x)$ (cf. the second part of Lemma~\ref{lem:xfxunique}). From Lemma~\ref{lem:uniq2} we have that $\limsup_{n\rightarrow \infty}|p_{n+1}-p_{n}|=0$ and hence by our previous arguments, we obtain that \eqref{eq:4rf5tg} also holds. From \eqref{eq:semitranslation1} and using the uniqueness of $x$ (cf. Lemma~\ref{lem:xfxunique}), we obtain that the common value of $p^+$ and  $p^-$ is $x$ and thus $p_n\rightarrow x$.

 By the Mean Value theorem we have that there exists $\xi_n$ between $p_{n+1}$ and $p_n$ such that $|f'(\xi_n)|=|p_{n+1}-p_{n}|/|p_{n}-p_{n-1}|$. From $p_n\rightarrow x$, we also have that $\xi_n\rightarrow x$. Since $p_n\rightarrow x$, we have that for infinitely many $n$ it holds that $|p_{n+1}-p_n|\leq |p_{n}-p_{n-1}|$. For all such $n$, it holds that $|f'(\xi_n)|\leq 1$. Using that $f'$ is continuous and  $\xi_n\rightarrow x$, we obtain that $|f'(x)|\leq 1$ as desired.

This concludes the proof of Lemma~\ref{lem:technical1criterion}.
\end{proof}

The following lemma establishes the intuitive fact that, for the independent set model, uniqueness on $\Tree$ is a monotone property with respect to  $\Delta$.
\begin{lemma}\label{lem:uniquenessmonotonicity}
Let $k\geq 2$ be an integer. There exists $\Delta_c(k)\geq 3$ such that the following holds for all integer $\Delta\geq 2$. The independent set model has uniqueness on $\Tree$ whenever $\Delta<\Delta_c(k)$ and non-uniqueness whenever $\Delta>\Delta_c(k)$.
\end{lemma}
\begin{proof}
Fix an integer $k\geq 2$. For integer $\Delta\geq 2$, parameterise the function $f(z)$ in \eqref{eq:pppmrecursion} by $\Delta$, i.e., set
\[f_{\Delta}(z)=\frac{(1-z^{k-1})^{\Delta-1}}{1+(1-z^{k-1})^{\Delta-1}}.\]
Let $x_\Delta$ be the unique solution of $z=f_{\Delta}(z)$ (cf. the second part of Lemma~\ref{lem:xfxunique}). Recall also from Lemma~\ref{lem:xfxunique} that
\begin{equation}\label{deratDelta}
|f_{\Delta}'(x_\Delta)|=(\Delta-1)(k-1)h(x_\Delta), \mbox{ where } h(z):=\frac{z^{k-1}(1-z)}{1-z^{k-1}}.
\end{equation}
We will show that $|f_{\Delta}'(x_\Delta)|$ is a (strictly) increasing function of $\Delta$, i.e., $|f_{\Delta+1}'(x_{\Delta+1})|> |f_{\Delta}'(x_\Delta)|$. This follows immediately by multiplying the following inequalities:
\begin{gather}
\Delta (x_{\Delta+1})^{k-1}> (\Delta-1)(x_\Delta)^{k-1},\label{eq:45hard45}\\
\frac{(x_{\Delta})^{k-1}-1}{x_{\Delta}-1}\geq  \frac{(x_{\Delta+1})^{k-1}-1}{x_{\Delta+1}-1}\label{eq:44hard44}.
\end{gather}
Thus, to show the desired monotonicity, we only need to argue for the validity of \eqref{eq:45hard45} and \eqref{eq:44hard44}. We will need the following simple fact:
\begin{equation}\label{eq:mgbbgm31}
\mbox{for all $z\in (0,1)$ and all $\Delta\geq 2$, it holds that $f_{\Delta+1}(z)< f_{\Delta}(z)$}.
\end{equation}
To see \eqref{eq:mgbbgm31}, fix $z\in(0,1)$ and let $\Delta\geq 2$. Then we have
\begin{equation*}
\frac{f_{\Delta+1}(z)}{1-f_{\Delta+1}(z)}=(1-z^{k-1})^{\Delta}<(1-z^{k-1})^{\Delta-1}=\frac{f_{\Delta}(z)}{1-f_{\Delta}(z)},
\end{equation*}
and thus \eqref{eq:mgbbgm31} follows.

We next proceed to showing \eqref{eq:45hard45} and \eqref{eq:44hard44}. The crucial step will be to show that $x_{\Delta}$ is a decreasing function of $\Delta$, i.e., $x_{\Delta+1}< x_{\Delta}$. Suppose for the sake of contradiction that $x_{\Delta}\leq x_{\Delta+1}$ for some $\Delta$. Using \eqref{eq:mgbbgm31} for $z=x_{\Delta}$ and the fact that $f_{\Delta+1}(z)$ is decreasing (by Lemma~\ref{lem:xfxunique}), we have
\[x_{\Delta}=f_\Delta(x_{\Delta})> f_{\Delta+1}(x_{\Delta})\geq f_{\Delta+1}(x_{\Delta+1})=x_{\Delta+1},\]
contradiction. It thus follows that $x_{\Delta}>x_{\Delta+1}$. From this, it is simple to conclude \eqref{eq:44hard44}:
\[\frac{(x_{\Delta})^{k-1}-1}{x_{\Delta}-1}=\sum^{k-2}_{j=0}(x_{\Delta})^j\geq \sum^{k-2}_{j=0}(x_{\Delta+1})^j=\frac{(x_{\Delta+1})^{k-1}-1}{x_{\Delta+1}-1}.\]

We next show the slightly harder \eqref{eq:45hard45}. From $x_{\Delta}>x_{\Delta+1}$ and the facts  $x_{\Delta}=f_{\Delta}(x_{\Delta})$ and $x_{\Delta+1}=f_{\Delta+1}(x_{\Delta+1})$, it follows that $f_{\Delta}(x_\Delta)>f_{\Delta+1}(x_{\Delta+1})$. This in turn  yields
\[(1-x_\Delta^{k-1})^{\Delta-1}=\frac{f_{\Delta}(x_{\Delta})}{1-f_{\Delta}(x_{\Delta})}>\frac{f_{\Delta+1}(x_{\Delta+1})}{1-f_{\Delta+1}(x_{\Delta+1})}=(1-x_{\Delta+1}^{k-1})^{\Delta}.\]
Using Bernoulli's inequality, we thus obtain
\[1-x_\Delta^{k-1}>(1-x_{\Delta+1}^{k-1})^{\frac{\Delta}{\Delta-1}}\geq 1-\frac{\Delta}{\Delta-1}x_{\Delta+1}^{k-1},\]
and, by rearranging,  we obtain \eqref{eq:45hard45}.

Since $|f_{\Delta}'(x_{\Delta})|$ is increasing for $\Delta\geq 2$, to complete the proof of the lemma, it suffices to show that for $\Delta=2$ it holds that $|f_{\Delta}'(x_{\Delta})|<1$ and that for some $\Delta$ it holds that $|f_{\Delta}'(x_{\Delta})|>1$.

Note that for all $z\in(0,1)$, we have $(k-1)z^{k}+1>kz^{k-1}$ (the function $(k-1)z^{k}+1-kz^{k-1}$ is decreasing for $z\in[0,1]$ and it has a zero at $z=1$). Rearranging, we obtain that $h(z)<1/(k-1)$, so that for $\Delta=2$ we have from \eqref{deratDelta} that $|f_{\Delta}'(x_\Delta)|<1$ as needed.

For large $\Delta$, since $x_{\Delta}$ is decreasing and bounded, we have that $x_{\Delta}$ converges. From $x_{\Delta}=f_{\Delta}(x_{\Delta})$, we thus obtain that $x_{\Delta}\downarrow 0$ as $\Delta\rightarrow \infty$. We claim also that $\Delta (x_{\Delta})^{k-1}\rightarrow \infty$ as $\Delta\rightarrow \infty$, from which it clearly follows that $|f_{\Delta}'(x_{\Delta})|\rightarrow \infty$ (cf. \eqref{deratDelta}) and hence $|f_{\Delta}'(x_{\Delta})|>1$ for large $\Delta$. To see the claimed limit,  assume for the sake of contradiction that there was $M>0$ such that $\Delta (x_{\Delta})^{k-1}<M$ for infinitely many $\Delta$. For all such $\Delta$, it holds that
\[x_{\Delta}=f_\Delta(x_\Delta)=\frac{(1-x^{k-1}_{\Delta})^{\Delta-1}}{1+(1-x^{k-1}_{\Delta})^{\Delta-1}}\geq \frac{1}{2}(1-x^{k-1}_{\Delta})^{\Delta-1}\geq \frac{1}{2}\Big(1-\frac{M}{\Delta}\Big)^{\Delta-1}.\]
It follows that $\limsup_{\Delta\rightarrow\infty} x_{\Delta}\geq \frac{1}{2}\exp(-M)>0$,  contradicting that $x_{\Delta}\downarrow 0$ as $\Delta\rightarrow \infty$.

This completes the proof of Lemma~\ref{lem:uniquenessmonotonicity}.
\end{proof}

\begin{corollary}\label{cor:uniqsmall}
Let $k=6$. The independent set model has uniqueness on $\Tree$ iff $\Delta\leq 28$.
\end{corollary}
\begin{proof}
We will show that for $\Delta=28$, the independent set model has uniqueness on $\Tree$ and non-uniqueness for $\Delta=29$. The result then follows from Lemma~\ref{lem:uniquenessmonotonicity}.

For $\Delta=28$, we have that $|f'(x)|<0.996$ where $x,f(x)$ are as in Lemma~\ref{lem:xfxunique}. For $\Delta=29$, we have that $|f'(x)|>1.01$. Lemma~\ref{lem:technical1criterion} thus shows that the independent set model has uniqueness for $\Delta=28$ and does not have uniqueness for $\Delta=29$, as needed.
\end{proof}
 The following lemma asserts that, asymptotically in $k$, the critical value $\Delta_c(k)$ in Lemma~\ref{lem:uniquenessmonotonicity} satisfies $\Delta_c(k)= (1+o_k(1))2^k/k$. The precise statement is as follows.
\begin{lemma}\label{lem:technical2criterion}
Let $\Delta_c(k)$ be as in Lemma~\ref{lem:uniquenessmonotonicity}. For all $\epsilon>0$, there exists an integer $k_0(\epsilon)>0$ such that the following holds. For all integer $k\geq k_0(\epsilon)$,
\[(1-\epsilon)2^k/k\leq \Delta_c(k)\leq (1+\epsilon)2^k/k.\]
\end{lemma}
\begin{proof}
Let $\epsilon>0$ and $k_0(\epsilon)\geq 2$ be a large constant (depending only on $\epsilon$). We will henceforth assume that $k$ is an integer satisfying $k\geq k_0(\epsilon)$.

Let $\Delta^+_k:=\left\lfloor  (1+\epsilon)2^k/k\right\rfloor$ and $\Delta^-_k:=\left\lceil (1- \epsilon)2^k/k\right\rceil$. Note that for all sufficiently large $k$ it holds that $\Delta^\pm_k\geq 3$. Let $f_{k,\pm}(z)$ be the function $f(z)$ in \eqref{eq:pppmrecursion} when $\Delta=\Delta^{\pm}_k$, i.e.,
\[f_{k,+}(z)=\frac{(1-z^{k-1})^{\Delta^+_k-1}}{1+{(1-z^{k-1})^{\Delta^+_k-1}}} \mbox{ and } f_{k,-}(z)=\frac{(1-z^{k-1})^{\Delta^-_k-1}}{1+{(1-z^{k-1})^{\Delta^-_k-1}}}.\]

Further, let $x^\pm_k$ be the unique solution of $x=f_{k,\pm}(x)$. We will show that for all sufficiently large $k$, it holds that
\begin{equation}\label{eq:fderivativeslargek}
|f_{k,-}'(x^{-}_k)|<1\mbox{ and }|f_{k,+}'(x^{+}_k)|>1.
\end{equation}
Together with Lemmas~\ref{lem:technical1criterion} and~\ref{lem:uniquenessmonotonicity}, \eqref{eq:fderivativeslargek} yields that $\Delta^-_k\leq \Delta_c(k)\leq \Delta^+_k$, as desired.

The key to showing \eqref{eq:fderivativeslargek} is  that $x^{\pm}_k\rightarrow 1/2$ as $k\rightarrow \infty$. Assuming this for the moment, let us conclude \eqref{eq:fderivativeslargek}. Recall from Lemma~\ref{lem:xfxunique} that
\[|f_{k,-}'(x^{-}_k)|=\frac{(\Delta^-_k-1)(k-1)(x^{-}_k)^{k-1}(1-x^{-}_k)}{1-(x^{-}_k)^{k-1}}, \quad |f_{k,+}'(x^{+}_k)|=\frac{(\Delta^+_k-1)(k-1)(x^{+}_k)^{k-1}(1-x^{+}_k)}{1-(x^{+}_k)^{k-1}}.\]
Now just using that $x^{\pm}_k\rightarrow 1/2$, we obtain that for $k\rightarrow \infty$ it holds that
\[|f_{k,-}'(x^{-}_k)|\rightarrow 1-\epsilon, \qquad |f_{k,+}'(x^{+}_k)|\rightarrow 1+\epsilon.\]
This shows that \eqref{eq:fderivativeslargek} holds for all sufficiently large $k$.

We next show that $x^{\pm}_k\rightarrow 1/2$ as $k\rightarrow \infty$. From $x^{\pm}_k=f_{k,\pm}(x^{\pm}_k)$, we obtain that
\begin{equation}\label{interxDelta}
\frac{x^{\pm}_k}{1-x^{\pm}_k}=\big(1-(x^{\pm}_k)^{k-1}\big)^{\Delta-1},
\end{equation}
from which it clearly follows that $x^{\pm}_k\leq 1/2$. Thus, it suffices to show that for every $\epsilon'>0$, for all sufficiently large $k$, it holds that $x^{\pm}_k> (1/2)-\epsilon'$.

From \eqref{interxDelta}, by taking logarithms and then the $k$-th root, we obtain
\begin{equation}\label{eq:Deltapmk}
(\Delta^\pm_k-1)^{1/k}=g(x^{\pm}_k)^{1/k}, \mbox{ with } g(z):=\frac{\ln\big(\frac{z}{1-z}\big)}{\ln\big(1-z^{k-1}\big)}.
\end{equation}
We will use that the function $g(z)$ is decreasing for $z\in (0,1/2]$, since
\[g'(z)=\frac{(k-1) z^{k-2} \ln \big(\frac{z}{1-z}\big)}{\left(1-z^{k-1}\right) \ln ^2\left(1-z^{k-1}\right)}+\frac{\left(\frac{z}{(1-z)^2}+\frac{1}{1-z}\right) (1-z)}{z \ln \left(1-z^{k-1}\right)}\leq 0.\]
We will also use that for any $z\in(0,1/2)$ it holds that
\[\lim_{k\rightarrow \infty} (g(z))^{1/k}=1/z, \mbox{ since} \lim_{k\rightarrow \infty} \Big(-\ln\Big(\frac{z}{1-z}\Big)\Big)^{1/k}\rightarrow 1 \mbox{ and } \lim_{k\rightarrow \infty} \big(-\ln(1-z^{k-1})\big)^{1/k}\rightarrow z.\]

Now, for the sake of contradiction, assume that there was $\epsilon'>0$ such that $x_k\leq (1/2)-\epsilon'$ for infinitely many $k$. For all such $k$, we would have that $g(x^{\pm}_k)\geq g((1/2)-\epsilon')$ and thus, by taking the $\limsup$ in  \eqref{eq:Deltapmk}, we obtain
\[2=\limsup_{k\rightarrow\infty}(\Delta^\pm_k-1)^{1/k}=\limsup_{k\rightarrow\infty}\big(g(x^{\pm}_k)\big)^{1/k}\geq \limsup_{k\rightarrow\infty}\big(g\big((1/2)-\epsilon'\big)\big)^{1/k}=\frac{1}{\frac{1}{2}-\epsilon'}>2,\]
contradiction. This proves that $x^{\pm}_k\rightarrow 1/2$ as $k\rightarrow \infty$, thus concluding the proof of Lemma~\ref{lem:technical2criterion}.
\end{proof}

\section{Computer Assisted Proofs}\label{sec:edc}
The code in each subsection below can be executed by copying it in a Mathematica cell.
It is safest to quit the local kernel before executing each subsection, to avoid interference between them.

\subsection{Analysis for \texorpdfstring{$k\geq \Delta\geq 200$}{k>=Delta>=200}}

\subsubsection{Mathematica Code for Proof of \texorpdfstring{\eqref{eq:psi1plusr}}{(29)}}\label{sec:psi1plusr}
\begin{verbatim}
psi = 13/10; chi = 1/2; alpha = 1 - 10^(-4);
f = 1/(psi - (1 + r)^-chi)*(r (psi - r^chi))/(alpha (1 + r));
Resolve[Exists[r, f > 42/100 && 0 <= r <= 1]]
\end{verbatim}

\subsubsection{Mathematica Code for Proof of Lemma~\ref{lem:con23con45}}\label{sec:con23con45}
\begin{verbatim}
psi = 13/10; chi = 1/2;
f = (psi - r^chi)/(1 + r);
hatf = f /. {r -> Exp[t]/(1 - Exp[t])};
SECDER = (FullSimplify[D[hatf, {t, 2}]]) /. {Exp[t]->T,Exp[-t] -> 1/T};
Resolve[Exists[T, 0 <= T <= 1/2 && SECDER > 0]]
\end{verbatim}

\subsubsection{Mathematica Code for Proof of Lemma~\ref{lem:ploqaz1}}\label{sec:ploqaz1}
\noindent \verb|(** Verification of |\eqref{eq:hatghat12}\verb| for w=2,3,4,5 **)|
\begin{verbatim}
psi = 13/10; chi = 1/2;
K2 = 111614/10^5; K3 = 103/100; K4 = 101/100; K5 = 1;
g[w_, t_] := t^w/(1 - t^w) (1 - t) (psi - (t/(1 - t))^chi)
Resolve[Exists[t, g[2, t] >  K2 * g[2, 1/2] && 0 < t <= 1/2]]
Resolve[Exists[t, g[3, t] >  K3 * g[3, 1/2] && 0 < t <= 1/2]]
Resolve[Exists[t, g[4, t] >  K4 * g[4, 1/2] && 0 < t <= 1/2]]
Resolve[Exists[t, g[5, t] >  K5 * g[5, 1/2] && 0 < t <= 1/2]]
\end{verbatim}

\subsubsection{Mathematica Code for Proof of Lemma~\ref{lem:xidecreasing}}\label{sec:xidecreasing}
\begin{verbatim}
alpha=1-10^(-4); c=7/10;
\end{verbatim}

\noindent \verb|(** Verification of |\eqref{eq:Mwac23}\verb| **)|
\begin{verbatim}
For[w = 2, w <= 5, w++,
  Print[105/
      100 (2*alpha*c)^(-1) (1 +
       w^(-1))*(1 - 2^(-w))/(1 - 2^(-w - 1)) < 1]];

(2/(1 + c) (2 alpha*c)^(-1) (1 + w^(-1)) /. w -> 6) < 1
\end{verbatim}

\noindent \verb|(** Verification of |\eqref{eq:cstuffcrude}\verb| **)|
\begin{verbatim}
1+(c^(1-7/200)/(1-c^(1-7/200)))*(1-c^(1/200))<105/100
\end{verbatim}

\subsubsection{Mathematica Code for Proof of Lemma~\ref{lem:zeta}}\label{sec:zeta}
\begin{verbatim}
alpha=1-10^(-4); c=7/10;
\end{verbatim}

\noindent \verb|(** Verification of |\eqref{eq:xi2t}\verb| **)|
\begin{verbatim}
alpha^(-1)*(2c)^(-2)*2*(1-2^(-2))^(-1)/(1-c^(1-3/200))<45931/10000
\end{verbatim}

\noindent \verb|(** Verification of |\eqref{eq:xi6t}\verb| **)|
\begin{verbatim}
alpha^(-2)*(2c)^(-6)*6*(1-2^(-6))^(-1)/(1-c^(1-7/200))<278045/100000
\end{verbatim}

\subsubsection{Mathematica Code for Proof of Lemma~\ref{lem:weak}}\label{sec:kapwidewide}
\noindent \verb|(** Verification that the r.h.s. in |\eqref{eq:kapwidewide}\verb| is less than 1 **)|
\begin{verbatim}
psi = 13/10; K2 = 111614/100000; c=7/10; tau2= 45932/10000;

c^(264/100)/(psi-1)*(15/100)*K2*tau2<1
\end{verbatim}

\subsubsection{Mathematica Code for Proof of \texorpdfstring{\eqref{eq:gh1yup}}{(51)}}\label{sec:gh1yup}
\begin{verbatim}
alpha = 1 - 10^(-4); c = 7/10;
psi = 13/10; chi = 1/2;
l[2] = 1; l[3] = 1; l[4] = 1; l[5] = 1;

g[y_, w_] :=  c^(-w) alpha^(-l[w])*w*(1 - y^2)/y^2*(1 - (1 - y^2)^(1/w))*
               (psi - ((1 - y^2)^(1/w)/(1 - (1 - y^2)^(1/w)))^chi);
ylow[w_] := (1 - 2^(-w))^(1/2); yup[w_] := 1;

For[ww = 2, ww <= 5, ww++,
  Print[Resolve[
     Exists[y, g[y, ww] > 2279/10^4 && ylow[ww] <= y < yup[ww]]]];
  ];

For[ww = 2, ww <= 5, ww++,
  Print[Resolve[
     Exists[y,
      g[y, ww] > -15/10 y + 16276/10^4 && ylow[ww] <= y < yup[ww]]]];
  ];

For[ww = 2, ww <= 5, ww++,
  Print[Resolve[
     Exists[y,
      g[y, ww] > -8 y + 8052/10^3 && ylow[ww] <= y < yup[ww]]]];
  ];
\end{verbatim}

\subsubsection{Mathematica Code for Proof of Lemma~\ref{lem:sigma-6-bound}}\label{sec:sigma-6-bound}
\noindent \verb|(** The exact values of Y0 and Y1 in |\eqref{eqn:Y0Y1}\verb| **)|
\begin{verbatim}
SOL=Flatten[Solve[{-15/10 YY0 + 16276/10^4 == 2279/10^4,
           -15/10 YY1 + 16276/10^4 == -8 YY1 + 8052/10^3}, {YY0, YY1}]];
Y0=YY0/.SOL; Y1=YY1/.SOL;

c=7/10; c5=c^(1-6/200); tau6=27805/10^4; psi=13/10;
h[y_] := Min[-15/10 y + 16276/10^4, -8 y + 8052/10^3]
\end{verbatim}

\noindent \verb|(** Verification of |\eqref{eq:sigma06ya}\verb| **)|
\begin{verbatim}
15/100*c*tau6/(psi-1)<1
\end{verbatim}

\noindent \verb|(** Verification of |\eqref{eq:sigma06yb}\verb| **)|
\begin{verbatim}
Resolve[Exists[y1,
  Y0 <= y1 <= 1 && c/(psi - y1)*(h[y1] + 15/100*c5*tau6) > 1 ] ]
\end{verbatim}

\noindent \verb|(** Verification of |\eqref{eq:sigmat6a}\verb| for 2<=t<=7, 3<=i<=7 **)|
\begin{verbatim}
Print[ Table[ Y0^(t-1)-(15/10)*c*c5^(i-1)>0, {t,2,7}, {i,3,t} ] ]
\end{verbatim}

\noindent \verb|(** Verification of |\eqref{eq:sigmat6b}\verb| for 2<=t<=7 **)|
\begin{verbatim}
Print[ Table[ Y0*Y1^(t-2)-(15/10)*c*c5>0, {t,2,7} ] ]
\end{verbatim}

\noindent \verb|(** Verification of |\eqref{eq:sigmat6c}\verb| for t=2,..,7 and 2<=i<=Min[5,t] **)|
\begin{verbatim}
Print[ Table[
c/(psi-1)*( h[Y0]+ c5*h[Y1]*(1-c5^(t-1))/(1-c5)+15/100*c5^t*tau6 )
      -8c*c5^(i-1)>=0, {t,2,7}, {i,2,Min[5,t]}
] ]
\end{verbatim}

\noindent \verb|(** Proof that |\eqref{eq:falsifyt6a}\verb| is false for t=2,3,4,5 **)|
\begin{verbatim}
For[tt = 2, tt <= 5, tt++,
Print[Resolve[Exists[y1,
  Y0 <= y1 <= 1 &&
  c/(psi - y1*Y1^(tt-1) ) *
   (h[y1] + c5*h[Y1]*(1-c5^(tt-1))/(1-c5)  +15/100*c5^tt*tau6) > 1
   ] ] ];
]
\end{verbatim}

\noindent \verb|(** Proof that |\eqref{eq:falsifyt6b}\verb| is false **)|
\begin{verbatim}
Resolve[Exists[y1,
  Y0 <= y1 <= 1 &&
  c/(psi - y1*Y1^4 ) *
   (h[y1] + c5*h[Y1]*(1-c5^5)/(1-c5)  +15/100*c5^6*tau6) > 1
   ] ]
\end{verbatim}

\noindent \verb|(** Proof that |\eqref{eq:falsifyt6c}\verb| is false **)|
\begin{verbatim}
Resolve[Exists[y1,
  Y0 <= y1 <= 1 &&
  c/(psi - y1*Y1^4 ) *
   (h[y1] + c5*h[Y1]*(1-c5^6)/(1-c5)  +15/100*c5^7*tau6) > 1
   ] ]
\end{verbatim}

\subsubsection{Mathematica Code for Proof of Lemma~\ref{lem:sigma-2-bound}}\label{sec:sigma-2-bound}
\noindent \verb|(** The exact values of Y0 and Y1 in |\eqref{eqn:Y0Y1}\verb| **)|
\begin{verbatim}
SOL=Flatten[Solve[{-15/10 YY0 + 16276/10^4 == 2279/10^4,
           -15/10 YY1 + 16276/10^4 == -8 YY1 + 8052/10^3}, {YY0, YY1}]];
Y0=YY0/.SOL; Y1=YY1/.SOL;

c=7/10; c5=c^(1-6/200); tau2=45932/10^4; K2=111614/10^5; psi=13/10;
h[y_] := Min[-15/10 y + 16276/10^4, -8 y + 8052/10^3]
\end{verbatim}

\noindent \verb|(** Verification of |\eqref{eq:sigmat6a}\verb| for t=8, 3<=i<=8 **)|
\begin{verbatim}
Print[ Table[ Y0^(t-1)-(15/10)*c*c5^(i-1)>0, {t,8,8}, {i,3,t} ] ]
\end{verbatim}

\noindent \verb|(** Verification of |\eqref{eq:sigmat6b}\verb| for t=8, i=2 **)|
\begin{verbatim}
Print[ Table[ Y0*Y1^(t-2)-(15/10)*c*c5>0, {t,8,8} ] ]
\end{verbatim}

\noindent \verb|(** Verification of |\eqref{eq:rferferfe}\verb| for 2<=i<=5 **)|
\begin{verbatim}
Print[ Table[
c/(psi-1)*( h[Y0]+ c5*h[Y1]*(1-c5^(7))/(1-c5)+15/100*c5^8*K2*tau2 )
-8c*c5^(i-1)<0, {i,2,5}
] ]
\end{verbatim}

\noindent \verb|(** Verification of |\eqref{eq:tvbgfdb3ed3}\verb| **)|
\begin{verbatim}
Resolve[Exists[y1,
  Y0 <= y1 <= 1 &&
   c Y1^4/(psi - y1*Y1^4)*
          (h[y1] + c5*h[Y1]*(1 - c5^7)/(1 - c5) + 15/100*c5^8*K2*tau2)
       - 8 c*c5^5 >= 0]]
\end{verbatim}

\noindent \verb|(** Proof that |\eqref{eq:falsifyt2a}\verb| is false **)|
\begin{verbatim}
Resolve[Exists[y1,
Y0 <= y1 <= 1 &&
c/(psi - y1*Y1^5 ) *
(h[y1] + c5*h[Y1]*(1-c5^7)/(1-c5) +15/100*c5^8*K2*tau2) > 1
] ]
\end{verbatim}

\subsection{Analysis for \texorpdfstring{$k\geq 3$, $\Delta=6$}{k>=3, Delta=6}}
\subsubsection{Mathematica Code for Proof of Lemma~\ref{lem:potentialfunction}}\label{sec:potentialfunction}
\noindent \verb|(** Verification of |\eqref{eq:numerical123}\verb| **)|
\begin{verbatim}
alpha = 1 - 10^(-4); psi = 13/10; Delta = 6; d0 = Delta - 1;
eps[0]=0; eps[1]=6/10; eps[2]=7/10; eps[3]=83/100; eps[4]=91/100; eps[5]=alpha;
M = 25/1000;

For[BB = 0, BB<=d0, BB++,
  Print["Running for BB=", BB];
  Print[Resolve[1/alpha (eps[BB] + (d0 - BB) M/(psi-1)) > 1]];
  ];
\end{verbatim}

\subsubsection{Mathematica Code for Lemma~\ref{lem:thnmi}}\label{sec:fghc}

\vskip 0.2cm \noindent \verb|(** Verification of |\eqref{eq:v3max}\verb| **)|
\begin{verbatim}
M1 = 1/410;
h= (1-t) (psi- (t/(1-t))^chi)/.{chi -> 1/2, psi -> 13/10};
f = t^6/(1 - t^6) h;
Resolve[Exists[t, f > M1 && 0 <= t <= 1/2]]
\end{verbatim}

\vskip 0.2cm \noindent \verb|(** Verification of |\eqref{eq:somederivative}\verb| **)|
\begin{verbatim}
f = (t^(w - 6) (1 - t^6))/(1 - t^w);
p = 6 t^w - t^6 w + w - 6;
Simplify[D[f, t] - (t^(w - 7) p)/(1 - t^w)^2]
\end{verbatim}

\vskip 0.2cm \noindent \verb|(** Verification of the first inequality in |\eqref{eq:Kwlargeb3}\verb| **)|
\begin{verbatim}
M = 25/1000; M1 = 1/410; Delta = 6; delta = 9789/10000;
Kw = (1/delta)^((w - 1) (Delta - 1));
f = (63 M1/(2^w - 1))* Kw * w * (w + 1)^(11*10^(-5)/Log[6]);
(f /. {w -> 6}) <= M
\end{verbatim}

\vskip 0.2cm \noindent \verb|(** Verification of  |\eqref{eq:Kwlargeb303030}\verb| **)|
\begin{verbatim}
Delta = 6; delta = 9789/10000;
(1/2)*(1/delta)^(Delta-1)(1+6^(-1))(1+7^(-1))^(11*10^(-5)/Log[6])<1
\end{verbatim}

\subsubsection{Mathematica Code for Proof of Lemma~\ref{lem:concav2}}\label{sec:concav2}
\noindent \verb|(** Verification of |\eqref{eq:fppz}\verb| **)|
\begin{verbatim}
h= (1-t) (psi- (t/(1-t))^chi)/.{chi -> 1/2, psi -> 13/10};

For[w = 1, w <= 5, w++,
  Print["Running for w=", w];
  g[w] = (1 - y)/y * (h /. {t -> (1 - y)^(1/w)});
  f[w] = g[w] /. {y -> Exp[z]};
  f2[w] = (D[f[w], {z, 2}])/. {Exp[z] -> T, Exp[-z] -> 1/T, Exp[2*z] -> T^2};
  Print[Resolve[Exists[T, f2[w] > 0 && 1 - (1/2)^w <= T <= 1]]];
];
\end{verbatim}

\subsubsection{Mathematica Code for Proof of Lemma~\ref{lem:assym2}}\label{sec:rfvbbbb}
\noindent \verb|(** Verification of values |$C^{(b_3)}_{\delta}$ \verb| **)|
\begin{verbatim}
CC[1]=1; CC[2]=102/100; CC[3]=103/100; CC[4]=104/100; CC[5]=105/100;
Delta=6; delta=9789/10000; p=27/2;
Cproof=(1/delta)^(b3-1) * ( (1-delta^(b3*p))/(b3 (1-delta^p)) )^(1/p);
For[bb3=1, bb3<=Delta-1, bb3++,
   Print["Checking for b3=",bb3];
   Print[Resolve[ (Cproof/.{b3->bb3}) > CC[bb3] ]];
]
\end{verbatim}

\subsubsection{Mathematica code for Proof of Lemma~\ref{lem:bootphase1}}\label{sec:bootphase1}
\begin{verbatim}
SUBS1 = {(v1/(1 - v1))^chi -> (2 z1)/(1 - z1^2),
         (1 - v1)^(b2 chi) -> ((1 - z1^2)/(1 + z1^2))^b2,
         v1 -> (4 z1^2)/(1 + z1^2)^2
        };
SUBS2 = {(v2/(1 - v2))^chi -> (2z2)/(1 - 2z2^2),
         (1-v2^2)^(b3 chi) -> ((1 - 4z2^4)/(1+4z2^4))^b3,
         v2 -> (4z2^2)/(1+4z2^4)
        };
SUBS = Join[SUBS1, SUBS2];
kappa = delta^b2 *
        (b2 v1(psi-(v1/(1-v1))^chi)
           +2 b3 K2 CC[b3]((v2^2)/(1+v2))(psi-(v2/(1-v2))^chi))/
        (psi - (1-v1)^(b2 chi)(1-v2^2)^(b3 chi))/. SUBS

delta0=9789/10000; psi0=13/10; KK2=1069/1000; d0=5; alpha=1-10^(-4);
CC[0]=0; CC[1]=1; CC[2]=102/100; CC[3]=103/100; CC[4]=104/100; CC[5]=105/100;
tau[0, 0] = 0; tau[0, 1] = 42/100; tau[1, 0] = 42/100;
tau[0, 2] = 54/100; tau[1, 1] = 59/100; tau[2, 0] = 63/100;
tau[0, 3] = 72/100; tau[1, 2] = 74/100; tau[2, 1] = 76/100; tau[3, 0] = 79/100;
tau[0, 4] = 864/1000; tau[1, 3] = 868/1000; tau[2, 2] = 876/1000;
tau[3, 1] = 886/1000; tau[4, 0] = 901/1000;
For[bb2 = 0, bb2 <= d0, bb2++,
 tau[bb2, d0 - bb2] = alpha;
 ]

kappa = kappa/.{d->d0, delta->delta0, psi->psi0, K2->KK2, M->M0};
z1up = 2^(1/2) - 1; z2up = 1/2(3^(1/2)-1);
For[bb2 = 0, bb2 <= d0, bb2++,
  For[bb3 = 0, bb3 <= d0 - bb2, bb3++,
    Print["b2=", bb2, ", b3=", bb3];
    EXP = kappa /. {b2 -> bb2, b3 -> bb3};
    Print[
     Resolve[Exists[{z1, z2},
       EXP > tau[bb2,bb3] && 0 <= z1 <= z1up && 0 <= z2 <= z2up]]];
    ];
  ];
\end{verbatim}

\subsubsection{Mathematica Code for Proof of Lemma~\ref{lem:bootphase2}}\label{sec:bootphase2}
\noindent \verb|(** Verification of |\eqref{eq:bootphase2a}\verb| **)|
\begin{verbatim}
kappa = (tau[b2, b3] (psi-A) +
       3* delta^b2 * K3 * b4 * (v3^3)/(1+v3+v3^2)*(psi-(v3/(1-v3))^chi))/
        (psi - A(1-v3^3)^(b4*chi))

Delta=6; dd0=Delta-1; delta0=9789/10000; psi0=13/10; chi0=1/2; alpha=1-10^(-4);
tau[0, 0]=0; tau[0, 1]=42/100; tau[1, 0]=42/100;
tau[0, 2]=54/100; tau[1, 1]=59/100; tau[2, 0]=63/100;
tau[0, 3]=72/100; tau[1, 2]=74/100; tau[2, 1]=76/100; tau[3, 0]=79/100;
tau[0, 4]=864/1000; tau[1, 3]=868/1000; tau[2, 2]=876/1000;
tau[3, 1]=886/1000; tau[4, 0]=901/1000;
tau[5, 0]=alpha;

kappa = kappa/.{delta->delta0, psi->psi0, chi->chi0, K3->1160/1000}
For[bb2 = 0, bb2 <= dd0, bb2++,
  For[bb3 = 1, bb3 <= dd0 - bb2, bb3++,
    For[bb4 = 1, bb4 <= dd0 - bb2 - bb3, bb4++,
      Print["Running for b2=", bb2, " b3=", bb3, " and b4=", bb4];
      EXP = kappa /. {b2 -> bb2, b3 -> bb3, b4 -> bb4};
      B = bb2 + bb3 + bb4;
      Print[
       Resolve[Exists[{A, v3},
         EXP > tau[B, 0] && 0 <= v3 <= 1/2  && 0 <= A <= 1]]
        ];
      ];
    ];
  ];
\end{verbatim}

\vskip 0.2cm
\noindent \verb|(** Verification of |\eqref{eq:bootphase2b}\verb| **)|
\begin{verbatim}
SUBS = {(v1/(1 - v1))^chi->u1, v1->(u1^2)/(1+u1^2), (v3/(1-v3))^chi->u3,
   v3->(u3^2)/(1+u3^2)};
kappa = delta^b2 *(  b2 v1 (psi - (v1/(1 - v1))^chi)
        +3*K3*b4*(v3^3)/(1+v3+v3^2)*(psi - (v3/(1 - v3))^chi))/
        (psi - (1-v1)^(b2*chi) (1-v3^3)^(b4*chi)) /. SUBS

Delta=6; dd0=Delta-1; delta0=9789/10000; psi0=13/10; chi0=1/2; alpha= 1-10^(-4);
tau[1,0]=42/100; tau[2,0]=63/100; tau[3,0]=79/100; tau[4,0]=901/1000; tau[5,0]=alpha;

kappa = kappa /.{delta->delta0, psi->psi0, chi->chi0, K3->1160/1000}
For[bb2 = 0, bb2 <= dd0, bb2++,
  For[bb4 = 1, bb4 <= dd0 - bb2, bb4++,
    Print["Running for b2=", bb2, " and b4=", bb4];
    EXP = kappa/.{b2 -> bb2, b3 -> bb3, b4 -> bb4};
    B = bb2 + bb4;
    Print[
     Resolve[Exists[{u1, u3},
       EXP > tau[B, 0] && 0 <= u1 <= 1 && 0 <= u3 <= 1]]];
    ];
  ];
\end{verbatim}

\subsubsection{Mathematica Code for Proof of Lemma~\ref{lem:bootphase3}}\label{sec:bootphase3}
The following code can be executed by copying it in a Mathematica cell. \vskip 0.2cm

\noindent \verb|(** Verification of |\eqref{eq:bootphase3}\verb| **)|
\begin{verbatim}
kappa = (tau[Bp, 0] (psi-A) +
       4 * K4 * b5 * (v4^4)/(1+v4+v4^2+v4^3)*(psi-(v4/(1-v4))^chi))/
        (psi - A(1-v4^4)^(b5*chi))

Delta=6; dd0=Delta-1; delta0=9789/10000; psi0=13/10; chi0=1/2; alpha=1-10^(-4);
tau[0,0]=0; tau[1,0]=42/100; tau[2,0]=63/100; tau[3,0]=79/100; tau[4,0]=901/1000;
tau[5,0]=alpha;

kappa = kappa/.{delta->delta0, psi->psi0, chi->chi0, K4->1225/1000}
For[BBp = 0, BBp <= dd0, BBp++,
  For[bb5 = 1, bb5 <= dd0 - BBp, bb5++,
      Print["Running for Bp=", BBp, " and b5=", bb5];
      EXP = kappa /. {Bp -> BBp, b5 -> bb5};
      B = BBp + bb5;
      Print[
       Resolve[Exists[{A, v4},
         EXP > tau[B, 0] && 0 <= v4 <= 1/2  && 0 <= A <= 1]]
        ];
      ];
    ];
\end{verbatim}

\subsubsection{Mathematica Code for Proof of Lemma~\ref{lem:bootphase4}}\label{sec:bootphase4}
\noindent \verb|(** Verification of |\eqref{eq:bootphase4}\verb| **)|
\begin{verbatim}
kappa = (tau[Bp, 0] (psi-A) +
       5 * K5 * b6 * (v5^5)/(1+v5+v5^2+v5^3+v5^4)*(psi-(v5/(1-v5))^chi))/
        (psi - A(1-v5^5)^(b6*chi))

Delta=6; dd0=Delta-1; delta0=9789/10000; psi0=13/10; chi0=1/2; alpha=1-10^(-4);
tau[0,0]=0; tau[1,0]=42/100; tau[2,0]=63/100; tau[3,0]=79/100; tau[4,0]=901/1000;
tau[5,0]=alpha;
Resolve[(1/delta0)^(4*dd0)>1532/1000]

kappa = kappa/.{delta->delta0, psi->psi0, chi->chi0, K5->1532/1000}
For[BBp = 0, BBp <= dd0, BBp++,
  For[bb6 = 1, bb6 <= dd0 - BBp, bb6++,
      Print["Running for Bp=", BBp, " and b6=", bb6];
      EXP = kappa /. {Bp -> BBp, b6 -> bb6};
      B = BBp + bb6;
      Print[
       Resolve[Exists[{A, v5},
         EXP > tau[B, 0] && 0 <= v5 <= 1/2  && 0 <= A <= 1]]
        ];
      ];
    ];
\end{verbatim}

\subsubsection{Mathematica Code for Proof of Lemma~\ref{lem:concav1}}\label{sec:concav1}
\noindent \verb|(** Verification of |\eqref{eq:fppy2}\verb| **)|
\begin{verbatim}
f = (1 - t) (psi - (t/(1 - t))^chi) /. {psi->13/10, chi->1/2, t -> Exp[y]}
f2 = FullSimplify[D[f, {y, 2}]]
Resolve[Exists[y, f2 > 0 && y <= Log[1/2]]]
\end{verbatim}

\subsubsection{Mathematica Code for Proof of Lemma~\ref{lem:hgeneralproof}}\label{sec:hgeneralproof}
\noindent \verb|(** Verification of |\eqref{eq:cvbnm15}\verb| and |\eqref{eq:toverifysecond}\verb| **)|
\begin{verbatim}
psi=13/10; delta=9789/10000; K2=1069/1000; K4=1225/1000;
AA[1]={A1->1, A2->1/delta^5, A->2K2};
AA[2]={A1->2, A2->2K2/delta^5, A->4*1120/1000};
AA[3]={A1->1, A2->1/delta^15, A->5/2};
AA[4]={A1->2K2/delta^5, A2->5/2, A->4K4};

LHS = (4 x1 x2)/((1 + x1^2) (1 + x2^2) - 4 x1 x2);
RHS = psi - (A1*(1 - x1^2)^2/(1 + x1^2)^2 (psi - (2 x1)/(1 - x1^2)) +
      A2*(1 - x2^2)^2/(1 + x2^2)^2 (psi - (2 x2)/(1 - x2^2)))/
      (A (1 - (4 x1 x2)/((1 + x1^2) (1 + x2^2))))
EXPR = RHS^2 - LHS

xup = Sqrt[2] - 1;
For[i=1, i<= 4, i++,
    Print["Running for A1=", A1/.AA[i], " A2=", A2/.AA[i], " A=", A/.AA[i]];
    RHS0=RHS/.AA[i]; EXPR0 = EXPR /. AA[i];
    Print[Resolve[Exists[{x1, x2},
            RHS0 < 0 && 0 <= x1 <= xup && 0 <= x2 <= xup]]
    ];
    Print[Resolve[Exists[{x1, x2},
            EXPR0 < 0 && 0 <= x1 <= xup && 0 <= x2 <= xup]]
    ];
];
\end{verbatim}

\subsubsection{Mathematica Code for Proof of Lemma~\ref{lem:g2q}}\label{sec:rfv}
\noindent \verb|(** Verification of |\eqref{eq:xswxsw34}\verb| **)|
\begin{verbatim}
q = 27/25; Delta = 6;
h = (1-t) (psi - (t/(1 - t))^chi) /. {psi->13/10, chi->1/2, t->(1-y)^(1/2)};
g2 = ((1-y)/y) * h;
FUN = Simplify[(q - 1) y*(D[g2, y])^2 + g2 * ((D[g2, y]) + y*(D[g2, {y, 2}]))];
Resolve[Exists[y, FUN > 0 && 3/4 <= y <= 1 - 1/(2^(Delta-1) + 1)^2]]
\end{verbatim}

\subsection{Remaining Proofs}
\subsubsection{Proof of Theorem~\ref{thm:domsethardness}}\label{sec:domsethardness}
\noindent \verb|(** Verification that |\eqref{eq:wed456456}\verb| has multiple solutions **)|
\begin{verbatim}
Delta = 17;
lambda = 1/2; beta = 1/2; gamma = 1;
Resolve[Exists[{x, y},
  x == lambda ((beta y + 1)/(y + gamma))^(Delta - 1) &&
   y == lambda ((beta x + 1)/(x + gamma))^(Delta - 1) && x > 0 &&
   y > 0 && x > y]]
\end{verbatim}
\end{document}